\documentclass[%
reprint,
amsmath,amssymb,
aps,
]{revtex4-1}

\usepackage{graphicx}% Include figure files
\usepackage{dcolumn}% Align table columns on decimal point
\usepackage{bm}% bold math
\usepackage{subfigure}

\usepackage{stmaryrd} % for \llbrackets
\usepackage{amsthm}  % for proof
\usepackage{IEEEtrantools} %because I like it
\usepackage{dsfont} % for unity matrix

\usepackage{mathtools} % For \triangleq
\usepackage{amssymb} % for \triangleq

\newtheorem{lem}{Lemma}
\newtheorem{theorem}{Theorem}

\newcommand{\Id}[1]{\boldsymbol{\mathbb{I #1}}}

\begin{document}

\preprint{APS/123-QED}

\title{Windowed multipole representation \\ of \\ $R$-matrix cross sections}

\author{Pablo Ducru}
\email{p\_ducru@mit.edu ; pablo.ducru@polytechnique.org}
\altaffiliation[]{Also from \'Ecole Polytechnique, France. \& Schwarzman Scholars, Tsinghua University, China.}
\affiliation{%
Massachusetts Institute of Technology\\
Department of Nuclear Science \& Engineering\\
77 Massachusetts Avenue, Cambridge, MA, 02139 U.S.A.\\
}%
\author{Vladimir Sobes}%
\email{sobesv@utk.edu}
\affiliation{%
University of Tennessee\\
Department of Nuclear Engineering\\
1412 Circle Drive, Knoxville, TN, 37996, U.S.A.
}%

\author{Abdulla Alhajri}
\email{alhajri@mit.edu}
\author{Isaac Meyer}
\email{icmeyer@mit.edu}
\author{Benoit Forget}
\email{bforget@mit.edu}
\affiliation{%
Massachusetts Institute of Technology\\
Department of Nuclear Science \& Engineering\\
77 Massachusetts Avenue, Cambridge, MA, 02139 U.S.A.\\
}%

\author{Colin Josey}
\email{cjosey@lanl.gov}
\affiliation{%
Los Alamos National Laboratory\\
P.O. Box 1663 MS A143, Los Alamos, NM 87545, U.S.A.
}%
%PABLO->COLIN: Please tell me what contact information + affiliation I should add. It has been a while! I hope you are all good at Los Alamos. 

\author{Jingang Liang}
\email{jingang@mit.edu}
\affiliation{%
Tsinghua University\\
Institute of Nuclear and New Energy Technology, Beijing 100084, China
}
%PABLO->JINGANG: Please tell me what contact information + affiliation I should add. 好久不见！我希望你和你的家人都非常好。清华大学回开门了吗？
%Jingang: This a great work. Thanks, Pablo. p.s.: Tsinghua re-opened in August.


% \affiliation{%
%  Authors' institution and/or address\\
%  This line break forced with \textbackslash\textbackslash
% }%


\date{\today}% It is always \today, today,

\begin{abstract}

Nuclear cross sections are basic inputs to any nuclear computation. 
Campaigns of experiments are fitted with the parametric R-matrix model of quantum nuclear interactions, and the resulting cross sections are documented -- both point-wise and as resonance parameters (with uncertainties) -- in standard evaluated nuclear data libraries (ENDF, JEFF, BROND, JENDL, CENDL, TENDL): these constitute our common knowledge of fundamental nuclear physics. 
In the past decade, a collaborative effort has been deployed to establish a new nuclear cross section library format --- the \textit{Windowed Multipole Library} --- with the goal of considerably reducing the cost of cross section calculations in nuclear transport simulations. % \cite{Hwang_1987, Hwang_1992, Hwang_1998, Hwang_2003, Forget_2013, Josey_JCP_2016, Ducru_PHYSOR_conversion_2016, Multipole_regulatized_VF_2018}.

This article lays the theoretical foundations underpinning these efforts. 
From general R-matrix scattering theory, we derive the \textit{windowed multipole representation} of nuclear cross sections.
Though physically and mathematically equivalent, the windowed multipole representation is particularly well suited for subsequent temperature treatment of angle-integrated cross sections: we show that accurate Doppler broadening can be performed analytically up to the first reaction threshold; and we derive cross sections temperature derivatives to any order. 
Furthermore, we here establish a way of converting the R-matrix resonance parameters uncertainty (covariance matrices) into windowed multipole parameters uncertainty. 
% propagating (first-order) the epistemic uncertainty of R-matrix resonance parameters to the new windowed multipole parameters.
We show that generating stochastic nuclear cross sections by sampling from the resulting windowed multipole covariance matrix can reproduce the cross section uncertainty in the original nuclear data file. 

Through this foundational article, we hope to make the Windowed Multipole Representation accessible, reproducible, and usable for the nuclear physics community, as well as provide the theoretical basis for future research on expanding its capabilities.   
%*BEN Not sure this should be in the abstract, but could be in the conclusion.
%PABLO->BEN: me neither, but in the end it is the most important. Lets see what the reviewers say.

\end{abstract}

\maketitle

%\tableofcontents

\clearpage

%%%%%%%%%%%%%%%%%%%%%%%%%%%%%%%%%%%%%%%%%%%%%%%%%%%%%%%%%%%%%%%%%%%%%%%%%%%%%%%%
%*******************************************************************************
%%%%%%%%%%%%%%%%%%%%%%%%%%%%%%%%%%%%%%%%%%%%%%%%%%%%%%%%%%%%%%%%%%%%%%%%%%%%%%%%
\section{\label{sec:Introduction}Introduction}
%%%%%%%%%%%%%%%%%%%%%%%%%%%%%%%%%%%%%%%%%%%%%%%%%%%%%%%%%%%%%%%%%%%%%%%%%%%%%%%%
%*******************************************************************************
%%%%%%%%%%%%%%%%%%%%%%%%%%%%%%%%%%%%%%%%%%%%%%%%%%%%%%%%%%%%%%%%%%%%%%%%%%%%%%%%

Our knowledge of nuclear reactions is progressively built-up by undertaking experiments and analyzing their outcomes through the prism of a quantum model of nuclear collisions called R-matrix theory \cite{Kapur_and_Peierls_1938, Wigner_and_Eisenbud_1947, Bloch_1957, Lane_and_Thomas_1958}. This is known as the nuclear data evaluation process.
Evaluators conduct campaigns to measure nuclear cross sections and fit them with R-matrix parameters.
To account for the epistemic uncertainty introduced, evaluators generate nuclear resonance parameters covariance matrices to reproduce the variance observed in the measurements. 

Other parametrizations of nuclear cross sections exist -- such as the Humblet-Rosenfeld pole expansions in wavenumber space \cite{Theory_of_Nuclear_Reactions_I_resonances_Humblet_and_Rosenfeld_1961, Theory_of_Nuclear_Reactions_II_optical_model_Rosenfeld_1961, Theory_of_Nuclear_Reactions_III_Channel_radii_Humblet_1961_channel_Radii, Theory_of_Nuclear_Reactions_IV_Coulomb_Humblet_1964, Theory_of_Nuclear_Reactions_V_low_energy_penetrations_Jeukenne_1965, Theory_of_Nuclear_Reactions_VI_unitarity_Humblet_1964, Theory_of_Nuclear_Reactions_VII_Photons_Mahaux_1965, Theory_of_Nuclear_Reactions_VIII_evolutions_Rosenfeld_1965,Theory_of_Nuclear_Reactions_IX_few_levels_approx_Mahaux_1965} -- but none have proven as practical to document or use as R-matrix theory, which is why our standard evaluated nuclear data libraries (ENDF\cite{ENDFBVIII8th2018brown}, JEFF\cite{JEFF_2020_plompenJointEvaluatedFission2020}, BROND\cite{BROND_2016}, JENDL\cite{JENDL_shibataJENDL4NewLibrary2011}, CENDL\cite{CENDLProjectChinese2017}, TENDL\cite{TENDLkoningModernNuclearData2012, koningTENDLCompleteNuclear2019}) are constituted of R-matrix parameters (and their covariance uncertainties). \\
At the end of the 20$^{\text{th}}$ century, R. Hwang from Argonne National Laboratory found a way to calculate from R-matrix parameters the Humblet-Rosenfeld pole expansion of neutron cross sections without thresholds, where the wavenumber is proportional to the square root of energy $k_c(E) \propto \sqrt{E}$. He also showed that this \textit{pole representation} in $z \triangleq \sqrt{E}$ space presents a major advantage for subsequent temperature treatment: integral Doppler broadening can be accurately computed with analytic expressions \cite{Hwang_1987, Hwang_1991, Hwang_1992, Hwang_1998, Hwang_2003}. This formalism was further developed into the \textit{windowed multipole representation} in order to perform efficient on-the-fly computations of no-threshold neutron cross sections with a lesser computational memory footprint \cite{Forget_2013, joseyWindowedMultipoleSensitivity2015, Josey_JCP_2016, Ducru_PHYSOR_conversion_2016, Multipole_regulatized_VF_2018}. 
In this article, we extend the windowed multipole representation to all cross sections in the context of R-matrix theory: be they Coulomb, photon, neutrons, with or without thresholds. We also provide means of converting resonance parameters uncertainties into windowed multipole uncertainties. We thus lay the foundations to constitute a full \textit{Windowed Multipole Library}, encompassing all present nuclear data \cite{MitcrpgWMPLibrary2019}. 

In section \ref{sec:From R-matrix to Windowed Multipole}, we derive the windowed multipole representation from general R-matrix theory, showing it is the meromorphic continuation of cross sections to complex energies, and discuss numerical ways of computing the mutipoles, either from resonance parameters or point-wise cross section data.
In section \ref{sec::Windowed Multipole Covariance}, we expand the windowed multipole representation to account for the epistemic uncertainty of the nuclear cross sections \cite{Ducru_MnC_2019_Embedded_Monte_Carlo, Alhajri_MnC_2019_Sensitivity_Monte_Carlo}. We establish the analytic sensitivities of the windowed multipole parameters to the Wigner-Eisenbud R-matrix resonance parameters. This enables us to convert to first-order the standard resonance parameters covariance matrix into a windowed multipoles covariance matrix, and show the latter reproduces the statistical properties of nuclear cross section uncertainties.
%, thereby enabling us to generate stochastic windowed multipole nuclear cross sections.
Having done so, we consider temperature effects in section \ref{sec:DB of WMP}, showing how to analytically Doppler broaden windowed multipole angle-integrated cross sections, and how to compute arbitrary-order temperature derivatives that can prove useful in multiphysics simulations \cite{harperCalculatingReactionRate2016, Sterling_Harper_Cancun_2018_PHYSOR}.
%*BEN add a reference to Sterling's work, I think it was PHYSOR2018?
%PABLO->BEN: cannot find this citation well: do you have a DOI and the page numbers? 

By deriving conversion methods of R-matrix resonance parameters and their uncertainties (covariance matrices) to windowed multipoles, and showing how to account for temperature effects, we thus establish the windowed multipole representation as a general, physically equivalent parametrization of R-matrix cross sections.
By its efficient on-the-fly treatment of uncertainty and Doppler broadening, the windowed multipole representation can achieve considerable computational gains, and has already found several new nuclear reactor physics applications, from the establishment of a new analytic benchmark for neutron slowing down that resolves nuclear resonances overlap \cite{Analytic_Benchmark_1_2020}, or explicit resonance treatment for thermal up-scattering of angular cross sections \cite{Liang_Ducru_ANS_2017}, to differential temperature tallies for higher-order neutronics-thermohydraulics coupling schemes in nuclear transport solvers \cite{harperCalculatingReactionRate2016, Sterling_Harper_Cancun_2018_PHYSOR}, or enabling new uncertainty inference and propagation methods across intractable nuclear systems \cite{Ducru_MnC_2019_Embedded_Monte_Carlo}. 

% PABLO: Add reference to analytic benchmark 2, update them accordingly. + Embedded Monte Carlo / NIP article update

%%%%%%%%%%%%%%%%%%%%%%%%%%%%%%%%%%%%%%%%%%%%%%%%%%%%%%%%%%%%%%%%%%%%%%%%%%%%%%%%
%*******************************************************************************
%%%%%%%%%%%%%%%%%%%%%%%%%%%%%%%%%%%%%%%%%%%%%%%%%%%%%%%%%%%%%%%%%%%%%%%%%%%%%%%%
\section{\label{sec:From R-matrix to Windowed Multipole}From R-matrix to \\ Windowed Multipole}
%%%%%%%%%%%%%%%%%%%%%%%%%%%%%%%%%%%%%%%%%%%%%%%%%%%%%%%%%%%%%%%%%%%%%%%%%%%%%%%%
%*******************************************************************************
%%%%%%%%%%%%%%%%%%%%%%%%%%%%%%%%%%%%%%%%%%%%%%%%%%%%%%%%%%%%%%%%%%%%%%%%%%%%%%%%

We here establish the \textit{Windowed Multipole Representation}, deriving it from general R-matrix scattering theory. 
In doing so, we show that R-matrix cross sections are the sum of two phenomena: thresholds and resonances. 
Thresholds have a behavior in the wavenumber $k_c$ space of the channel $c$, so that in the vicinity of a threshold the cross section admits a Laurent expansion in powers of $k_c$ (starting at $k_c^{-2}$).
Resonances have a behavior in the energy space $E$, and can thus be locally expressed as a sum of Single-Level Breit-Wigner (SLBW) resonances, with both symmetric and anti-symmetric Lorenztian functions. %, with the addition of a smooth analytical background term. 
In \cite{Ducru_Scattering_Matrix_of_Complex_Wavenumbers_2019}, we linked the R-matrix parametrization of the scattering matrix $\boldsymbol{U}(E)$ to its wavenumber $k_c$ expansion, established by Humblet and Rosenfeld in their \textit{Theory of Nuclear Reactions} \cite{Theory_of_Nuclear_Reactions_I_resonances_Humblet_and_Rosenfeld_1961, Theory_of_Nuclear_Reactions_II_optical_model_Rosenfeld_1961, Theory_of_Nuclear_Reactions_III_Channel_radii_Humblet_1961_channel_Radii, Theory_of_Nuclear_Reactions_IV_Coulomb_Humblet_1964, Theory_of_Nuclear_Reactions_V_low_energy_penetrations_Jeukenne_1965, Theory_of_Nuclear_Reactions_VI_unitarity_Humblet_1964, Theory_of_Nuclear_Reactions_VII_Photons_Mahaux_1965, Theory_of_Nuclear_Reactions_VIII_evolutions_Rosenfeld_1965,Theory_of_Nuclear_Reactions_IX_few_levels_approx_Mahaux_1965}.
In this article, we use this connection to establish the Windowed Multipole Representation, which is the meromorphic continuation of R-matrix cross sections in $z \triangleq \sqrt{E}$ space, locally expressing open channels as pole expansions. %wavenumber $k_c$ space.
We build upon our previous work on such expansions \cite{Ducru_shadow_Brune_Poles_2019, Ducru_Scattering_Matrix_of_Complex_Wavenumbers_2019}, using the same consistent notation as reference.

%%%%%%%%%%%%%%%%%%%%%%%%%%%%%%%%%%%%%%%%%%%%%%%%%%%%%%%%%%%%%%%%%%%%%%%%%%%%%%%%
%*******************************************************************************
%%%%%%%%%%%%%%%%%%%%%%%%%%%%%%%%%%%%%%%%%%%%%%%%%%%%%%%%%%%%%%%%%%%%%%%%%%%%%%%%
\subsection{\label{subsec:R-matrix cross section parametrization}R-matrix cross section parametrization}
%%%%%%%%%%%%%%%%%%%%%%%%%%%%%%%%%%%%%%%%%%%%%%%%%%%%%%%%%%%%%%%%%%%%%%%%%%%%%%%%
%*******************************************************************************
%%%%%%%%%%%%%%%%%%%%%%%%%%%%%%%%%%%%%%%%%%%%%%%%%%%%%%%%%%%%%%%%%%%%%%%%%%%%%%%%

R-matrix theory models two-body-in/two-body-out scattering events interacting with a ``black-box'' Hamiltonian \cite{Kapur_and_Peierls_1938, Wigner_and_Eisenbud_1947, Bloch_1957, Lane_and_Thomas_1958}. 
Each pair of possible two-body-inputs/two-body-outputs, along with all the corresponding quantum numbers that describe then, constitutes a channel $c$. 
It is assumed that for each channel, the Hamiltonian can be partitioned into two regions: within an ``inner region'' sphere of channel radius $a_c$, the many bodies interacting through the strong nuclear forces are considered an intractable ``black-box'' Hamiltonian; past the channel radius $a_c$, the ``outer region'' Hamiltonian is well known (say Coulomb potential or free-wave).
For each channel $c$, R-matrix theory studies the many-body scattering event into the reduced one-body system, where the solution of the Schr\"odinger equation is a superposition of an incoming wavefunction $I_c$ and an outgoing wavefunction $O_c$, both function of the wavenumber $k_c$.
The latter can be multiplied by the arbitrary (but fixed) channel radius $a_c$ to yield the \textit{dimensionless wavenumber} 
\begin{equation}
\begin{IEEEeqnarraybox}[][c]{rcl}
\rho_c  & \ \triangleq  \ &  k_c a_c 
\IEEEstrut\end{IEEEeqnarraybox}
\label{eq:rho_c def}
\end{equation}
and we define the corresponding diagonal matrix over all the channels $\boldsymbol{\rho}=\boldsymbol{\mathrm{diag}}\left(\rho_c\right)$.

%%%%%%%%%%%%%%%%%%%%%%%%%%%%%%%%%%%%%%%%%%%%%%%%%%%%%%%%%%%%%%%%%%%%%%%%%%%%%%%%
\subsubsection{\label{subsubsec:Wavenumber-Energy mapping}Wavenumber-Energy mapping}
%%%%%%%%%%%%%%%%%%%%%%%%%%%%%%%%%%%%%%%%%%%%%%%%%%%%%%%%%%%%%%%%%%%%%%%%%%%%%%%%

Each wavenumber is related to the total energy $E$ of the system, which is an eigenvalue of the Hamiltonian in the reduced center-of-mass frame.
In the semi-classical limit, a two massive particles channel (i.e. not photons) of respective masses $m_{c,1}$ and $ m_{c,2}$ will have a wavenumber $k_c$ of:
\begin{equation}
\begin{IEEEeqnarraybox}[][c]{rcl}
k_c(E)  & \ = \ & \sqrt{\frac{2m_{c,1} m_{c,2}}{\left(m_{c,1}+m_{c,2}\right) \mathrm{\hbar}^2} \left(E - E_{T_c}\right)}
\IEEEstrut\end{IEEEeqnarraybox}
\label{eq:rho_c massive}
\end{equation}
where $E_{T_c}$ denotes a threshold energy below which the channel $c$ is closed, as energy conservation cannot be respected ($E_{T_c} = 0 $ for reactions without threshold).
In the same semi-classical limit, a photon particle interacting with a massive body of mass $m_{c,1}$, the center-of-mass wavenumber $k_c$ is linked to the total center-of-mass energy $E$ according to:
\begin{equation}
k_c(E) = \frac{\left( E - E_{T_c} \right)}{2 \mathrm{\hbar}\mathrm{c}}\left[ 1 + \frac{m_{c,1} \mathrm{c}^2}{\left( E - E_{T_c} \right) + m_{c,1}\mathrm{c}^2}\right]
\label{eq:rho_c photon}
\end{equation}
These two semi-classical limits can be encompassed within a single relativistic framework as discussed in equations (4) and (5), section II.A. of \cite{Ducru_shadow_Brune_Poles_2019}. 
Because one must choose the sign of the square root $\pm\sqrt{\cdot}$ in (\ref{eq:rho_c massive}), these $k_c(E)$ relations engender a wavenumber-energy mapping
\begin{equation}
\rho_c(E)  \quad \longleftrightarrow \quad E
\label{eq:rho_c(E) mapping}
\end{equation}
which forms a complex multi-sheeted Riemann surface with branch-points at (or close to) the threshold energies $E_{T_c}$, as discussed in section II.A. p.2 of \cite{Ducru_shadow_Brune_Poles_2019}. 

%%%%%%%%%%%%%%%%%%%%%%%%%%%%%%%%%%%%%%%%%%%%%%%%%%%%%%%%%%%%%%%%%%%%%%%%%%%%%%%%
\subsubsection{\label{subsubsec:Transmission matrix and cross section expressions} Transmission matrix and cross section expressions}
%%%%%%%%%%%%%%%%%%%%%%%%%%%%%%%%%%%%%%%%%%%%%%%%%%%%%%%%%%%%%%%%%%%%%%%%%%%%%%%%

General scattering theory expresses the incoming channel $c$ and outgoing channel $c'$ angle-integrated partial cross section $\sigma_{c,c'}(E)$ at energy $E$ as a function of the probability \textit{transmission matrix} $T_{cc'}(E)$, according to eq.(3.2d) VIII.3. p.293 of \cite{Lane_and_Thomas_1958}:
\begin{equation}
\begin{IEEEeqnarraybox}[][c]{rcl}
\sigma_{cc'}(E) & \ = \ & 4\pi g_{J^\pi_c} \left|\frac{T_{cc'}(E)}{k_c(E)} \right|^2 
\IEEEstrut\end{IEEEeqnarraybox}
\label{eq:partial sigma_cc'}
\end{equation}
where the \textit{spin statistical factor} is defined eq.(3.2c) VIII.3. p.293. of \cite{Lane_and_Thomas_1958} as:
\begin{equation}
    g_{J^\pi_c} \ \triangleq \ \frac{2 J + 1 }{\left(2 I_1 + 1 \right)\left(2 I_2 + 1 \right) }
    \label{eq: spin statistical factor}
\end{equation}
The transmission matrix is itself derived from the \textit{scattering matrix} $\boldsymbol{U}$ of the interaction:
\begin{equation}
\begin{IEEEeqnarraybox}[][c]{rcl}
\boldsymbol{T} & \ \triangleq \ & \frac{ \Id{} - \boldsymbol{\mathrm{e}}^{-\mathrm{i}\boldsymbol{\omega}} \boldsymbol{U} \boldsymbol{\mathrm{e}}^{-\mathrm{i}\boldsymbol{\omega}} }{2}
\IEEEstrut\end{IEEEeqnarraybox}
\label{eq: Transmission matrix}
\end{equation}
where $\boldsymbol{\omega} \triangleq \boldsymbol{\mathrm{diag}}\big( \omega_c \big)$ is the diagonal matrix composed of $\omega_c \triangleq \sigma_{\ell_c}(\eta_c) - \sigma_{0}(\eta_c) $, that is the difference in \textit{Coulomb phase shift}, $\sigma_{\ell_c}(\eta_c) $, which are linked to the phases (argument) of the Gamma function as defined by Ian Thompson in eq.(33.2.10) of \cite{NIST_DLMF} for angular momentum $\ell_c$
\begin{equation}
\begin{IEEEeqnarraybox}[][c]{rcl}
\sigma_{\ell_c}(\eta_c) & \ \triangleq  \ &  \mathrm{arg}\Big( \Gamma\left(1+ \ell_c + \mathrm{i}\eta_c\right) \Big)
\IEEEstrut\end{IEEEeqnarraybox}
\label{eq: Coulomb phase shift def}
\end{equation}
 and \textit{dimensionless Coulomb field parameter}:
\begin{equation}
\begin{IEEEeqnarraybox}[][c]{rcl}
\eta_c & \ \triangleq  \ &   \frac{Z_1 Z_2 e^2 M_\alpha a_c}{\hbar^2 \rho_c}
\IEEEstrut\end{IEEEeqnarraybox}
\label{eq: eta_c def}
\end{equation}
Note that this transmission matrix (\ref{eq: Transmission matrix}) definition $ T_{cc'} \triangleq \frac{\delta_{cc'} - \mathrm{e}^{-\mathrm{i}\omega_c} U_{cc'} \mathrm{e}^{-\mathrm{i}\omega_{c'}} }{2}$ is a scaled rotation of the one defined by Lane and Thomas $T_{cc'}^{\text{L\&T}} \ \triangleq \  \delta_{cc'}\mathrm{e}^{2\mathrm{i}\omega_c} - U_{cc'}$ (c.f. eq. (2.3), VIII.2. p.292 and eq.(3.2d) VIII.3. p.293 of \cite{Lane_and_Thomas_1958}). We introduce definition (\ref{eq: Transmission matrix}) for better physical interpretability, algebraic simplicity and numerical stability.

Unitarity of the scattering matrix entails that $ \sum_{c'}\left|\delta_{cc'} - \mathrm{e}^{-\mathrm{i}\omega_{c} }  U_{cc'} \mathrm{e}^{-\mathrm{i}\omega_{c'} }   \right|^2 =  2\left(1 - \Re\left[ \mathrm{e}^{-2\mathrm{i}\omega_{c} }  U_{cc} \right]\right)$, which in turn leads to the following expression for the total cross section of a given channel:
\begin{equation}
    \begin{IEEEeqnarraybox}[][c]{rcl}
      \sigma_{c}(E)  \ \triangleq \ \sum_{c'}\sigma_{cc'}(E) & \ = \ & 4 \pi g_{J^\pi_c} \frac{ \Re\left[  T_{cc}(E) \right] }{ \left|k_c(E)\right|^2}
    \IEEEstrut\end{IEEEeqnarraybox}
     \label{eq:total σ_c}
\end{equation}
In both cross section expressions (\ref{eq:partial sigma_cc'}) and (\ref{eq:total σ_c}), the $1/\left|k_c\right|^{2}$ term links the cross section to the probability of interaction, and expresses the channel reversibility equivalence:
\begin{equation}
    \begin{IEEEeqnarraybox}[][c]{rcl}
      \frac{k_c^2\sigma_{cc'}}{g_{J^\pi_c} }  \ = \  \frac{k_{c'}^2\sigma_{c'c}}{g_{J^\pi_{c'}} } 
    \IEEEstrut\end{IEEEeqnarraybox}
     \label{eq:cross section reversibility equivalence}
\end{equation}

The incoming $I_c$ and outgoing $O_c$ waves are functions of the dimensionless wavenumber $\rho_c \triangleq a_c k_c$ and are linked to the regular and irregular Coulomb wave functions (or Bessel functions in the case of neutral particle channels), defined in eq.(2.13a)-(2.13b) III.2.b p.269 \cite{Lane_and_Thomas_1958}:
\begin{equation}
\begin{IEEEeqnarraybox}[][c]{rcl}
O_c  & \ = \ &   {H_{+}}_c \mathrm{e}^{- \mathrm{i} \omega_c} = \left(G_c + \mathrm{i} F_c\right)\mathrm{e}^{- \mathrm{i} \omega_c} \\
I_c & \ = \ &  {H_{-}}_c \mathrm{e}^{\mathrm{i} \omega_c}  = \left(G_c - \mathrm{i} F_c\right) \mathrm{e}^{\mathrm{i} \omega_c}
\IEEEstrut\end{IEEEeqnarraybox}
\label{eq:def H_pm I and O}
\end{equation}
and for properties of which we refer to Ian J. Thompson's Chapter 33, eq.(33.2.11) in \cite{NIST_DLMF}, or Abramowitz \& Stegun chapter 14, p.537 \cite{Abramowitz_and_Stegun}.
In polar notation:
\begin{equation}
\begin{IEEEeqnarraybox}[][c]{rcl}
{H_{+}}_c & \ = \ &  \left| {H_{+}}_c\right|\mathrm{e}^{ \mathrm{i} \phi_c}  \\
\left| {H_{+}}_c\right| & =  &  \left|\sqrt{\left|G_c\right|^2 + \left|F_c\right|^2} \right| \\
\phi_c & \ \triangleq \ &  \mathrm{arg}\left({H_{+}}_c\right) = 2 \; \mathrm{arctan}\left( \frac{\left|F_c\right|}{ \left| {H_{+}}_c\right| + \left|F_c\right|}\right) 
\IEEEstrut\end{IEEEeqnarraybox}
\label{eq:def H_+ polar decomposition}
\end{equation}

%*BEN What are H, G and F?
%PABLO->BEN: they are: "the regular and irregular Coulomb wave functions (or Bessel functions in the case of neutral particle channels), defined in eq.(2.13a)-(2.13b) III.2.b p.269 \cite{Lane_and_Thomas_1958} [...] and for properties of which we refer to Ian J. Thompson's Chapter 33, eq.(33.2.11) in \cite{NIST_DLMF}, or Abramowitz \& Stegun chapter 14, p.537 \cite{Abramowitz_and_Stegun}." Tell me if not clear. 
% I am actually going to propose Ian Thompson as a reviewer: he would be great on this part. 

%%%%%%%%%%%%%%%%%%%%%%%%%%%%%%%%%%%%%%%%%%%%%%%%%%%%%%%%%%%%%%%%%%%%%%%%%%%%%%%%
\subsubsection{\label{subsubsec:R-matrix scattering matrix parametrization} R-matrix scattering matrix parametrization}
%%%%%%%%%%%%%%%%%%%%%%%%%%%%%%%%%%%%%%%%%%%%%%%%%%%%%%%%%%%%%%%%%%%%%%%%%%%%%%%%

R-matrix theory parametrizes the energy dependence of the scattering matrix $\boldsymbol{U}(E)$ as:
\begin{equation}
\begin{IEEEeqnarraybox}[][c]{rcl}
\boldsymbol{U} & \ = \ & \boldsymbol{O}^{-1}\boldsymbol{I} + 2\mathrm{i} \boldsymbol{\rho}^{1/2} \boldsymbol{O}^{-1} \boldsymbol{R}_{L} \boldsymbol{O}^{-1} \boldsymbol{\rho}^{1/2}
\IEEEstrut\end{IEEEeqnarraybox}
\label{eq:U expression}
\end{equation}
where the incoming and outgoing wavefunctions, $\boldsymbol{I}=\boldsymbol{\mathrm{diag}}\left(I_c\right)$ and $\boldsymbol{O}=\boldsymbol{\mathrm{diag}}\left(O_c\right)$, are subject to the following Wronksian condition for all channel $ c$, $ w_c \triangleq O_c^{(1)}I_c - I_c^{(1)}O_c = 2\mathrm{i} $, and where $\boldsymbol{R}_{L}$ is the \textit{Kapur-Peierls operator}, defined as (see equation (20) section II.D of \cite{Ducru_shadow_Brune_Poles_2019}):

\begin{equation}
\begin{IEEEeqnarraybox}[][c]{rcl}
\boldsymbol{R}_{L}  & \ \triangleq  \ &  \left[ \Id{} -  \boldsymbol{R}\boldsymbol{L^0} \right]^{-1} \boldsymbol{R}  =   \boldsymbol{\gamma}^\mathsf{T} \boldsymbol{A} \boldsymbol{\gamma}
\label{eq: def Kapur-Peierls operator}
\IEEEstrut\end{IEEEeqnarraybox}
\end{equation}
where $\boldsymbol{R}$ is the Wigner-Eisenbud \textit{R-matrix} \cite{Wigner_and_Eisenbud_1947}:
\begin{equation}
\begin{IEEEeqnarraybox}[][c]{rcl}
R_{cc'}(E) & \; \triangleq \; &  \sum_{\lambda=1}^{N_\lambda}\frac{\gamma_{\lambda,c}\gamma_{\lambda,c'}}{E_\lambda - E}
\IEEEstrut\end{IEEEeqnarraybox}
\label{eq:R expression}
\end{equation}
parametrized by the real \textit{resonance energies} $E_\lambda \in \mathbb{R}$ and the real \textit{resonance widths} $\gamma_{\lambda,c} \in \mathbb{R}$ -- of which we respectively build the diagonal matrix $\boldsymbol{e} = \boldsymbol{\mathrm{diag}}\left( E_\lambda \right)$ of size the number of levels (resonances) $N_\lambda$, and the rectangular matrix $\boldsymbol{\gamma} = \boldsymbol{\mathrm{mat}}\left( \gamma_{\lambda,c} \right)$ of size $N_\lambda \times N_c$ where $N_c$ is the number of channels. 
The Kapur-Peierls operator (\ref{eq: def Kapur-Peierls operator})
is thus a function of $\boldsymbol{R}$ and $ \boldsymbol{L^0} \triangleq \boldsymbol{L} - \boldsymbol{B}$, where $\boldsymbol{B} = \boldsymbol{\mathrm{diag}}\left( B_c \right)$ is the diagonal matrix of real arbitrary boundary conditions $B_c$, and $\boldsymbol{L} = \boldsymbol{\mathrm{diag}}\left( L_c \right)$ where $L_c(\rho_c)$ is the dimensionless reduced logarithmic derivative of the outgoing-wave function at the channel surface:
\begin{equation}
    L_c(\rho_c)   \triangleq  \frac{\rho_c}{O_c} \frac{\partial O_c}{\partial \rho_c}
    \label{eq: L operator} 
\end{equation}
An equivalent definition (\ref{eq: def Kapur-Peierls operator}) of the Kapur-Peierls operator $\boldsymbol{R}_{L}$ can be expressed with the \textit{level matrix} $\boldsymbol{A}$  (see equations (17) and (18) of section II.C of \cite{Ducru_shadow_Brune_Poles_2019}):
\begin{equation}
\begin{IEEEeqnarraybox}[][c]{rcl}
\boldsymbol{A^{-1}} & \ \triangleq \ & \boldsymbol{e} - E\Id{} - \boldsymbol{\gamma}\left( \boldsymbol{L} - \boldsymbol{B} \right)\boldsymbol{\gamma}^\mathsf{T}
\IEEEstrut\end{IEEEeqnarraybox}
\label{eq:inv_A expression}
\end{equation}

As such, provided with the threshold energies, the channel radius, boundary conditions, and Wigner-Eisenbud resonance energies and widths, which we can collectively call the set of R-matrix parameters $\Big\{ E_{T_c}, a_c, B_c, E_{\lambda}, \gamma_{\lambda,c} \Big\}$, one can entirely determine the energy behavior of the scattering matrix $\boldsymbol{U}$ through (\ref{eq:U expression}), and therefore the cross sections through (\ref{eq:partial sigma_cc'}) and (\ref{eq:total σ_c}).

%%%%%%%%%%%%%%%%%%%%%%%%%%%%%%%%%%%%%%%%%%%%%%%%%%%%%%%%%%%%%%%%%%%%%%%%%%%%%%%%
\subsubsection{\label{subsubsec:Reich-Moore and Breit-Wigner approximations to R-matrix theory} Reich-Moore and Breit-Wigner approximations to R-matrix theory}
%%%%%%%%%%%%%%%%%%%%%%%%%%%%%%%%%%%%%%%%%%%%%%%%%%%%%%%%%%%%%%%%%%%%%%%%%%%%%%%%

In practice, many evaluations in standard nuclear data libraries are carried out with approximations of R-matrix theory.
The most important and common is the \textit{Reich-Moore} approximation. It reduces the R-matrix to only the channels of interest, % (for instance alpha particles, neutron scattering, fission, or some explicit photon channels), 
%*BEN I thought Reich Moore only included fission and scattering explicitly, and when they included a specific photon channel they called it R-matrix limited?  Maybe just eliminate the parenthesis since you kinda of explain what is done right after.
%PABLO->BEN: removed. You are right it is clearer that way. "R-matrix Limited" is a format in ENDF library, but in theory it is just Reich-Moore with some photon channel explicitly accounted for. Nothing in the Reich-Moore approximation specifies what type of channels can be included explicitly or not. However, I do not want to confuse people even more: it is already a difficult article. 
and accounts for the effect of all the other channels not explicitly treated by means of the Teichmann and Wigner channel elimination method (c.f. \cite{Teichmann_and_Wigner_1952} or section X, p.299 of \cite{Lane_and_Thomas_1958}). 
Usually, photon channels ($\gamma$ ``gamma capture'') are eliminated, so that in practice the Reich-Moore approximation of R-matrix theory \cite{Reich_Moore_1958} consists of adding a partial eliminated capture width $\Gamma_{\lambda,\gamma}$  to every resonance energy $E_\lambda$, shifting the latter into the complex plane (c.f. section IV.A of \cite{Ducru_shadow_Brune_Poles_2019}):
\begin{equation}
\begin{IEEEeqnarraybox}[][c]{rcl}
\boldsymbol{e}_{\mathrm{R.M.}} & \ \triangleq \ & \boldsymbol{\mathrm{diag}}\left( E_\lambda - \mathrm{i}\frac{\Gamma_{\lambda,\gamma}}{2}\right)
\IEEEstrut\end{IEEEeqnarraybox}
\label{eq:e diagonal matrix Reich-Moore}
\end{equation}
The R-matrix (\ref{eq:R expression}) without the eliminated photon channels becomes:
\begin{equation}
\begin{IEEEeqnarraybox}[][c]{rcl}
R_{c,c'\not\in\gamma_{\mathrm{elim.}} } & \; \triangleq \; &  \sum_{\lambda=1}^{N_\lambda}\frac{\gamma_{\lambda,c}\gamma_{\lambda,c'}}{E_\lambda - \mathrm{i}\frac{\Gamma_{\lambda,\gamma}}{2} - E}  \\
 \mathrm{i.e.}  \quad \boldsymbol{R}_{\mathrm{R.M.}} & \; = \; & \boldsymbol{\gamma}^\mathsf{T} \left(\boldsymbol{e}_{\mathrm{R.M.}} - E\Id{}\right)^{-1}\boldsymbol{\gamma}
\IEEEstrut\end{IEEEeqnarraybox}
\label{eq:R expression Reich Moore}
\end{equation}
and the Reich-Moore inverse level matrix (\ref{eq:inv_A expression}) becomes:
\begin{equation}
\begin{IEEEeqnarraybox}[][c]{rcl}
\boldsymbol{A^{-1}}_{\mathrm{R.M.}} & \ \triangleq \ & \boldsymbol{e}_{\mathrm{R.M.}} - E\, \Id{} - \boldsymbol{\gamma}\left( \boldsymbol{L} - \boldsymbol{B} \right)\boldsymbol{\gamma}^\mathsf{T}
\IEEEstrut\end{IEEEeqnarraybox}
\label{eq:inv_A expression Reich-Moore}
\end{equation}
All the other R-matrix expressions linking these operators to the scattering matrix (\ref{eq:U expression}), and therefore the cross sections, remain unchanged.
Practically, the only consequence of the Reich-Moore formalism is to introduce complex resonance energies (\ref{eq:e diagonal matrix Reich-Moore}). In this sense, one can consider the Reich-Moore formalism as a generalization of R-matrix theory, even though it finds its source in the elimination of intractable channels. It can thus also be seen as a compression algorithm. 
Indeed, it is possible to convert Reich-Moore parameters into standard R-matrix ones (not complex resonance energies) by means of the \textit{Generalized Reich-Moore} formalism, as established in \cite{Generalized_Reich_Moore_2017}. Yet this comes at the cost of introducing many more parameters, thereby considerably increasing memory requirements. This is because Generalized Reich-Moore converts the eliminated channels R-matrix ($N_c \times N_c$ with $c \not \in \gamma_{\mathrm{elim.}}$) into a square R-matrix of the size of the levels ($N_\lambda \times N_\lambda$), and we often have $N_\lambda \gg N_c$, specially for large nuclides (c.f. \cite{Generalized_Reich_Moore_2017}).

Also, note that some older evaluations are made in the \textit{Multi-Level Breit-Wigner} approximation, which simply consists of assuming the level matrix (\ref{eq:inv_A expression}) is diagonal. 
This can be expressed using the Hadamard product `` $\circ$ '' with the identity matrix as:
\begin{equation}
\begin{IEEEeqnarraybox}[][c]{rcl}
\boldsymbol{A^{-1}}_{\mathrm{MLBW}} & \ \triangleq \ & \boldsymbol{A^{-1}} \circ \Id{}
\IEEEstrut\end{IEEEeqnarraybox}
\label{eq:inv_A expression MLBW}
\end{equation}

Apart from these modified expressions of the level matrix, neither the Reich-Moore nor the Multi-Level Breit-Wigner approximations have any further incidence on how to convert R-matrix cross sections to Windowed Multipole Representation: it suffices to take the corresponding level matrix and proceed as follows.

%%%%%%%%%%%%%%%%%%%%%%%%%%%%%%%%%%%%%%%%%%%%%%%%%%%%%%%%%%%%%%%%%%%%%%%%%%%%%%%%
\subsubsection{\label{subsubsec:Parametrizing R-matrix cross sections} Parametrizing R-matrix cross sections}
%%%%%%%%%%%%%%%%%%%%%%%%%%%%%%%%%%%%%%%%%%%%%%%%%%%%%%%%%%%%%%%%%%%%%%%%%%%%%%%%

By substituting the R-matrix parametrization (\ref{eq:U expression}) of the scattering matrix $\boldsymbol{U}$ into the transmission matrix $\boldsymbol{T}$ definition (\ref{eq: Transmission matrix}), and noticing that wavefunction relations (\ref{eq:def H_pm I and O}) entail $\frac{\boldsymbol{H_+} - \boldsymbol{H_-}}{2\mathrm{i}} = \boldsymbol{F}$, one finds the transmission matrix can be decomposed into the rotated (by a factor of imaginary $\mathrm{i}$) difference between a diagonal \textit{potential matrix} $\boldsymbol{D}$ and a full \textit{resonance matrix} $\boldsymbol{Z}$:
\begin{equation}
\begin{IEEEeqnarraybox}[][c]{rcl}
\boldsymbol{T} & \ = \ & \mathrm{i}\left(\boldsymbol{D} - \boldsymbol{Z} \right) \\
\boldsymbol{Z} & \ \triangleq \ &  \boldsymbol{H_+}^{-1} \boldsymbol{\rho}^{1/2} \boldsymbol{R}_{L} \boldsymbol{\rho}^{1/2} \boldsymbol{H_+}^{-1} \\
\boldsymbol{D} & \ \triangleq \ &  \boldsymbol{H_+}^{-1}\boldsymbol{F} = \frac{\Id{} - \boldsymbol{Y} }{2\mathrm{i}} \\
\boldsymbol{Y} & \ \triangleq \ &  \boldsymbol{H_+}^{-1}\boldsymbol{H_-}
\IEEEstrut\end{IEEEeqnarraybox}
\label{eq:Transmission matrix decomposition Z and D}
\end{equation}
From cross section expression (\ref{eq:partial sigma_cc'}), the transmission probabilities from channel $c$ to channel $c'$ are then the square-modulus $\left| T_{cc'}\right|^2$. Decomposition (\ref{eq:Transmission matrix decomposition Z and D}) expresses this as: 
\begin{equation}
    \left| T_{cc'}\right|^2 = \left| Z_{cc'}\right|^2 + \left| D_{c}\right|^2 \delta_{cc'} - 2 \Re\left[ Z_{cc'} D_{c}
^*\right] \delta_{cc'} 
\end{equation}
where $\left[ \; \cdot\; \right]^*$ designates the complex conjugate. 
For the total cross section (\ref{eq:total σ_c}), it is the real part of the transmission matrix that appears: $\Re\left[T_{cc}\right] = \Re\left[\mathrm{i}D_{c}\right] - \Re\left[\mathrm{i}Z_{cc}\right]$.
Note that $\boldsymbol{D}$ definition (\ref{eq:Transmission matrix decomposition Z and D}) entails $2 \boldsymbol{D}^* =  \mathrm{i}\left( \Id{} - \boldsymbol{Y}^{*} \right)$ and $ \left|\boldsymbol{D}\right|^2  \ = \  \Re\left[\mathrm{i}\boldsymbol{D} \right]$, since definition (\ref{eq:def H_+ polar decomposition}) yields
\begin{equation}
\begin{IEEEeqnarraybox}[][c]{rcl}
\Re\left[ \mathrm{i} D_c \right] & \ = \ & \frac{\left|F_c\right|^2}{\left|G_c\right|^2 + \left|F_c\right|^2} = \left|D_c\right|^2 = \sin^2\left(\phi_c\right)
\IEEEstrut\end{IEEEeqnarraybox}
\label{eq: |D|^2 to Re[iD]}
\end{equation}
We can thus decompose the cross sections into the following components, all expressed as the real part of some matrix elements calculable from R-matrix theory:
\begin{itemize}
    \item Potential cross section (of channel $c$):
\begin{equation}
\begin{IEEEeqnarraybox}[][l]{rcl}
\sigma_c^{\text{pot}}(E) & \ \triangleq \ & 4\pi g_{J^\pi_c} \left|\frac{D_{c}}{k_c} \right|^2  =  4\pi g_{J^\pi_c} \frac{\Re\left[\mathrm{i}D_{c}\right]}{\left|k_c\right|^2}  
\IEEEstrut\end{IEEEeqnarraybox}
\label{eq: potential cross section}
\end{equation}
    \item Total cross section (of channel $c$):
\begin{equation}
\begin{IEEEeqnarraybox}[][l]{rcl}
\sigma_c(E) & \ \triangleq \ & \sigma_c^{\text{pot}}(E) + 4\pi g_{J^\pi_c} \frac{\Re\left[-\mathrm{i}Z_{cc}\right]}{\left|k_c\right|^2}
\IEEEstrut\end{IEEEeqnarraybox}
\label{eq: total cross section}
\end{equation}
    \item Self-scattering cross section (of channel $c$):
\begin{equation}
\begin{IEEEeqnarraybox}[][l]{rcl}
\sigma_c^{\text{scat}}(E) & \ \triangleq \ &  4\pi g_{J^\pi_c} \frac{\Re\left[-2 Z_{cc}D_c^*\right]}{\left|k_c\right|^2}
\IEEEstrut\end{IEEEeqnarraybox}
\label{eq: scatering cross section}
\end{equation}
    \item Interference cross section (of channel $c$):
\begin{equation}
\begin{IEEEeqnarraybox}[][l]{rcl}
\sigma_c^{\text{int}}(E) & \ \triangleq \ &  4\pi g_{J^\pi_c} \frac{\Re\left[-\mathrm{i}Z_{cc}Y_c^*\right]}{\left|k_c\right|^2}
\IEEEstrut\end{IEEEeqnarraybox}
\label{eq: interference cross section}
\end{equation}
    \item Reaction cross section (from channel $c$ to $c'$):
\begin{equation}
\begin{IEEEeqnarraybox}[][l]{rcl}
\sigma_{cc'}^{\text{react}}(E) & \ \triangleq \ & 4\pi g_{J^\pi_c} \left|\frac{Z_{cc'}}{k_c} \right|^2 
\IEEEstrut\end{IEEEeqnarraybox}
\label{eq: reaction cross section}
\end{equation}
    \item Partial (angle-integrated) cross section (from channel $c$ to $c'$):
\begin{equation}
\begin{IEEEeqnarraybox}[][l]{rcl}
\sigma_{cc'}(E) & \ \triangleq \ & \Big( \sigma_c^{\text{pot}}(E) + \sigma_c^{\text{scat}}(E) \Big)\delta_{cc'}  + \sigma_{cc'}^{\text{react}}(E)  \\
& \ = \ & \Big( \sigma_c^{\text{tot}}(E) - \sigma_c^{\text{int}}(E) \Big)\delta_{cc'}  + \sigma_{cc'}^{\text{react}}(E)
\IEEEstrut\end{IEEEeqnarraybox}
\label{eq: partial cross section breakdown}
\end{equation}
\end{itemize}

Writing these expressions as functions of the dimensionless wavenumbers of each channel, $\rho_c \triangleq k_c a_c$, cross sections appear as proportional to the area of the channel radius disc $ \sigma_c(E) \propto 4\pi a_c^2$, and the modulation of this area is linked to both the transmission matrix amplitudes $\left| T_{cc'}(E) \right|^2$ -- which exhibit the resonance behavior -- and the $1/k_c^2$ wavenumber effect that dominates the total cross section close to the zero-energy threshold.

%%%%%%%%%%%%%%%%%%%%%%%%%%%%%%%%%%%%%%%%%%%%%%%%%%%%%%%%%%%%%%%%%%%%%%%%%%%%%%%%
%*******************************************************************************
%%%%%%%%%%%%%%%%%%%%%%%%%%%%%%%%%%%%%%%%%%%%%%%%%%%%%%%%%%%%%%%%%%%%%%%%%%%%%%%%
\subsection{\label{subsec:Kapur-Peierls operator pole expansion in Siegert-Humblet radioactive states}Kapur-Peierls operator pole expansion in Siegert-Humblet radioactive states}
%%%%%%%%%%%%%%%%%%%%%%%%%%%%%%%%%%%%%%%%%%%%%%%%%%%%%%%%%%%%%%%%%%%%%%%%%%%%%%%%
%*******************************************************************************
%%%%%%%%%%%%%%%%%%%%%%%%%%%%%%%%%%%%%%%%%%%%%%%%%%%%%%%%%%%%%%%%%%%%%%%%%%%%%%%%

The first step towards the Windowed Multipole Representation consists of performing the pole expansion of the Kapur-Peierls operator $\boldsymbol{R}_{L}$ into what are called the Siegert-Humblet radioactive states \cite{Siegert_1939,Breit_radioactive_1940, Radioactive_Mahaux_1969, Eigenchannel_Treatment_Of_R_Matrix_Theory_1997,Radioactive_states_Roumania_2012}. We here summarize this process for the usual case of non-degenerate solutions, and we refer to sections II and IV of \cite{Ducru_Scattering_Matrix_of_Complex_Wavenumbers_2019} for a detailed study. 

%*BEN is radioactive a good term for this?
%PABLO->BEN: I added all the literature to back the name "radioactive state" -- which was approved in the other Phys.Rev.C article. There is indeed a naming difficulty, because these are the poles of the Kapur-Peirels operator, as which I would like to call this the "Kapur-Peierls problem". Nonetheless, the parameters (poles and residues) are called in Lane and Thomas the "Siegert-Humblet poles and residues or parameters", because due to some confusion, the "Kapur-Peierls" poles known in the literature are those from another equations, and are different and energy-dependent. Yet, Humblet is also responsible for a totally different representation (the Humblet-Rosenfeld one), which can also be misleading. At the same time, an entire branch of the literature calls theses states (the Siegert-Humblet parameters, poles and residues of the Kapur-Peirels operator): "radioactive states", precisely because they decompose the cross section in a way that can be shown to represent decay etc. So all in all, I had trouble with the naming and chose to call it "radioactive problem" because it is not the misleading "Kapur-Peirels parameters", nor the "Siegert-Humblet" one (also somewhat misleading, plus a bit out of the blue), and this decomposition is called decompositions into "radioactive states" in the physics literature... Honestly, not sure what is the right naming convention here, I just tried a short name people will remember easily that is consistent with the literature and evocative. 
The \textit{radioactive states problem} consists of finding the poles $\big\{\mathcal{E}_j\big\}$ and residue widths vectors  $\left\{\boldsymbol{r_j}\right\}$ of the Kapur-Peierls operator $\boldsymbol{R}_{L}$, that is solving the following generalized eigenvalue problem \cite{Siegert_1939,Breit_radioactive_1940, Radioactive_Mahaux_1969, Eigenchannel_Treatment_Of_R_Matrix_Theory_1997,Radioactive_states_Roumania_2012}:
\begin{equation}
\left.\boldsymbol{R}_{L}^{-1}(E)\right|_{E = \mathcal{E}_j} \boldsymbol{r_j} = \boldsymbol{0}
\label{eq:R_L radioactive problem}
\end{equation}
where the residue widths vectors $\left\{\boldsymbol{r_j}\right\}$ are subject to the following normalization:
\begin{equation}
\boldsymbol{r_j}^\mathsf{T} \left( \left. { \frac{\partial \boldsymbol{R}_{L}^{-1}}{\partial E} }\right|_{E=\mathcal{E}_j} \right) \boldsymbol{r_j} = 1 
\label{eq:R_L residues normalization}
\end{equation}
which can be calculated using
\begin{equation}
\left. { \frac{\partial \boldsymbol{R}_{L}^{-1}}{\partial E} }\right|_{E=\mathcal{E}_j} = \frac{\partial \boldsymbol{R}^{-1} }{\partial E} (\mathcal{E}_j) - \frac{\partial \boldsymbol{L} }{\partial E} (\mathcal{E}_j)
\end{equation}
where the R-matrix $\boldsymbol{R}$ is invertible (for $E \neq E_\lambda$) as % at the radioactive poles $\left\{\mathcal{E}_j\right\}$, with
\begin{equation}
\frac{\partial \boldsymbol{R}^{-1} }{\partial E} (E) = - \boldsymbol{R}^{-1} \boldsymbol{\gamma}^\mathsf{T} \left(\boldsymbol{e} - E\Id{}\right)^{-2} \boldsymbol{\gamma} \boldsymbol{R}^{-1}
\end{equation}
These \textit{radioactive poles} $\left\{\mathcal{E}_j\right\}$ and \textit{radioactive widths} $\left\{\boldsymbol{r_{j}} = \left[ r_{{j,c_1}}, \hdots, r_{{j,c}} , \hdots , r_{{j,c_{N_c}}} \right]^\mathsf{T} \right\}$, are the Siegert-Humblet parameters.
They can equivalently be obtained by solving the level matrix $\boldsymbol{A}$ radioactive eigenproblem:
% They can equivalently be obtained through the level matrix $\boldsymbol{A}$, by solving the generalized eigenproblem:
\begin{equation}
\left.\boldsymbol{A}^{-1}(E)\right|_{E=\mathcal{E}_j}\boldsymbol{a_j} = \boldsymbol{0}
\label{eq::invA det roots}
\end{equation}
where the eigenvectors $\boldsymbol{a_j}$ are subject to normalization:
\begin{equation}
{\boldsymbol{a_j}}^\mathsf{T} \left( \left. { \frac{\partial \boldsymbol{A}^{-1}}{\partial E} }\right|_{E=\mathcal{E}_j} \right) \boldsymbol{a_j} = 1 
\label{eq:A residues normalization}
\end{equation}
which is readily calculable from
\begin{equation}
\frac{\partial \boldsymbol{A}^{-1}}{\partial E} (\mathcal{E}_j) = - \Id{} - \boldsymbol{\gamma} \frac{\partial \boldsymbol{L} }{\partial E} (\mathcal{E}_j) \boldsymbol{\gamma}^\mathsf{T}
\label{eq: level matrix energy derivatives}
\end{equation}
The level-matrix residues widths vectors are then linked to the radioactive widths by the following relation:
\begin{equation}
\boldsymbol{r_j} = \boldsymbol{\gamma}^\mathsf{T}\boldsymbol{a_j}
\label{eq::radioactive widths rj from aj}
\end{equation}
The radioactive energy poles are complex and usually decomposed as:
\begin{equation}
\mathcal{E}_j \triangleq E_j - \mathrm{i}\frac{\Gamma_j}{2}
\label{eq:E_j pole def}
\end{equation}
It can be shown (c.f. discussion section IX.2.d pp.297--298 in \cite{Lane_and_Thomas_1958}, or section 9.2 eq. (9.11) in \cite{Theory_of_Nuclear_Reactions_I_resonances_Humblet_and_Rosenfeld_1961}) that fundamental physical properties (conservation of probability, causality and time reversal) ensure that the poles reside either on the positive semi-axis of purely-imaginary $k_c \in \mathrm{i}\mathbb{R}_+$ -- corresponding to bound states for real sub-threshold energies, i.e. $E_j < E_{T_c} $ and $\Gamma_j = 0$ -- or that all the other poles are on the lower-half $k_c$ plane, with $\Gamma_j > 0$, corresponding to ``resonance'' or ``radioactively decaying'' states. All poles enjoy the specular symmetry property: if $k_c \in \mathbb{C}$ is a pole of the Kapur-Peierls operator, then $-k_c^*$ is too.
Additional discussion on these radioactive poles and residues can be found in \cite{Lane_and_Thomas_1958}, sections IX.2.c-d-e p.297-298, or in \cite{Siegert_1939, Breit_radioactive_1940, Radioactive_Mahaux_1969, Eigenchannel_Treatment_Of_R_Matrix_Theory_1997, Radioactive_states_Roumania_2012}.

For our purpose of constructing the Windowed Multipole Representation for R-matrix cross sections, the key property of the radioactive states is that they allow, by virtue of the Mittag-Leffler theorem \cite{Mittag-Leffler_1884, Pacific_Journal_Mittag_Leffler_and_spectral_theory_1960}, to locally decompose the Kapur-Peierls operator into a sum of poles and residues and a holomorphic entire part $\boldsymbol{\mathrm{Hol}}_{\boldsymbol{R}_{L}}(E)$, in the neighborhood $\mathcal{W}(E)$ (vicinity) of any complex energy $E\in\mathbb{C}$ away from the branch points (threshold energies $E_{T_c}$) of mapping (\ref{eq:rho_c(E) mapping}):
\begin{equation}
\boldsymbol{R}_{L}(E) \underset{\mathcal{W}(E)}{=} \sum_{j\geq 1} \frac{\boldsymbol{r_j}\boldsymbol{r_j}^\mathsf{T}}{E - \mathcal{E}_j} + \boldsymbol{\mathrm{Hol}}_{\boldsymbol{R}_{L}}(E)
\label{eq::RL Mittag Leffler}
\end{equation}
Theorem 1 of \cite{Ducru_Scattering_Matrix_of_Complex_Wavenumbers_2019} presents the branch structure of the radioactive poles $\mathcal{E}_j$ on the Riemann surface of the energy-wavenumber mapping (\ref{eq:rho_c(E) mapping}). We also show that when solving in dimensionless wavenumber space $\rho_c$, there are $N_L$ number of solutions to the radioactive problem (\ref{eq:R_L radioactive problem}). In the case of massive neutral particles (neutrons and neutrinos) we have
\begin{equation}
N_L = \left(2 N_\lambda  + \sum_{c=1}^{N_c} \ell_c \right)\; \times 2^{(N_{E_{T_c} \neq E_{T_{c'}}} - 1 ) }
\label{eq::NL number of poles}
\end{equation}
where $N_{E_{T_c} \neq E_{T_{c'}}} $ denotes the number of channels with different thresholds.
For charged particles, there is an infinite number (countable) of radioactive poles: $N_L = \infty$. This is in essence because the R-matrix contributes to $N_\lambda$ poles in energy space, and the reduced logarithmic derivative $L_c(\rho_c)$ contributes its poles in $\rho_c$ space. 
Nonetheless, the key is that these poles are of two types: there are $2N_\lambda$ ``principal'' poles, corresponding to the observable resonances, and all the rest which are ``distant poles''. Moreover, the ``principal poles'' themselves come in pairs, of which usually only one significantly contributes to the resonance of the cross section. 

A critical property of the radioactive poles $\mathcal{E}_j$ is that these are exactly all the poles of the scattering matrix $\boldsymbol{U}(E)$ (proof in theorem 3 of \cite{Ducru_Scattering_Matrix_of_Complex_Wavenumbers_2019}). From decomposition (\ref{eq:Transmission matrix decomposition Z and D}), this entails that the transmission matrix readily admits the following Mittag-Leffler expansion:
\begin{equation}
\boldsymbol{T}(E) \underset{\mathcal{W}(E)}{=} -\mathrm{i}\sum_{j\geq 1} \frac{\boldsymbol{\tau_j}\boldsymbol{\tau_j}^\mathsf{T}}{E - \mathcal{E}_j} + \boldsymbol{\mathrm{Hol}}_{\boldsymbol{T}}(E)
\label{eq::T Mittag Leffler}
\end{equation}
where the residue width vectors are obtained by evaluating the functions in (\ref{eq:Transmission matrix decomposition Z and D}) at the pole values:
\begin{equation}
  \boldsymbol{\tau_j} = \boldsymbol{H_+}^{-1}\left(\mathcal{E}_j\right) \boldsymbol{\rho}^{1/2}\left(\mathcal{E}_j\right) \cdot \boldsymbol{r_j}
  \label{eq: T_j as function of r_j}
\end{equation}

%%%%%%%%%%%%%%%%%%%%%%%%%%%%%%%%%%%%%%%%%%%%%%%%%%%%%%%%%%%%%%%%%%%%%%%%%%%%%%%%
%*******************************************************************************
%%%%%%%%%%%%%%%%%%%%%%%%%%%%%%%%%%%%%%%%%%%%%%%%%%%%%%%%%%%%%%%%%%%%%%%%%%%%%%%%
\subsection{\label{subsec:Resonance matrix Z expansion in square root of energy space}Transmission matrix $T$ and resonance matrix $\boldsymbol{Z}$ expansions in square root of energy $z$-space}
%%%%%%%%%%%%%%%%%%%%%%%%%%%%%%%%%%%%%%%%%%%%%%%%%%%%%%%%%%%%%%%%%%%%%%%%%%%%%%%%
%*******************************************************************************
%%%%%%%%%%%%%%%%%%%%%%%%%%%%%%%%%%%%%%%%%%%%%%%%%%%%%%%%%%%%%%%%%%%%%%%%%%%%%%%%

Though energy $E$-space expansion (\ref{eq::T Mittag Leffler}) is correct, we will nonetheless also introduce expansions in the square root of energy $z$-space:
\begin{equation}
    z \triangleq \sqrt{E}
    \label{def: z = sqrt(E)}
\end{equation}
We do this to better express the behavior of massive particles (not massless photons) near the zero-energy threshold, and in order to perform analytic Doppler broadening of massive particles. Indeed, for the zero threshold $E_{T_c} = 0$, the wavenumber of massive particles is simply proportional to the square root of energy: $k \propto z$. Hwang noticed this entails a remarkable property: for neutral particles without threshold, the Kapur-Peierls operator $R_L(z)$ is a rational function of $z$ (c.f. \cite{Hwang_1987}), and therefore the radioactive problem (\ref{eq:R_L radioactive problem}) can be completely solved using polynomial root finders (c.f. section \ref{subsec:Hwang's special case: zero-threshold neutron cross sections}).

The general Mittag-Leffler expansion (\ref{eq::RL Mittag Leffler}) of the Kapur-Peierls operator in $z$-space is 
\begin{equation}
\boldsymbol{R}_{L}(z) \underset{\mathcal{W}(z)}{=} \sum_{j\geq 1} \frac{\boldsymbol{\kappa_j}\boldsymbol{\kappa_j}^\mathsf{T}}{z - p_j} + \boldsymbol{\mathrm{Hol}}_{\boldsymbol{R}_{L}}(z)
\label{eq::RL Mittag Leffler in z-space}
\end{equation}
where square root of energy $z$-space poles are 
\begin{equation}
    p_j \triangleq \sqrt{\mathcal{E}_j}
    \label{def: p_j = sqrt(E_j)}
\end{equation}
and  the residue widths are connected to the poles as:
\begin{equation}
    \boldsymbol{\kappa_j} \triangleq \frac{\boldsymbol{r_j}}{\sqrt{2p_j}}
    \label{eq: link kappa to r residues}
\end{equation}
This is readily obtained from previous $E$-space expressions using partial fraction decomposition of simple poles:
\begin{equation*}
    \frac{\boldsymbol{r_j}\boldsymbol{r_j}^\mathsf{T}}{E - \mathcal{E}_j} = \frac{\frac{\boldsymbol{r_j}\boldsymbol{r_j}^\mathsf{T}}{2\sqrt{\mathcal{E}_j}}}{\sqrt{E} - \sqrt{\mathcal{E}_j}} +  \frac{-\frac{\boldsymbol{r_j}\boldsymbol{r_j}^\mathsf{T}}{2\sqrt{\mathcal{E}_j}}}{\sqrt{E} + \sqrt{\mathcal{E}_j}} 
\end{equation*}
The poles $p_j$ come in opposite pairs ($p_j^+ = + \sqrt{\mathcal{E}_j}$ and $p_j^- = - \sqrt{\mathcal{E}_j}$), and the corresponding residue widths are thus rotated by $\pm\pi/2$ (multiplication by $\pm\mathrm{i}$): $\boldsymbol{\kappa_j}^- \triangleq \boldsymbol{r_j}/\sqrt{2p_j^-} =  - \mathrm{i}\boldsymbol{r_j}/\sqrt{2p_j^+}$. The same $\boldsymbol{r_j}$ is shared by both poles $p_j^+$ and $p_j^-$, so that  $\boldsymbol{\kappa_j}^-{\boldsymbol{\kappa_j}^-}^\mathsf{T} = - \boldsymbol{\kappa_j}^+{\boldsymbol{\kappa_j}^+}^\mathsf{T}$.

Alternatively, Mittag Leffler expansion (\ref{eq::RL Mittag Leffler in z-space}) can also be directly obtained by solving the radioactive problem in square-root-of-energy $z$ space: 
\begin{equation}
\left.\boldsymbol{R}_{L}^{-1}(z)\right|_{z = p_j} \boldsymbol{\kappa_j} = \boldsymbol{0}
\label{eq:R_L radioactive problem in z-space}
\end{equation}
where the residue widths vectors $\left\{\boldsymbol{\kappa_j}\right\}$ are subject to the following normalization:
\begin{equation}
\boldsymbol{\kappa_j}^\mathsf{T} \left( \left. { \frac{\partial \boldsymbol{R}_{L}^{-1}}{\partial z} }\right|_{z=p_j} \right) \boldsymbol{\kappa_j} = 1 
\label{eq:R_L residues normalization z-space}
\end{equation}
which yields relationship (\ref{eq: link kappa to r residues}) (from $z = \sqrt{E}$), and can be calculated directly using
\begin{equation}
\left. { \frac{\partial \boldsymbol{R}_{L}^{-1}}{\partial z} }\right|_{z= p_j} = \frac{\partial \boldsymbol{R}^{-1} }{\partial z} (p_j) - \frac{\partial \boldsymbol{L} }{\partial z} (p_j)
\end{equation}
where $\boldsymbol{R}$ is invertible at $z$-space radioactive poles $\left\{p_j\right\}$ as
\begin{equation}
\frac{\partial \boldsymbol{R}^{-1} }{\partial z} (z) = - 2 z \boldsymbol{R}^{-1} \boldsymbol{\gamma}^\mathsf{T} \left(\boldsymbol{e} - z^2\Id{}\right)^{-2} \boldsymbol{\gamma} \boldsymbol{R}^{-1}
\end{equation}
and where the partial derivatives $\frac{\partial \boldsymbol{L} }{\partial z} (p_j)$ can be derived from the Mittag-Leffler expansion of $\boldsymbol{L}(\rho)$ established in theorem 1 of \cite{Ducru_shadow_Brune_Poles_2019}:
\begin{equation}
\frac{\partial \boldsymbol{L} }{\partial z} = \left[ \mathrm{i} +  \sum_{n \geq 1}  \frac{1}{\rho-\omega_n} + \frac{\rho}{\left(\rho-\omega_n\right)^2} \right] \frac{\partial \rho}{\partial z} 
\label{eq: partial L partial z}
\end{equation}
where $\left\{\omega_n \right\}$ are the roots of the $O_c(\rho)$ outgoing wavefunctions, also roots of ${H_{+}}_c(\rho)$ from (\ref{eq:def H_pm I and O}): $\forall n , \; {H_{+}}_c(\omega_n) = 0$.
For neutral particles, there are a finite number of such roots, reported in table \ref{tab::L_values_neutral}.

Equivalently, we can solve for the level matrix $\boldsymbol{A}$ radioactive problem in $z$-space:
\begin{equation}
\left.\boldsymbol{A}^{-1}(z)\right|_{z=p_j}\boldsymbol{\alpha_j} = \boldsymbol{0}
\label{eq::invA det roots z-space}  
\end{equation}
with eigenvectors $\boldsymbol{\alpha_j} \triangleq \frac{\boldsymbol{a_j}}{\sqrt{2p_j}}$ subject to normalization:
\begin{equation}
{\boldsymbol{\alpha_j}}^\mathsf{T} \left( \left. { \frac{\partial \boldsymbol{A}^{-1}}{\partial z} }\right|_{z=p_j} \right) \boldsymbol{\alpha_j} = 1 
\label{eq:A residues normalization z-space}
\end{equation}
which is readily calculable from
\begin{equation}
\frac{\partial \boldsymbol{A}^{-1}}{\partial z} (p_j) = - 2z \Id{} - \boldsymbol{\gamma} \frac{\partial \boldsymbol{L} }{\partial z} (p_j) \boldsymbol{\gamma}^\mathsf{T}
\label{eq: level matrix energy derivatives z-sapce}
\end{equation}
The level-matrix residues widths vectors are then linked to the radioactive widths by the following relation:
\begin{equation}
\boldsymbol{\kappa_j} = \boldsymbol{\gamma}^\mathsf{T}\boldsymbol{\alpha_j}
\label{eq::radioactive widths rj from aj z-space}
\end{equation}

Regardless of the method deployed to obtain (\ref{eq::RL Mittag Leffler in z-space}), the latter entails the following Mittag Leffler expansion for resonance matrix $\boldsymbol{Z}$ 
\begin{equation}
\boldsymbol{Z}(z) \underset{\mathcal{W}(z)}{=} \sum_{j\geq 1} \frac{\boldsymbol{\zeta_j}\boldsymbol{\zeta_j}^\mathsf{T}}{z - p_j} + \boldsymbol{\mathrm{Hol}}_{\boldsymbol{Z}}(z)
\label{eq::Z Mittag Leffler in z-space}
\end{equation}
Where the residue widths are connected to the poles as:
\begin{equation}
\begin{IEEEeqnarraybox}[][c]{rCl}
    \boldsymbol{\zeta_j} & = & \frac{\boldsymbol{\tau_j}}{\sqrt{2p_j}} \\
    & = & \boldsymbol{H_+}^{-1}\left(p_j\right) \boldsymbol{\rho}^{1/2}\left(p_j\right) \cdot \boldsymbol{\kappa_j} \\
    & =  & \boldsymbol{H_+}^{-1}\left(p_j\right) \boldsymbol{\rho}^{1/2}\left(p_j\right) \cdot \frac{\boldsymbol{r_j}}{\sqrt{2p_j}}
    \label{eq: link zeta to r residues}
\IEEEstrut\end{IEEEeqnarraybox}
\end{equation}
This links back to the transmission matrix Mittag Leffler expansion (\ref{eq::T Mittag Leffler}), which in $z$-space entails: 
\begin{equation}
\boldsymbol{T}(z) \underset{\mathcal{W}(z)}{=} -\mathrm{i}\sum_{j\geq 1} \frac{\boldsymbol{\zeta_j}\boldsymbol{\zeta_j}^\mathsf{T}}{z - p_j} + \boldsymbol{\mathrm{Hol}}_{\boldsymbol{T}}(z)
\label{eq::T Mittag Leffler z-space}
\end{equation}
This transmission matrix Mittag Leffler expansion (\ref{eq::T Mittag Leffler z-space}) corresponds to the Humblet-Rosenfeld scattering matrix expansion in equation (1.54) section I.1.4, p.538, of \cite{Theory_of_Nuclear_Reactions_I_resonances_Humblet_and_Rosenfeld_1961}, where they denote the holomorphic (entire) part $ \boldsymbol{\mathrm{Hol}}_{\boldsymbol{T}}(z)$ as $Q_\ell(k)$. 
As they discuss, the natural variable for this non-resonant part is indeed the wavenumber $k_c$. 
Equations (\ref{eq::T Mittag Leffler z-space}) and (\ref{eq: link zeta to r residues}) thus explicitly link the residues of the Humblet-Rosenfeld expansions to the Wigner-Eisenbud R-matrix parameters. 
 Unfortunately, there exists no simple general method to express the expansion coefficients of this entire part directly from R-matrix parameters.

%%%%%%%%%%%%%%%%%%%%%%%%%%%%%%%%%%%%%%%%%%%%%%%%%%%%%%%%%%%%%%%%%%%%%%%%%%%%%%%%
%*******************************************************************************
%%%%%%%%%%%%%%%%%%%%%%%%%%%%%%%%%%%%%%%%%%%%%%%%%%%%%%%%%%%%%%%%%%%%%%%%%%%%%%%%
\subsection{\label{subsec:Hwang's conjugate continuation}Hwang's conjugate continuation}
%%%%%%%%%%%%%%%%%%%%%%%%%%%%%%%%%%%%%%%%%%%%%%%%%%%%%%%%%%%%%%%%%%%%%%%%%%%%%%%%
%*******************************************************************************
%%%%%%%%%%%%%%%%%%%%%%%%%%%%%%%%%%%%%%%%%%%%%%%%%%%%%%%%%%%%%%%%%%%%%%%%%%%%%%%%

The Windowed Multipole Representation is essentially an analytic continuation of R-matrix cross sections into the complex plane, in $z$-space.
R-matrix cross sections (\ref{eq:partial sigma_cc'}) and (\ref{eq:total σ_c}) are the square moduli and real parts of the transmission matrix $T_{cc'}(E)$ and the wavenumber $k_c(E)$, yielding real cross sections. 
Yet one can analytically continue these cross sections by performing the \textit{conjugate continuation} of all R-matrix operators, which consists of taking the value of the modulus and real parts on the real axis $z \in \mathbb{R}$, and continuing them to the complex plane. This was the key insight introduced by Hwang in \cite{Hwang_1987}.

For any meromorphic function $f(z)$, we define its \textit{continued conjugate} $f^*(z)$ as:
\begin{equation}
f^*(z) \triangleq f(z^*)^* 
\label{eq::continued conjugate}
\end{equation}
As such, the continued conjugate real part is defined as
\begin{equation}
\Re_{\mathrm{conj}}\left[f(z)\right] \triangleq \frac{f(z) + f^*(z)}{2} 
\label{eq::continued conjugate real part}
\end{equation}
and the continued conjugate square modulus as
\begin{equation}
\left|f\right|_{\mathrm{conj}}^2(z) \triangleq f(z)\times f^*(z) 
\label{eq::continued square modulus definition}
\end{equation}
These are meromorphic complex functions: $\Re_{\mathrm{conj}}\left[f(z)\right] \in \mathbb{C}$ and $\left|f\right|_{\mathrm{conj}}^2(z) \in \mathbb{C} $. They are the analytic continuation to complex $z\in\mathbb{C}$ of the real part and the square modulus, which they match on the real axis $z\in\mathbb{R}$.
Consider a meromorphic function $f(z)$ with simple poles and Mittag-Leffler expansion 
\begin{equation}
f(z) \underset{\mathcal{W}(z)}{=} \sum_{j\geq 1} \frac{r_j}{z - p_j} + \sum_{n\geq 0} a_n z^n
\label{eq::f Mittag Leffler in z-space}
\end{equation}
Its continued conjugate square modulus is thus
\begin{equation}
\begin{IEEEeqnarraybox}[][c]{rCl}
\left|f\right|_{\mathrm{conj}}^2(z) &  \underset{\mathcal{W}(z)}{=} & \left(\sum_{j\geq 1} \frac{r_j}{z - p_j} + \sum_{n\geq 0} a_n z^n\right) \\ 
& & \times \left(\sum_{j\geq 1} \frac{r_j^*}{z - p_j^*} + \sum_{n\geq 0} a_n^* z^n\right)
\label{eq::f Mittag Leffler in z-space squared modulus}
\IEEEstrut\end{IEEEeqnarraybox}
\end{equation}
The unicity of poles and residues entails all the poles of $\left|f\right|_{\mathrm{conj}}^2(z)$ are the poles $p_j$ of $f(z)$ and their complex conjugate $p_j
^*$. By evaluating the corresponding residues, one finds the following Mittag-Leffler expansion for the conjugate continuation:
\begin{equation}
\left|f\right|_{\mathrm{conj}}^2(z) \underset{\mathcal{W}(z)}{=} \sum_{j\geq 1} \frac{r_j \cdot f(p_j^*)^*}{z - p_j} + \frac{r_j^* \cdot f(p_j^*)}{z - p_j^*}  + \sum_{n\geq 0} c_n z^n
\label{eq::|f|2 conjugate continuation}
\end{equation}
where
\begin{equation}
   c_n \triangleq \Re\left[ \sum_{k=0}^n a_{n-k}a_k^* + 2\sum_{j\geq 1} \frac{ r_j \cdot \left[ f(p_j^*)^* - a_n^* \right] }{p_j}  \right]
   \label{eq: conjugate continuation entire coefficients}
\end{equation}
which can be obtained by developing (\ref{eq::f Mittag Leffler in z-space squared modulus}) and applying Cauchy's residues theorem to (\ref{eq::|f|2 conjugate continuation}) with contour integrations of $\frac{\left| f\right|^2(z)}{z^{n+1}}$ . 
In (\ref{eq::|f|2 conjugate continuation}), one recognizes the remarkable property that the continued square modulus can be expressed as a continued conjugate real part
\begin{equation}
\left|f\right|_{\mathrm{conj}}^2(z) \underset{\mathcal{W}(z)}{=} \Re_{\mathrm{conj}}\left[ \sum_{j\geq 1} \frac{\widetilde{r}_j}{z - p_j} + \sum_{n\geq 0} c_n z^n \right]
\label{eq::|f|2 conjugate continuation as real part}
\end{equation}
with
\begin{equation}
    \widetilde{r}_j \triangleq 2 \cdot r_j \cdot f(p_j^*)^*
\end{equation}

Therefore, by using Hwang's conjugate continuation, one can express all R-matrix cross sections as the continued conjugate real part of conjugate continued R-matrix operators: this is the key to converting R-matrix cross sections to Windowed Multipole Representation.

\subsection{\label{subsec:Windowed Multipole Representation}Windowed Multipole Representation}
%%%%%%%%%%%%%%%%%%%%%%%%%%%%%%%%%%%%%%%%%%%%%%%%%%%%%%%%%%%%%%%%%%%%%%%%%%%%%%%%
%*******************************************************************************
%%%%%%%%%%%%%%%%%%%%%%%%%%%%%%%%%%%%%%%%%%%%%%%%%%%%%%%%%%%%%%%%%%%%%%%%%%%%%%%%

The Windowed Multipole Representation is the analytic continuation of the pole expansion of R-matrix cross sections.
For open channels (energies above thresholds $E > E_{T_c}$), the energy dependence of R-matrix cross sections -- described by equations (\ref{eq:partial sigma_cc'}) and (\ref{eq:total σ_c}) -- is expanded along the real energy axis $E\in \mathds{R}$, and the corresponding expressions are analytically continued to all complex energies $E\in \mathds{C}$. The Windowed Multipole Representation can thus be seen as a generalization of R-matrix cross sections to the complex plane, for open channels, as shown in figure \ref{fig: Windowed Multipole Representation}. As such, windowed multipole cross sections only match R-matrix cross sections for real energies above the channel threshold: $E > E_{T_c}$.

%%%%%%%%%%%%%%%%%%%%%%%%%%%%%%%%%%%%%%%%%%%%%%%%%%%%%%%%%%%%%%%%%%%%%%%%%%%%%%%%
\subsubsection{\label{subsubsec:Windowed Pole Representation: Transmission matrix approach} Windowed Pole Representation: Transmission matrix approach}
%%%%%%%%%%%%%%%%%%%%%%%%%%%%%%%%%%%%%%%%%%%%%%%%%%%%%%%%%%%%%%%%%%%%%%%%%%%%%%%%

The most straightforward approach is to consider the transmission matrix $\boldsymbol{T}(E)$ Mittag-Leffler expansion (\ref{eq::T Mittag Leffler}), and apply Hwang's conjugate continuation in energy space, which yields:
\begin{equation}
\left|\boldsymbol{T}\right|^2_{\mathrm{conj}}(E) \underset{\mathcal{W}(E)}{=} \Re_{\mathrm{conj}}\left[ \sum_{j\geq 1} \frac{-\mathrm{i} \boldsymbol{\widetilde{\tau}_j} }{E - \mathcal{E}_j} + \boldsymbol{\mathrm{Hol}}_{\left|\boldsymbol{T}\right|^2}(E) \right]
\label{eq:: |T|^2 Mittag Leffler}
\end{equation}
where we use the Hadamard product `` $\circ$ '' to express the residues as:
\begin{equation}
   \boldsymbol{\widetilde{\tau}_j} \triangleq 2\cdot \boldsymbol{\tau_j}\boldsymbol{\tau_j}^\mathsf{T} \circ \boldsymbol{T}\left(\mathcal{E}_j^*\right)^*
   \label{eq: tau tilde}
\end{equation}
Thus, for real energies with open channels, the partial and total cross sections can be expressed respectively as
\begin{equation}
\begin{IEEEeqnarraybox}[][c]{rcl}
\sigma_{cc'}(E) & \ \underset{\mathcal{W}(E)}{=} \ &  \frac{4\pi g_{J^\pi_c}}{\left|k_c(E)\right|^2}   \Re_{\mathrm{conj}}\mkern-6mu\left[ \sum_{j\geq 1} \frac{-\mathrm{i} \left[\boldsymbol{\widetilde{\tau}_j}\right]_{cc'} }{E - \mathcal{E}_j} + \boldsymbol{\mathrm{Hol}}_{\left|\boldsymbol{T}\right|^2}(E) \mkern-2mu \right]
\IEEEstrut\end{IEEEeqnarraybox}
\label{eq:partial sigma_cc' Energy expansion}
\end{equation}
and
\begin{equation}
    \begin{IEEEeqnarraybox}[][c]{rcl}
      \sigma_{c}(E)  & \ \underset{\mathcal{W}(E)}{=} \ & \frac{4 \pi g_{J^\pi_c} }{ \left|k_c(E)\right|^2}  \Re_{\mathrm{conj}}\mkern-6mu\left[ \sum_{j\geq 1} \frac{-\mathrm{i}\left[\boldsymbol{\tau_j}\boldsymbol{\tau_j}^\mathsf{T}\right]_{cc}}{E - \mathcal{E}_j} + \boldsymbol{\mathrm{Hol}}_{\boldsymbol{T}}(E) \mkern-2mu \right]
    \IEEEstrut\end{IEEEeqnarraybox}
     \label{eq:total σ_c Energy expansion}
\end{equation}
Expressions (\ref{eq:partial sigma_cc' Energy expansion}) and (\ref{eq:total σ_c Energy expansion}) are general, they apply to any cross section described by R-matrix theory (be it massless photons or massive charged or neutral particle channels). They are local expressions, only valid on the neighborhood $\mathcal{W}(E)$ of any given energy $E$ away from the thresholds (branch points $E_{T_c}$ of (\ref{eq:rho_c(E) mapping}) mapping), though this neighborhood can be as large as the distance between thresholds (for more discussion on this point, we refer to the penultimate paragraph of section II.D in \cite{Ducru_Scattering_Matrix_of_Complex_Wavenumbers_2019}). They reflect the fact that two physical phenomena dictate the behavior of R-matrix cross sections: resonances and thresholds.

Away from threshold energies $E_{T_c}$ -- the branch-points of wavenumber-energy mapping (\ref{eq:rho_c(E) mapping}) -- each resonance can be accurately represented by a Single-Level Breit-Wigner (SLBW) profile in energy space $E$, that is the combination of symmetric and anti-symmetric Lorenztian functions. These are made evident by recalling definition (\ref{eq:E_j pole def}), which splits the radioactive poles into real and imaginary components $\mathcal{E}_j \triangleq E_j - \mathrm{i}\frac{\Gamma_j}{2} $, and noticing that each resonance of the Windowed Multipole Representation (\ref{eq:partial sigma_cc' Energy expansion}) can be expressed as: 
\begin{equation}
\begin{IEEEeqnarraybox}[][c]{C}
     \Re\left[\frac{a + \mathrm{i}b}{ E - \mathcal{E}_j}\right] = a \frac{\left( E - E_j \right)}{\left( E - E_j \right)^2 + \frac{\Gamma_j^2}{4}} + b \frac{\frac{\Gamma_j}{2}}{\left( E - E_j \right)^2 + \frac{\Gamma_j^2}{4}}
\label{eq:: SLBW Lorenzian profiles for resonances}
\IEEEstrut\end{IEEEeqnarraybox}
\end{equation}

The sum of resonances is complemented by the holomorphic background term $\boldsymbol{\mathrm{Hol}}_{\boldsymbol{T}}$, and modulated by the $\frac{1}{\left|k_c(E)\right|^2}$ term. This illustrates the fact that the wavenumber $k_c$ dominates the behavior of R-matrix cross sections near thresholds $E_{T_c}$, where $k_c \to 0$. Moreover, the holomorphic (entire) part is itself more naturally described as a function of the wavenumber $k_c$ rather then the energy: $\boldsymbol{\mathrm{Hol}}_{\boldsymbol{T}}(\boldsymbol{k})$, as explained by Humblet and Rosenfeld through equations (1.64) and (1.67) section I.1.4, p.539-540, of \cite{Theory_of_Nuclear_Reactions_I_resonances_Humblet_and_Rosenfeld_1961}.
The threshold behavior of R-matrix cross sections was detailed by Wigner in \cite{Wigner_Thresholds_1948}. Depending on the angular momenta $\ell$ and $\ell'$ and the charges of the particles, reaction and scattering cross sections either: a) have threshold behaviors in powers of $k_c^{N(\ell, \ell')} $, where $N(\ell, \ell') \in \mathbb{Z}$ is some integer depending on the different angular momenta, but never smaller than negative two  ($N(\ell, \ell') \geq -2$); or b) in some cases of Coulomb repulsion, modulate this with an exponential decay $\propto \exp(- a / k_c) $ with some real positive $a > 0$ (see section III of \cite{Wigner_Thresholds_1948} for more details). 
This means we can represent in all generality the threshold behavior as a Laurent expansion around the threshold: $\sigma_{cc'} \underset{k_c \to 0}{\sim} \sum_{n\geq - 2} a_n k_c^n$. 

By thus expressing the threshold behavior explicitly, we can constitute the \textit{Windowed Multipole Representation} of R-matrix cross sections:
\begin{equation}
\begin{IEEEeqnarraybox}[][c]{rcl}
\sigma_{cc'}(E) & \ \underset{\mathcal{W}(E)}{\triangleq} \ & \sum_{n\geq -2} \widetilde{a}_n^{cc'} k_c^n(E) +  \frac{1}{E}  \Re_{\mathrm{conj}}\left[ \sum_{j\geq 1} \frac{\widetilde{R}_j^{cc'} }{E - \mathcal{E}_j}  \right]
\IEEEstrut\end{IEEEeqnarraybox}
\label{eq:partial sigma_cc' WMP Energy}
\end{equation}
and
\begin{equation}
    \begin{IEEEeqnarraybox}[][c]{rcl}
      \sigma_{c}(E)  & \ \underset{\mathcal{W}(E)}{\triangleq} \ & \sum_{n\geq -2} a_n^{c} k_c^n(E) +  \frac{1}{E}  \Re_{\mathrm{conj}}\left[ \sum_{j\geq 1} \frac{R_j^{c} }{E - \mathcal{E}_j}  \right]
    \IEEEstrut\end{IEEEeqnarraybox}
     \label{eq:total σ_c WMP Energy}
\end{equation}
where the residues are obtained by evaluating at the pole values as: 
\begin{equation}
\begin{IEEEeqnarraybox}[][c]{rcl}
\widetilde{R}_j^{cc'} & \ \triangleq \ &  -\mathrm{i}  \frac{4\pi g_{J^\pi_c} \mathcal{E}_j}{\left|k_c(\mathcal{E}_j)\right|^2}\left[\boldsymbol{\widetilde{\tau}_j}\right]_{cc'}  
\IEEEstrut\end{IEEEeqnarraybox}
\label{eq:partial sigma_cc' WMP Residues Energy}
\end{equation}
and
\begin{equation}
    \begin{IEEEeqnarraybox}[][c]{rcl}
      R_j^{c} & \ \triangleq \ & -\mathrm{i}  \frac{4 \pi g_{J^\pi_c}\mathcal{E}_j}{ \left|k_c(\mathcal{E}_j)\right|^2} \left[\boldsymbol{\tau_j}\boldsymbol{\tau_j}^\mathsf{T}\right]_{cc}
    \IEEEstrut\end{IEEEeqnarraybox}
     \label{eq:total σ_c WMP Residues Energy }
\end{equation}

Equivalently, the Windowed Multipoles Representation can be carried out in square root of energy $z$-space, in what constitutes theorem \ref{theo::WMP Representation}.

\begin{theorem}\label{theo::WMP Representation}\textsc{Windowed Multipole Formalism} \\
Let $\mathcal{E}_j$ be the energy-space poles of the Kapur-Peierls operator $\boldsymbol{R}_L$, defined in (\ref{eq: def Kapur-Peierls operator}), and let $z \triangleq \sqrt{E}$ be the square root of energy.
The energy dependence of R-matrix cross sections can be exactly expressed as a Laurent expansion in wavenumber $k_c$, of order no less than $k_c^{-2}$, plus the conjugate continuation real part (\ref{eq::continued conjugate real part}) of a sum of energy-space resonances with poles $\mathcal{E}_j$, which in $z$-space yield pairs of opposite poles $p_j = \pm \sqrt{\mathcal{E}_j}$, so that partial cross sections (\ref{eq:partial sigma_cc'}) take the Windowed Multipole Representation:
\begin{equation}
\begin{IEEEeqnarraybox}[][c]{rcl}
\sigma_{cc'}(z) & \ \underset{\mathcal{W}(z)}{\triangleq} \ & \sum_{n\geq -2} \widetilde{a}_n^{cc'} k_c^n(z) +  \frac{1}{z^2}  \Re_{\mathrm{conj}}\left[ \sum_{j\geq 1} \frac{\widetilde{r}_j^{cc'} }{z - p_j}  \right]
\IEEEstrut\end{IEEEeqnarraybox}
\label{eq:partial sigma_cc' WMP z-space}
\end{equation}
and the total cross section (\ref{eq:total σ_c}) takes the form:
\begin{equation}
    \begin{IEEEeqnarraybox}[][c]{rcl}
      \sigma_{c}(z)  & \ \underset{\mathcal{W}(z)}{\triangleq} \ & \sum_{n\geq -2} a_n^{c} k_c^n(z) +  \frac{1}{z^2}  \Re_{\mathrm{conj}}\left[ \sum_{j\geq 1} \frac{r_j^{c} }{z - p_j}  \right]
    \IEEEstrut\end{IEEEeqnarraybox}
     \label{eq:total σ_c WMP z-space}
\end{equation}
where the partial residues can be constructed from R-matrix parameters as:
\begin{equation}
\begin{IEEEeqnarraybox}[][c]{rcl}
\widetilde{r}_j^{cc'} & \ \triangleq \ &  -\mathrm{i}  \frac{4\pi g_{J^\pi_c} p_j^2}{\left|k_c(p_j)\right|^2}\left[2\cdot \boldsymbol{\zeta_j}\boldsymbol{\zeta_j}^\mathsf{T} \circ \boldsymbol{T}\left(p_j^*\right)^*\right]_{cc'}  
\IEEEstrut\end{IEEEeqnarraybox}
\label{eq:partial sigma_cc' WMP Residues z-space}
\end{equation}
and the total residues as
\begin{equation}
    \begin{IEEEeqnarraybox}[][c]{rcl}
      r_j^{c} & \ \triangleq \ & -\mathrm{i}  \frac{4 \pi g_{J^\pi_c} p_j^2}{ \left|k_c(p_j)\right|^2} \left[\boldsymbol{\zeta_j}\boldsymbol{\zeta_j}^\mathsf{T}\right]_{cc}
    \IEEEstrut\end{IEEEeqnarraybox}
     \label{eq:total σ_c WMP Residues z-space}
\end{equation}
where the $\boldsymbol{\zeta_j}$ residue widths vectors are linked to the Kapur-Peierls operator $\boldsymbol{R}_L$ poles and residues through relations (\ref{eq: link zeta to r residues}).

Alternatively, the residues can be numerically obtained through Cauchy's residues theorem contour integrals
\begin{equation}
\begin{IEEEeqnarraybox}[][c]{rcl}
\widetilde{r}_j^{cc'} & \ = \ &  \frac{1}{\mathrm{i}\pi}\oint_{\mathfrak{C}_{p_j}} z^2 \sigma_{cc'}(z) \mathrm{d}z
\IEEEstrut\end{IEEEeqnarraybox}
\label{eq: WMP partial residues contour integrals}
\end{equation}
where $\mathfrak{C}_{p_j}$ designates a positively oriented simple closed contour containing only pole $p_j$. 
For instance, if $\mathfrak{C}_{p_j}$ is a circle of small radius $\epsilon > 0 $ around pole $p_j$, this yields
\begin{equation}
\begin{IEEEeqnarraybox}[][c]{rcl}
r_j^{c} & \ = \ &  \frac{\epsilon}{\pi}\int_{\theta = 0}^{2\pi} \left( p_j + \epsilon \mathrm{e}^{\mathrm{i}\theta}\right)^2 \sigma_{c}\left( p_j + \epsilon \mathrm{e}^{\mathrm{i}\theta}\right) \mathrm{e}^{\mathrm{i}\theta} \mathrm{d}\theta
\IEEEstrut\end{IEEEeqnarraybox}
\label{eq: WMP total residues contour integrals}
\end{equation}
In order to perform these contour integrals, R-matrix cross sections (\ref{eq:partial sigma_cc'}) and (\ref{eq:total σ_c}) must have been meromorphically continued to complex energies by means of conjugate continuations (\ref{eq::continued square modulus definition}) and (\ref{eq::continued conjugate real part}) respectively. 
\end{theorem}

Therefore, by solving the radioactive problem (\ref{eq:R_L radioactive problem}) -- or level-matrix one (\ref{eq::invA det roots}) -- to find the poles $\mathcal{E}_j$ and residues $\boldsymbol{r_j}$ of the Kapur-Peierls operator (respectively $p_j$ and $\boldsymbol{\kappa_j}$ from (\ref{eq:R_L residues normalization z-space}) or level-matrix equivalent (\ref{eq::invA det roots z-space}) in $z$-space), we can compute the transmission matrix residues $\boldsymbol{\tau_j}$ from (\ref{eq: T_j as function of r_j}) and the conjugate continuation ones $\boldsymbol{\widetilde{\tau}_j}$ from (\ref{eq: tau tilde}) (respectively $\boldsymbol{\zeta_j}$ from (\ref{eq: link zeta to r residues}) in $z$-space), to find the poles and residues of the Windowed Multipole Representation of R-matrix cross sections, through equations (\ref{eq:partial sigma_cc' WMP Energy}), (\ref{eq:total σ_c WMP Energy}), (\ref{eq:partial sigma_cc' WMP Residues Energy}), and (\ref{eq:total σ_c WMP Residues Energy }); or respectively equations (\ref{eq:partial sigma_cc' WMP z-space}), (\ref{eq:total σ_c WMP z-space}), (\ref{eq:partial sigma_cc' WMP Residues z-space}), and (\ref{eq:total σ_c WMP Residues z-space}) for $z$-space. 

\begin{figure}[ht!!] % replace 't' with 'b' to force it to be on the bottom
  \centering
  \subfigure[\ $^{\text{238}}\text{U}$ first resonances (3 s-waves and 4 p-waves).]{\includegraphics[width=0.42\textwidth]{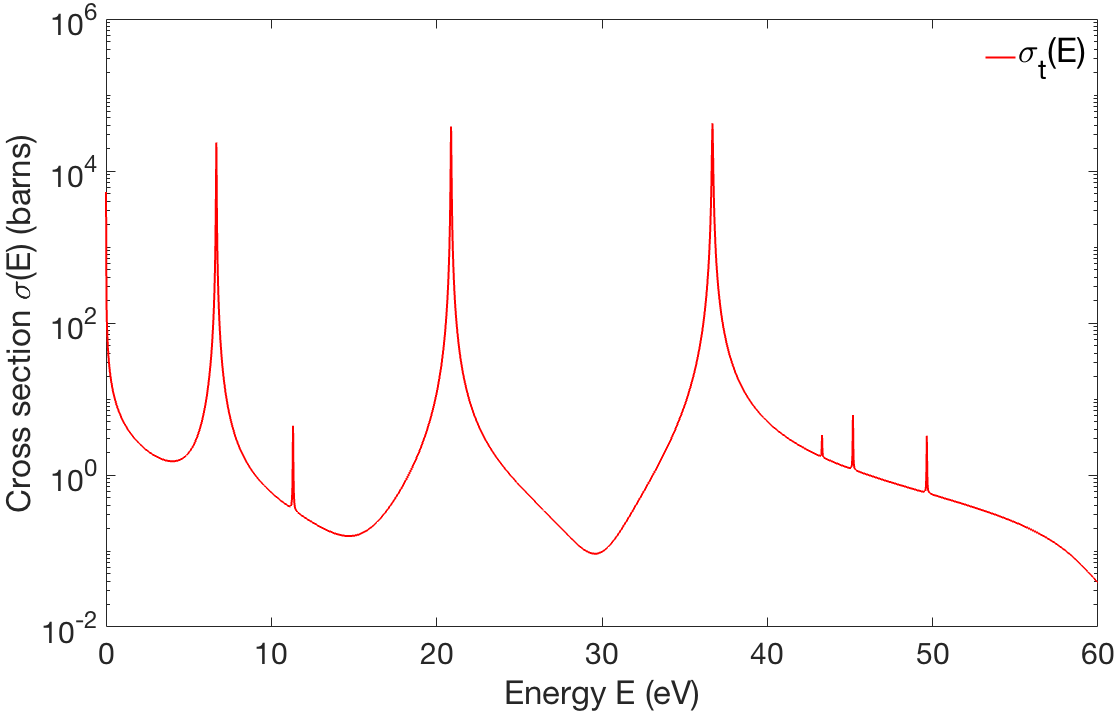}}
  \subfigure[\ $^{\text{238}}\text{U}$ windowed multipole cross section surface.\label{fig: Windowed Multipole Representation with belox-threshold continuation}]{\includegraphics[width=0.47\textwidth]{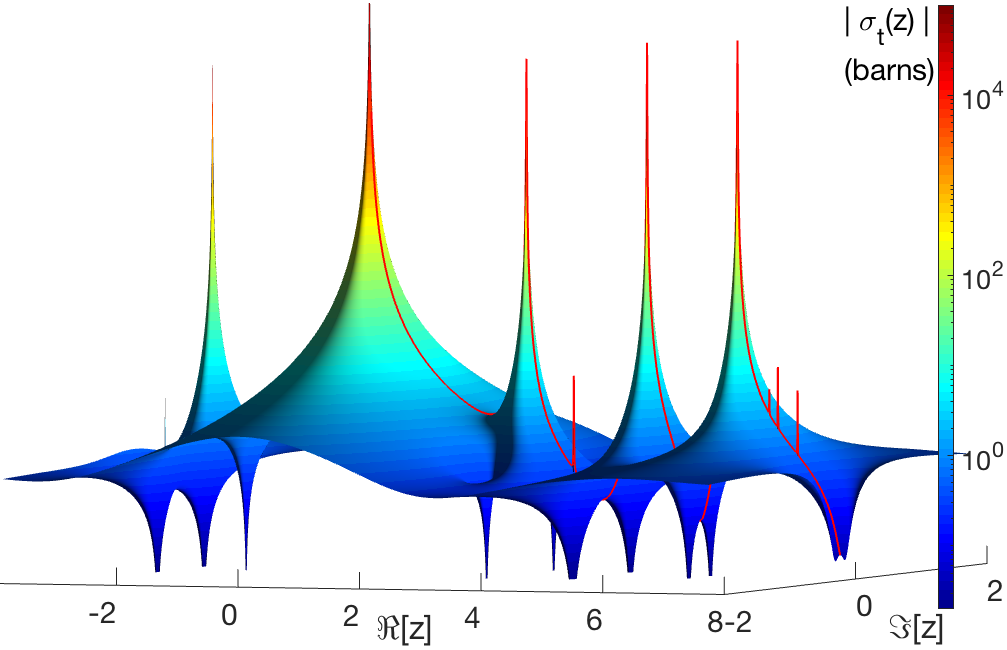}}
  \subfigure[\ $^{\text{238}}\text{U}$ first s-wave resonance peak.\label{fig: Windowed Multipole Representation s-wave peak and contour}]{\includegraphics[width=0.46\textwidth]{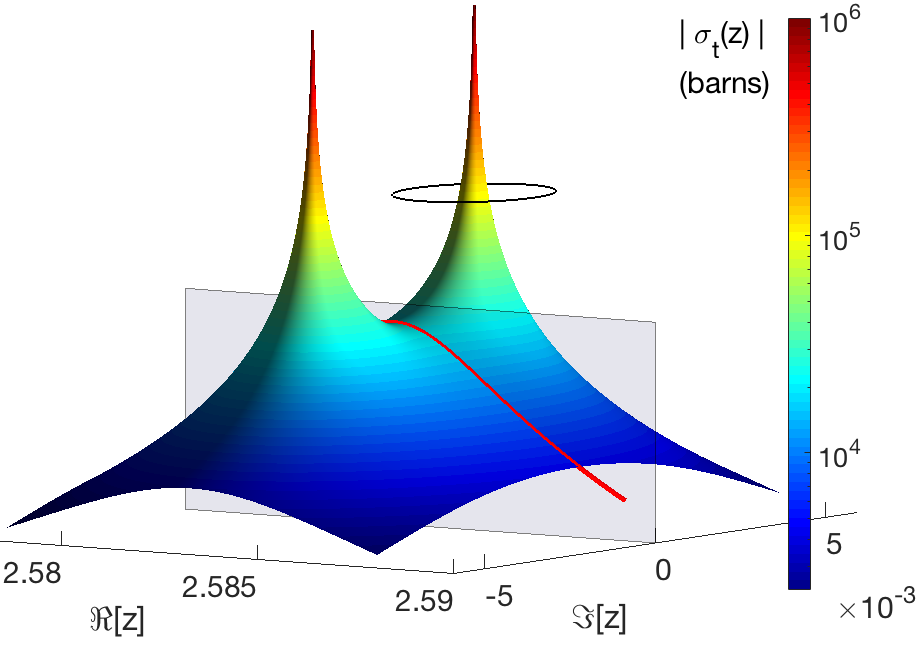}}
  \caption{\small{Windowed multipole representation of R-matrix cross sections: $^{\text{238}}\text{U}$ total cross section (minus potential scattering) meromorphic continuation into the complex $z$-plane, for $z = \pm \sqrt{E}$ in ($\mathrm{\sqrt{eV}}$). This surface's crest and thalweg line along the real axis is the R-matrix cross section above the zero threshold. FIG. \ref{fig: Windowed Multipole Representation with belox-threshold continuation} shows the resonance peaks are the saddle points between the complex conjugate poles. Negative $z$ in FIG.\ref{fig: Windowed Multipole Representation with belox-threshold continuation} are on the shadow branch $\left\{ E, - \right\}$ of mapping (\ref{eq:rho_c massive}). The black circle in FIG.\ref{fig: Windowed Multipole Representation s-wave peak and contour} represents the contour integrals around the poles of the complex cross section which enable both conversion to windowed multipole covariances (theorem \ref{theo::WMP covariance}) and analytic Doppler broadening (theorem \ref{theo::WMP Doppler broadening}).}}
  \label{fig: Windowed Multipole Representation}
\end{figure}

%%%%%%%%%%%%%%%%%%%%%%%%%%%%%%%%%%%%%%%%%%%%%%%%%%%%%%%%%%%%%%%%%%%%%%%%%%%%%%%%
\subsubsection{\label{subsubsec:Windowed Pole Representation: potential and resonance matrices approach} Windowed Pole Representation: potential and resonance matrices approach}
%%%%%%%%%%%%%%%%%%%%%%%%%%%%%%%%%%%%%%%%%%%%%%%%%%%%%%%%%%%%%%%%%%%%%%%%%%%%%%%%

The transmission matrix approach is exact, but it has three drawbacks: 1) it is not simple to interpret physically; 2) it does not give us information on the ``background'' behavior (non-resonant Laurent expansion $\sum_{n\geq-2} a_n k_c^n$); 3) it can be numerically unstable. 
Decomposition (\ref{eq:Transmission matrix decomposition Z and D}) of the transmission matrix helps us separate the cross sections into parts we can interpret physically: the potential cross section $\sigma_c^{\text{pot}}$ has no resonances (\ref{eq: potential cross section}); the reaction cross section $\sigma_{cc'}^{\text{react}}$ has all the resonances (\ref{eq: reaction cross section}); and both the partial cross section $\sigma_{cc'}$ from (\ref{eq: partial cross section breakdown}) and the total cross section $\sigma_c$ from (\ref{eq:total sigma_c WMP z-space decomposed}) also have interference resonances from the real part of the resonance matrix $\boldsymbol{Z}$. This means all the resonances of R-matrix cross sections can be recovered from the resonance matrix $\boldsymbol{Z}$ Mittag Leffler expansion (\ref{eq::Z Mittag Leffler in z-space}).
Applying Hwang's conjugate continuation method to construct the Windowed Multipole Representation then yields:

\begin{itemize}
    \item Potential scattering cross section (of channel $c$):
\begin{equation}
\begin{IEEEeqnarraybox}[][l]{rcl}
\sigma_c^{\text{pot}}(E) &  \ \underset{\mathcal{W}(z)}{=}  \ &  4\pi g_{J^\pi_c} \frac{\Re\left[\mathrm{i}D_{c}\right]}{\left|k_c\right|^2}  
\IEEEstrut\end{IEEEeqnarraybox}
\label{eq: potential cross section WMP}
\end{equation}
    \item Total cross section (of channel $c$):
\begin{equation}
\begin{IEEEeqnarraybox}[][l]{rcl}
\sigma_c(z) &  \underset{\mathcal{W}(z)}{\triangleq}   &  \sigma_c^{\text{pot}}(z) +  \frac{1}{z^2}  \Re_{\mathrm{conj}}\mkern-4mu\left[ \sum_{j\geq 1} \frac{r_j^{c} }{z - p_j}  \right] \mkern-3mu  +  \mkern-6mu \sum_{n\geq -2} b_n^{c} k_c^n(z) 
\IEEEstrut\end{IEEEeqnarraybox}
\label{eq:total sigma_c WMP z-space decomposed}
\end{equation}
where the total residues $r_j^{c} $ are defined in (\ref{eq:total σ_c WMP Residues z-space}).
    \item Self-scattering cross section (of channel $c$):
\begin{equation}
\begin{IEEEeqnarraybox}[][l]{rcl}
\sigma_c^{\text{scat}}(E) &  \ \underset{\mathcal{W}(z)}{=}  \ &  \frac{1}{z^2}  \Re_{\mathrm{conj}}\left[ \sum_{j\geq 1} \frac{^{\text{scat}}r_j^{c} }{z - p_j}  \right] + \sum_{n\geq -2} c_n^{c} k_c^n(z)\delta_{cc'}
\IEEEstrut\end{IEEEeqnarraybox}
\label{eq: scatering cross section WMP}
\end{equation}
with scattering residues:
\begin{equation}
\begin{IEEEeqnarraybox}[][c]{rcl}
^{\text{scat}}r_j^{c} & \ \triangleq \ &  -\frac{4\pi g_{J^\pi_c} p_j^2}{\left|k_c(p_j)\right|^2}\left[ 2\cdot \boldsymbol{\zeta_j}\boldsymbol{\zeta_j}^\mathsf{T} \circ \boldsymbol{D}\left(p_j^*\right)^*\right]_{cc'}\delta_{cc'}
\IEEEstrut\end{IEEEeqnarraybox}
\label{eq: scattering residues WMP}
\end{equation}
    \item Interference cross section (of channel $c$):
\begin{equation}
\begin{IEEEeqnarraybox}[][l]{rcl}
\sigma_c^{\text{int}}(E) &  \ \underset{\mathcal{W}(z)}{=}  \ &  \frac{1}{z^2}  \Re_{\mathrm{conj}}\left[ \sum_{j\geq 1} \frac{^{\text{int}}r_j^{c} }{z - p_j}  \right] + \sum_{n\geq -2} d_n^{c} k_c^n(z)\delta_{cc'}
\IEEEstrut\end{IEEEeqnarraybox}
\label{eq: interference cross section WMP}
\end{equation}
with interference residues:
\begin{equation}
\begin{IEEEeqnarraybox}[][c]{rcl}
^{\text{int}}r_j^{c} & \ \triangleq \ &  -\mathrm{i} \frac{4\pi g_{J^\pi_c} p_j^2}{\left|k_c(p_j)\right|^2}\left[ \boldsymbol{\zeta_j}\boldsymbol{\zeta_j}^\mathsf{T} \circ \boldsymbol{Y}\left(p_j^*\right)^*\right]_{cc'}\delta_{cc'}
\IEEEstrut\end{IEEEeqnarraybox}
\label{eq: interference residues WMP}
\end{equation}
    \item Reaction cross section (from channel $c$ to $c'$):
\begin{equation}
\begin{IEEEeqnarraybox}[][l]{rcl}
\sigma_{cc'}^{\text{react}}(z)  & \ \underset{\mathcal{W}(z)}{=}  \ &   \frac{1}{z^2}  \Re_{\mathrm{conj}}\left[ \sum_{j\geq 1} \frac{^{\text{react}}r_j^{cc'} }{z - p_j}  \right] + \sum_{n\geq -2} \widetilde{b}_n^{cc'} k_c^n(z)
\IEEEstrut\end{IEEEeqnarraybox}
\label{eq: reaction cross section WMP}
\end{equation}
with reaction residues:
\begin{equation}
\begin{IEEEeqnarraybox}[][c]{rcl}
^{\text{react}}r_j^{cc'} & \ \triangleq \ &  \frac{4\pi g_{J^\pi_c} p_j^2}{\left|k_c(p_j)\right|^2}\left[2\cdot \boldsymbol{\zeta_j}\boldsymbol{\zeta_j}^\mathsf{T} \circ \boldsymbol{Z}\left(p_j^*\right)^*\right]_{cc'}  
\IEEEstrut\end{IEEEeqnarraybox}
\label{eq: reaction residues WMP}
\end{equation}
    \item Partial (angle-integrated) cross section (\ref{eq: partial cross section breakdown}) (from channel $c$ to $c'$):
\begin{equation*}
\begin{IEEEeqnarraybox}[][l]{rcl}
\sigma_{cc'}(E) & \ = \ & \Big( \sigma_c^{\text{pot}}(E) + \sigma_c^{\text{scat}}(E) \Big)\delta_{cc'}  + \sigma_{cc'}^{\text{react}}(E)  \\
& \ = \ & \Big( \sigma_c^{\text{tot}}(E) - \sigma_c^{\text{int}}(E) \Big)\delta_{cc'}  + \sigma_{cc'}^{\text{react}}(E)
\IEEEstrut\end{IEEEeqnarraybox}
%\label{eq: partial cross section}
\end{equation*}
Noticing that $-\mathrm{i}\boldsymbol{T}^* = \boldsymbol{Z}^* - \boldsymbol{D}^*$, this entails the partial residues $\widetilde{r}_j^{cc'} $ from (\ref{eq:partial sigma_cc' WMP Residues z-space}) are connected to the total residues $r_j^{c} $ from (\ref{eq:total σ_c WMP Residues z-space}), the reaction residues $^{\text{react}}r_j^{cc'}$ from (\ref{eq: reaction residues WMP}), the scattering residues $^{\text{scat}}r_j^{c}$ from (\ref{eq: scattering residues WMP}), and the interference residues $^{\text{int}}r_j^{c}$ from (\ref{eq: interference residues WMP}), according to:
\begin{equation}
\begin{IEEEeqnarraybox}[][c]{rcl}
\widetilde{r}_j^{cc'}  & \  =  \ &  ^{\text{react}}r_j^{cc'} + ^{\text{scat}}r_j^{c} \\
& \  =  \ &  ^{\text{react}}r_j^{cc'} + r_j^{c} - ^{\text{int}}r_j^{c}
\IEEEstrut\end{IEEEeqnarraybox}
\label{eq:reaction cross section WMP Residues z-space decomposed}
\end{equation}
\end{itemize}

Total cross section decomposition (\ref{eq:total sigma_c WMP z-space decomposed}) is simpler to interpret physically than expression (\ref{eq:total σ_c WMP Residues z-space}) directly derived from the transmission matrix, because the potential cross section $\sigma_c^{\text{pot}}$ is extracted from the background Laurent expansion: $\sum_{n\geq -2} a_n^c z^n$. 
The same holds for the partial cross section (\ref{eq:partial sigma_cc' WMP z-space}), where the residues decomposition (\ref{eq:reaction cross section WMP Residues z-space decomposed}) untangles the direct expression (\ref{eq:partial sigma_cc' WMP Residues z-space}) from the transmission matrix approach. Though mathematically equivalent, some of these approaches may be more numerically stable than others. 

Importantly, we do not need the poles of the potential matrix $\boldsymbol{D}$ to express the partial and total cross sections. This is because any such poles (the zeros of $\boldsymbol{H_+}$) cancel out of the scattering matrix (\ref{eq:U expression}), and therefore of the cross sections. Before we proved this result in theorem 3 of \cite{Ducru_Scattering_Matrix_of_Complex_Wavenumbers_2019}, Hwang had to explicitly decompose the potential cross section $\sigma_c^{\text{pot}}$ into poles and residues in eq. (1) and (2) of \cite{Hwang_1998} (also in eq. (3) and (4) of \cite{Hwang_2003}), with severe numerical instability implications which he attempted to remedy by introducing pseudo-poles in \cite{Hwang_1992}. We now known that under proper analytic continuation, these spurious poles have zero residues in the transmission matrix, and thus cancel out of the partial and total cross sections.

%%%%%%%%%%%%%%%%%%%%%%%%%%%%%%%%%%%%%%%%%%%%%%%%%%%%%%%%%%%%%%%%%%%%%%%%%%%%%%%%
\subsubsection{\label{subsubsec:Pole expansion: R-matrix construct or rational fit}Pole expansion: R-matrix construct or rational fit}
%%%%%%%%%%%%%%%%%%%%%%%%%%%%%%%%%%%%%%%%%%%%%%%%%%%%%%%%%%%%%%%%%%%%%%%%%%%%%%%%

So far, we have constructed the transmission matrix Mittag-Leffler expansion (\ref{eq::T Mittag Leffler z-space}) by first solving the radioactive states problem (\ref{eq:R_L radioactive problem in z-space}) and then obtaining the transmission matrix residues from those of the Kapur-Peierls operator, through (\ref{eq: link zeta to r residues}). 
One could dispense of the intermediary steps and find the radioactive poles $\left\{p_j\right\}$ directly through the transmission matrix by solving the generalized eigenvalue problem 
\begin{equation}
\left.\boldsymbol{T}^{-1}(z)\right|_{z = p_j} \boldsymbol{\zeta_j} = \boldsymbol{0}
\label{eq:T Gohberg-Sigal in z-space}
\end{equation}
and subjecting the residue widths vectors $\left\{\boldsymbol{\zeta_j}\right\}$ to the following normalization:
\begin{equation}
\boldsymbol{\zeta_j}^\mathsf{T} \left( \left. { \frac{\partial \boldsymbol{T}^{-1}}{\partial z} }\right|_{z=p_j} \right) \boldsymbol{\zeta_j} = \mathrm{i} 
\label{eq:T residues normalization z-space}
\end{equation}
Though mathematically equivalent, this all-in-one approach can nonetheless be prone to numerical instabilities. 
Which leads us to the question of how to numerically solve the generalized eigenproblems - either the radioactive ones (\ref{eq:R_L radioactive problem in z-space}) or directly (\ref{eq:T Gohberg-Sigal in z-space}). On this issue, we direct the reader to the section of theorem 1 in \cite{Ducru_Scattering_Matrix_of_Complex_Wavenumbers_2019} for a more detailed discussion, in particular on the multi-sheeted nature of the Riemann mapping (\ref{eq:rho_c(E) mapping}), which can complicate the search for solutions.
We will here simply state that these are nonlinear eigenvalue problems, and general algorithms to solve them can be found in the Handbook of Linear Algebra \cite{Handbook_of_linear_algebra}, chapter 115.
One such algorithm is the Rayleigh-quotient method, used by Brune to find alternative parameters in \cite{Brune_2002}. 
Alternatively, it is sometimes more computationally advantageous to first find the radioactive poles $\left\{p_j\right\}$ directly by solving the channel determinant problem, $\mathrm{det}\left( \left.\boldsymbol{R}_{L}^{-1}(z)\right|_{z = p_j} \right) = 0$, or the corresponding level determinant one, $\mathrm{det}\left( \left.\boldsymbol{A}^{-1}(z)\right|_{z = p_j} \right) = 0$, and to second solve the associated eigenvalue problem (which is now linear), or even to directly evaluate the residues at the found poles by contour integrals (\ref{eq: WMP partial residues contour integrals}) and (\ref{eq: WMP total residues contour integrals}). Such methods tailored to find all the roots of the radioactive problem where introduced in \cite{Ducru_PHYSOR_conversion_2016}, in section 5 of \cite{Analytic_Benchmark_1_2020}, or in equations (200) and (204) of \cite{Frohner_Jeff_2000}. 
Also, solving the Kapur Peierls radioactive problem (\ref{eq:R_L radioactive problem in z-space}) will be advantageous over solving the level matrix one (\ref{eq::invA det roots z-space}) when the number of levels $N_\lambda$ far exceeds the number of channels $N_c$, and conversely.

Rather than starting from the Wigner-Eisenbud R-matrix resonance parameters $\Big\{ E_{T_c}, a_c, B_c, E_{\lambda}, \gamma_{\lambda,c} \Big\}$ to construct the Windowed Multipole Representation poles $p_j$ and residues $\widetilde{r}_j^{cc'}$ and $r_j^{c}$ as (\ref{eq:partial sigma_cc' WMP Residues z-space}) and (\ref{eq:total σ_c WMP Residues z-space}), an alternative approach is to simply curve-fit the point-wise energy-dependence of nuclear cross sections $\sigma_{cc}(E)$ with the corresponding Windowed Multipole Representation forms (\ref{eq:partial sigma_cc' WMP z-space}) and (\ref{eq:total σ_c WMP z-space}). For instance, this approach was successfully deployed in \cite{Liang_Ducru_ANS_2017} and in \cite{Multipole_regulatized_VF_2018}, where using the 
``black-box'' rational function approximating algorithm called ``vector-fitting'' \cite{Gustaven_VF_1999,Gustaven_RVF_2006} led to finding the exact resonant radioactive poles $\left\{p_j\right\}$ of $^{\text{16}}\text{O}$, for which no resonance parameters were published \cite{Multipole_regulatized_VF_2018}. This conversion of point-wise R-matrix cross sections into windowed multipoles representation approach was generalized to most of the nuclides found in the  ENDF/B-VII.1 nuclear data library \cite{liuGenerationWindowedMultipole2018, Jingang_Liang_2018_PHYSOR}, and could potentially be facilitated by recent advances in rational approximation algorithms -- such as RKFIT \cite{berljafaRKFITAlgorithmNonlinear2017} or AAA \cite{nakatsukasaAAAAlgorithmRational2018}.

%%%%%%%%%%%%%%%%%%%%%%%%%%%%%%%%%%%%%%%%%%%%%%%%%%%%%%%%%%%%%%%%%%%%%%%%%%%%%%%%
\subsubsection{\label{subsubsec:Windowing process: Laurent background fit} Windowing process: Laurent background fit}
%%%%%%%%%%%%%%%%%%%%%%%%%%%%%%%%%%%%%%%%%%%%%%%%%%%%%%%%%%%%%%%%%%%%%%%%%%%%%%%%

Regardless of the method deployed to find the poles $\big\{ p_j \big\}$ and their corresponding residues, there exists no general way to construct the thresholds Laurent expansions, $\sum_{n\geq-2} a_n k_c(z)^n$, from the R-matrix parameters. 
One must thus select an energy window $\mathcal{W}(E)$ and curve fit the background Laurent expansion $\sum_{n\geq-2} a_n k_c(z)^n$ by subtracting the resonances, that is the poles contribution $\sum_{j\in\mathcal{W}(E)} \frac{r_j^c}{z-p_j} $. Nonetheless, there is a difficulty as to which such poles one should include explicitly into the window. It is not necessary to explicitly call all the poles $\big\{ p_j \big\}$ for each window $\mathcal{W}(E)$, rather the contribution of far-away poles is best curve-fitted and included in the Laurent expansion $\sum_{n\geq-2} a_n k_c(z)^n$. 
The criterion used to decide which poles $\left\{ p_j \right\}$ to include within each window is to select an accuracy bound for the Doppler broadened cross section, and include in window $\mathcal{W}(E)$ all the poles whose Doppler broadened resonances have a significant impact on the cross section within that window. Thus, the greater the maximum temperature, the more far-away poles have to be included to compute the cross section within window $\mathcal{W}(E)$. 
Once the contributing poles (after Doppler broadening) have been found, we subtract them from the zero-kelvin cross sections and curve-fit the difference with a Laurent-expansion $\sum_{n\geq-2} a_n k_c(z)^n$. More detailed explanations on this windowing process can be found in \cite{Forget_2013, joseyWindowedMultipoleSensitivity2015, Josey_JCP_2016}.
%BEN But that is not really necessary anymore if we can Doppler broadened the Laurent expansion?  Doesn't it now just become a matter of setting a maximum order of Laurent expansion needed?
%PABLO->BEN: indeed, need to know wether this makes sense: I thought that either you fit all-in-one point-wise, or you use WHOPPER and then substract the resonancesand curve-fit. However, the question remains as to the "outer-window" and "inner window", the former being determined by the accuracy of Doppler broadening on the latter. 

Though the background Laurent expansion must be numerically fitted, and that the resonant poles themselves may be accurately found using rational approximation ``black-box'' algorithms, it is critical to understand that the Windowed Multipole Representation (\ref{eq:partial sigma_cc' WMP z-space}) and (\ref{eq:total σ_c WMP z-space}) is not a curve-fitting approximation: this is a rigorous representation, mathematically and physically equivalent to the exact R-matrix theory cross sections (for real energies in open channels), or the Humblet-Rosenfeld pole expansions in wavenumber space. 
This can be tested by curve-fitting in $E$ and $k_c$ space, both the resonances and the background Laurent expansions. One will notice that the $E$-space Breit-Wigner profiles (\ref{eq:: SLBW Lorenzian profiles for resonances}) capture exactly one-for-one the resonance behavior. However, the threshold behaviors are not well represented by the $E$ variable: while few coefficients suffice to reach high accuracy using Laurent expansions in $k_c$ (usually no more than $a_{-2}$, $a_{-1}$, $a_{0}$, and $a_{1}$), many more expansion coefficients are necessary when fitting the background with Laurent expansions with powers of $E$. 

Finally, remember that for non-massless particles, wavenumber-energy mapping (\ref{eq:rho_c massive}) entails that: $k_c^2 \propto z^2 - E_{T_c}$. Thus, for zero-threshold reactions ($E_{T_c} = 0$), we have a direct proportionality $k_c \propto z$.
In order to achieve closed-form Doppler-broadening expressions, we may be willing to sacrifice the physically accurate Laurent expansion in $k_c$, and replace it with an approximation in powers of $z$ -- that is a Laurent expansion $\sum_{n\geq -2} a_n z^n$ -- plus rational Pad\'e-type approximations with simple poles -- that is adding non-physical pseudo-poles -- so as to approximate the exact threshold behavior $\sum_{n\geq-2} a_n k_c(z)^n$ with powers of $z$ and pseudo-poles $\sum_{n\geq-2} \widetilde{a}_n z^n + \sum_{n\geq 1} \frac{\widetilde{r}_n}{z-\widetilde{p}_n}$. Runge's theorem guarantees such approximation can always be performed to high-accuracy, though this is often costly, as many more pseudo-poles and Laurent expansion coefficients have to be introduced. Nonetheless, this approximation will have advantages when Doppler-broadening massive (not massless photons) particles (both charged and neutral), and it also provides a unified Windowed Multipole formalism:
\begin{equation}
\begin{IEEEeqnarraybox}[][c]{C}
     \sigma(z)  \underset{\mathcal{W}(E)}{=} \sum_{n\geq -2} a_n z^n + \frac{1}{z^2}\Re_{\mathrm{conj}}\left[\sum_{j \geq 1} \frac{r_{j}}{z-p_j}\right]
\label{eq:: sigma(z) Windowed Multipole representation}
\IEEEstrut\end{IEEEeqnarraybox}
\end{equation}
In addition to the residues (\ref{eq: WMP partial residues contour integrals}) and (\ref{eq: WMP total residues contour integrals}) of theorem \ref{theo::WMP Representation}, one can now also obtain the Laurent expansion coefficients by means of contour integrals:
\begin{equation}
\begin{IEEEeqnarraybox}[][c]{rcl}
a_{-2} & \ = \ &  \Re_{\mathrm{conj}}\left[\sum_{j \geq 1} \frac{r_{j}}{p_j}\right] + \frac{1}{2\mathrm{i}\pi}\oint_{\mathfrak{C}_{0}} z \cdot \sigma(z) \mathrm{d}z \\
a_{-1} & \ = \ &  \Re_{\mathrm{conj}}\left[\sum_{j \geq 1} \frac{r_{j}}{p_j^2}\right] + \frac{1}{2\mathrm{i}\pi}\oint_{\mathfrak{C}_{0}} \sigma(z) \mathrm{d}z \\ 
a_{n} & \ \underset{n\geq 0}{= } \ &  \frac{1}{2\mathrm{i}\pi}\oint_{\mathfrak{C}_{0}} \frac{\sigma(z)}{z^{n+1}} \mathrm{d}z    =  \mkern-3mu \frac{1}{2\pi \epsilon^n}\mkern-6mu \int_{\theta = 0}^{2\pi} \mkern-14mu \sigma \mkern-4mu \left(\mkern-2mu \epsilon \mathrm{e}^{\mathrm{i}\theta} \mkern-1mu \right) \mkern-4mu \mathrm{e}^{\mkern-3mu -\mathrm{i}n\theta}  \mkern-2mu \mathrm{d}\theta
\IEEEstrut\end{IEEEeqnarraybox}
\label{eq: WMP Laurent expansion coefficients by contour integrals}
\end{equation}
where $\mathfrak{C}_{0}$ designates a positively oriented simple closed contour containing only pole $0$, for instance a circle centered at zero with small radius $\epsilon > 0 $. Relations (\ref{eq: WMP Laurent expansion coefficients by contour integrals}) are obtained by performing partial fraction decomposition:
\begin{equation*}
\begin{IEEEeqnarraybox}[][c]{rcl}
\frac{1}{z^2}\frac{r_{j}/2}{z-p_j} & \ = \ & \frac{r_{j}/2}{p_j^2} \left[\frac{1}{z-p_j} - \frac{1}{z}  - \frac{p_j}{z^2}  \right]
\IEEEstrut\end{IEEEeqnarraybox}
\label{eq: partial fraction decomposition of each poles}
\end{equation*}

Therefore, converting R-matrix cross sections to the unified Windowed Multipole Representation formalism (\ref{eq:: sigma(z) Windowed Multipole representation}) is conceptually simple: it suffices to solve for the $z$-space poles $\left\{ p_j \right\}$ of the $\boldsymbol{A}$ level matrix (\ref{eq:inv_A expression}) -- that is radioactive problem (\ref{eq::invA det roots z-space}) -- and then perform contour integrals (\ref{eq: WMP total residues contour integrals}), (\ref{eq: WMP partial residues contour integrals}) and (\ref{eq: WMP Laurent expansion coefficients by contour integrals}) on the continued conjugate (\ref{eq::continued conjugate}) R-matrix cross sections (\ref{eq:total σ_c}) and (\ref{eq:partial sigma_cc'}) to find their residues and Laurent expansion coefficients.

Henceforth, we will only treat this unified Windowed Multipole Representation formalism (\ref{eq:: sigma(z) Windowed Multipole representation}): it is physically exact for any R-matrix cross section of zero-threshold, and an approximation of the exact Windowed Multipole representations (\ref{eq:partial sigma_cc' WMP z-space}) and (\ref{eq:total σ_c WMP z-space}) only in windows that include non-zero thresholds.

%%%%%%%%%%%%%%%%%%%%%%%%%%%%%%%%%%%%%%%%%%%%%%%%%%%%%%%%%%%%%%%%%%%%%%%%%%%%%%%%
%*******************************************************************************
%%%%%%%%%%%%%%%%%%%%%%%%%%%%%%%%%%%%%%%%%%%%%%%%%%%%%%%%%%%%%%%%%%%%%%%%%%%%%%%%
\subsection{\label{subsec:Hwang's special case: zero-threshold neutron cross sections}Hwang's special case: \\ zero-threshold neutron cross sections}
%%%%%%%%%%%%%%%%%%%%%%%%%%%%%%%%%%%%%%%%%%%%%%%%%%%%%%%%%%%%%%%%%%%%%%%%%%%%%%%%
%*******************************************************************************
%%%%%%%%%%%%%%%%%%%%%%%%%%%%%%%%%%%%%%%%%%%%%%%%%%%%%%%%%%%%%%%%%%%%%%%%%%%%%%%%

\begin{table*}
\caption{\label{tab::L_values_neutral} Reduced logarithmic derivative $L_\ell(\rho) \triangleq \frac{\rho}{O_\ell} \frac{\partial O_\ell}{\partial r}(\rho)$ of outgoing wavefunction $O_\ell(\rho)$, and $L_\ell^0(\rho) \triangleq L_\ell(\rho) - B_\ell $ using $B_\ell = - \ell$, irreducible forms and Mittag-Leffler pole expansions for neutral particles, for angular momenta $0 \leq \ell \leq 4$. }
\begin{ruledtabular}
\begin{tabular}{c|c|c|c|c}
\ \ & $L_\ell(\rho)$ from recurrence (11) of \cite{Ducru_shadow_Brune_Poles_2019} & $\begin{array}{c}
     L_\ell^0(\rho) \triangleq L_\ell(\rho) - B_\ell  \\
     \text{setting } B_\ell = - \ell 
\end{array}$ & $\begin{array}{c}
    L_\ell(\rho) \text{ from theorem 1 of \cite{Ducru_shadow_Brune_Poles_2019},}  \\
    \text{poles } \big\{ \omega_n\big\} \text{ from table II of \cite{Ducru_shadow_Brune_Poles_2019}} 
\end{array} $ & $\begin{array}{c}
     \text{Outgoing wavefunction } \\
     O_{\ell}(\rho) \text{ from  (16) of \cite{Ducru_shadow_Brune_Poles_2019}}
\end{array} $ \tabularnewline
\hline
$\ell$  &  $ L_\ell(\rho) = \frac{\rho^2 }{\ell - L_{\ell-1}(\rho)} - \ell $ & $ L_\ell^0(\rho) =  \frac{\rho^2}{2\ell -1 - L_{\ell-1}^0(\rho) }$   & $L_\ell(\rho) = -\ell + \mathrm{i} \rho + \sum_{n\geq 1} \frac{\rho}{\rho - \omega_n}$ &$ O_\ell(\rho)  =  \mathrm{e}^{\mathrm{i}\left(\rho + \frac{1}{2}\ell \pi \right)}\frac{\prod_{n \geq 1}\left(\rho-\omega_n\right)}{\rho^\ell}$  \tabularnewline
\hline \hline
0  &  $\mathrm{i}\rho$  & $\mathrm{i}\rho$ & $\left\{ \emptyset \right\}$ &  $\mathrm{e}^{\mathrm{i}\rho}  $\tabularnewline
1  &  $ \frac{-1 + \mathrm{i}\rho + \rho^2}{1-\mathrm{i}\rho}$ & $ \frac{\rho^2}{1-\mathrm{i}\rho}$  &  $\omega_{1}^{\ell = 2} = -\mathrm{i} $ & $\mathrm{e}^{\mathrm{i}\rho}\left(\frac{1}{\rho} - \mathrm{i}\right)  $ \tabularnewline
2   &  $ \frac{- 6 + 6\mathrm{i}\rho + 3 \rho^2 - \mathrm{i}\rho^3 }{3 - 3\mathrm{i}\rho - \rho^2} $ & $ \frac{ \rho^2 - \mathrm{i}\rho^3 }{3 - 3\mathrm{i}\rho - \rho^2} $ & $\omega_{1,2}^{\ell = 2} \approx \pm 0.86602 - 1.5\mathrm{i} $ &  $\mathrm{e}^{\mathrm{i}\rho}\left(\frac{3}{\rho^2} -\frac{3\mathrm{i}}{\rho} - 1\right)  $ \tabularnewline
3  &  $\frac{- 45 + 45 \mathrm{i} \rho + 21\rho^2 - 6\mathrm{i}\rho^3 -\rho^4}{15 - 15\mathrm{i}\rho - 6\rho^2 + \mathrm{i}\rho^3}$ & $\frac{3\rho^2 - 3\mathrm{i}\rho^3 -\rho^4}{15 - 15\mathrm{i}\rho - 6\rho^2 + \mathrm{i}\rho^3}$ & $\begin{array}{rl}
     \omega_1^{\ell = 3} & \approx - 2.32219 \mathrm{i}  \\
     \omega_{2,3}^{\ell = 3}  & \approx \pm 1.75438 - 1.83891 \mathrm{i}  
\end{array} $ &  $\mathrm{e}^{\mathrm{i}\rho}\left(\frac{15}{\rho^3} -\frac{15\mathrm{i}}{\rho^2} -\frac{6}{\rho} + \mathrm{i}\right)  $  \tabularnewline
4  &  $\frac{ - 420 + 420\mathrm{i}\rho + 195\rho^2 -55\mathrm{i}\rho^3 - 10\rho^4 + \mathrm{i}\rho^5}{105 - 105\mathrm{i}\rho -45\rho^2 + 10\mathrm{i}\rho^3 + \rho^4}$ & $\frac{ 15\rho^2 -15\mathrm{i}\rho^3 - 6\rho^4 + \mathrm{i}\rho^5}{105 - 105\mathrm{i}\rho -45\rho^2 + 10\mathrm{i}\rho^3 + \rho^4}$  &  $ \begin{array}{rl}
   \omega_{1,2}^{\ell = 4} & \approx \pm 2.65742 - 2.10379\mathrm{i}   \\
    \omega_{3,4}^{\ell = 4} & \approx \pm 0.867234 - 2.89621 \mathrm{i} 
       \end{array}$ & $\mathrm{e}^{\mathrm{i}\rho}\left(\frac{105}{\rho^4} - \frac{105\mathrm{i}}{\rho^3} -\frac{45}{\rho^2} + \frac{10 \mathrm{i}}{\rho} + 1 \right)  $ \tabularnewline
\end{tabular}
\end{ruledtabular}
\end{table*}

There is one special case where it is possible to fully and explicitly convert R-matrix parameters into their exact windowed multipole representation (\ref{eq:: sigma(z) Windowed Multipole representation}), without any need of curve-fitting or truncating the Laurent expansion: this is the case of neutron cross sections with no thresholds, which Hwang first investigated in \cite{Hwang_1987}. In this case, because all channels have zero energy threshold ($E_{T_c} = 0$), every channel's wavenumber is proportional to the square root of energy, $k_c \propto z$, we can therefore write the dimensionless wavenumber as:
\begin{equation}
\begin{IEEEeqnarraybox}[][c]{rcl}
\rho_c  & \ = \ & {\rho_0}_c \cdot z \\
{\rho_0}_c & \ \triangleq \ &   a_c \sqrt{\frac{2m_{c,1} m_{c,2}}{\left(m_{c,1}+m_{c,2}\right) \mathrm{\hbar}^2}}
\IEEEstrut\end{IEEEeqnarraybox}
\label{eq:rho_c massive no threshold}
\end{equation}
Moreover, there are no branch-points to mapping (\ref{eq:rho_c massive}) other than zero, so that the Windowed Multipole Representation (\ref{eq:: sigma(z) Windowed Multipole representation}) is exact and valid everywhere for positive energies $E > 0$: the Laurent development $\sum_{n\geq-2} a_n z^n$ at zero accurately describes the threshold behavior (as there is no exponential dampening from charges). Because there are no charges, the dimensionless Coulomb field parameter (\ref{eq: eta_c def}) is null, $\eta_c = 0$, so that the difference $\omega_c \triangleq \sigma_{\ell_c}(\eta_c) - \sigma_{0}(\eta_c) $ in Coulomb phase shift (\ref{eq: Coulomb phase shift def}) is such that we always have $\mathrm{e}
^{\mathrm{i}\omega_c} = 1$. From this, definitions (\ref{eq:def H_pm I and O}) entail that the incoming and outgoing wavefunctions are then simply the $H_-$ and $H_+$ combination of regular and irregular Bessel functions:
\begin{equation}
\begin{IEEEeqnarraybox}[][c]{rcl}
O(\rho)  & \ = \ &   {H_{+}}(\rho)  = G(\rho) + \mathrm{i} F(\rho) = \rho \left(- y_\ell(\rho) + \mathrm{i} j_\ell(\rho)\right) \\
I(\rho) & \ = \ &  {H_{-}}(\rho)  = G(\rho) - \mathrm{i} F(\rho) = \rho \left(- y_\ell(\rho) - \mathrm{i} j_\ell(\rho)\right)
\IEEEstrut\end{IEEEeqnarraybox}
\label{eq:def H_pm I and O neutral }
\end{equation}
where $j_\ell(\rho)$ is the spherical Bessel function of the first kind, and $y_\ell(\rho)$ is the spherical Bessel function of the second kind, respectively defined in chapter 10, eq.(10.47.3) and eq.(10.47.4) of \cite{NIST_DLMF}, or in chapter 10, eq.(10.1.1) of Abramowitz \& Stegun \cite{Abramowitz_and_Stegun}.
This in turn entails the remarkable property that the reduced logarithmic derivative (\ref{eq: L operator}) of the outgoing-wave function $L(\rho)$ is now a rational function (that is the ratio of polynomials) in $\rho$, whose expressions, along with those of $O(\rho)$, are here reported in table \ref{tab::L_values_neutral}, and we refer to section II.B of \cite{Ducru_shadow_Brune_Poles_2019} for a more detailed description of these functions.

%%%%%%%%%%%%%%%%%%%%%%%%%%%%%%%%%%%%%%%%%%%%%%%%%%%%%%%%%%%%%%%%%%%%%%%%%%%%%%%%
\subsubsection{\label{subsubsec:Solving the radioactive states problem: polynomial rootfinding}Solving the radioactive states problem: \\ polynomial rootfinding}
%%%%%%%%%%%%%%%%%%%%%%%%%%%%%%%%%%%%%%%%%%%%%%%%%%%%%%%%%%%%%%%%%%%%%%%%%%%%%%%%

Crucially, in this special case of only neutron channels without threshold, the fact that $L(\rho)$ is now a rational function in $z$ entails that (\ref{eq: def Kapur-Peierls operator}), the Kapur-Peierls operator $\boldsymbol{R}_L$, is also a rational function in $z$-space. 
Therefore, the radioactive problem (\ref{eq:R_L radioactive problem in z-space}) itself becomes that of finding the roots of a rational function. 
We solve the radioactive problem through the level matrix approach (\ref{eq::invA det roots z-space}), where the residue width vectors are normalized as (\ref{eq:A residues normalization z-space}), which we can calculate through (\ref{eq: level matrix energy derivatives z-sapce}) where the partial derivative (\ref{eq: partial L partial z}) is now simply $ \frac{\partial \rho}{\partial z} = {\rho_0}_c $ from (\ref{eq:rho_c massive no threshold}).
The key is now to find the radioactive poles $\big\{p_j\big\}$ in $z$-space.
We can do so by solving for the roots of  the inverse level matrix (\ref{eq:inv_A expression}) determinant: $\mathrm{det}\left( \left.\boldsymbol{A}^{-1}(z)\right|_{z = p_j} \right) = 0$ . Since this determinant is a rational function in $z$, one can find all its zeros by expressing it in irreducible form, and solve for all the roots of the numerator polynomial. This can be accomplished by developing the determinant and applying lemma 3 of \cite{Ducru_shadow_Brune_Poles_2019} on \textit{diagonal divisibility and capped multiplicities}, in an analogous fashion as in the proof of theorem 3 in \cite{Ducru_shadow_Brune_Poles_2019}, to which we point for more detailed explanations.
More precisely, one can see in table \ref{tab::L_values_neutral} that $L_\ell(\rho) = \frac{u_{\ell + 1}(\rho)}{q_\ell(\rho)}$ is a proper rational function with simple poles, with a denominator $q_\ell(\rho)$ of degree $\ell$ and a numerator $u_{\ell + 1}(\rho)$ of degree $\ell + 1$. 
The polynomial factor $Q(z) $ that makes the denominator of the $\mathrm{det}\left(\boldsymbol{A}^{-1} \right)(z) $ rational function irreducible can then be found by applying lemma 3 of \cite{Ducru_shadow_Brune_Poles_2019} on diagonal divisibility and capped multiplicities, yielding 
\begin{equation}
\begin{IEEEeqnarraybox}[][c]{rcl}
Q(z) & \ \triangleq \ &  \prod_{c=1}^{N_c} q_{\ell_c}(z)
\IEEEstrut\end{IEEEeqnarraybox}
\label{eq:irreductible Q}
\end{equation}
so that for only neutron channels without thresholds, finding all the radioactive poles $\big\{p_j\big\}$ is akin to solving for all the roots of the following polynomial:
\begin{equation}
\begin{IEEEeqnarraybox}[][c]{rcl}
Q(z) \cdot \left.\mathrm{det}\left( \boldsymbol{e} - z^2\Id{} - \boldsymbol{\gamma}\left( \boldsymbol{L}(z) - \boldsymbol{B} \right)\boldsymbol{\gamma}^\mathsf{T} \right)\right|_{z=p_j} & \ = \ &  0
\IEEEstrut\end{IEEEeqnarraybox}
\label{eq: no threshold neutrons radioactive problem}
\end{equation}
The degree of this polynomial, and thus the number of (complex) roots $\big\{p_j\big\}$, is:
\begin{equation}
N_L = 2 N_\lambda  + \sum_{c=1}^{N_c} \ell_c 
\label{eq::NL number of poles no threshold neutrons}
\end{equation}
which is a particular case of the general number of radioactive poles $N_L$ we stated in (\ref{eq::NL number of poles}) (and proved in theorem 1 of \cite{Ducru_Scattering_Matrix_of_Complex_Wavenumbers_2019}), but with only one threshold, $E_{T_c} = 0$, so that the number of different thresholds is one: $N_{E_{T_c} \neq E_{T_{c'}}} = 1 $. 
%*BEN Why the not equal? 
%PABLO->BEN: help me phrase this please: in writing $N_{E_{T_c} \neq E_{T_{c'}}} = 1 $, I mean 'the number of different thresholds is 1, because there is only one threshold, $E_{T_c} = 0$)', if we had many different thresholds, then the number of poles would depend on the number of different such thresholds we have, N_{E_{T_c} \neq E_{T_{c'}}}.

In the simple case of Multi-Level Breit-Wigner approximation (\ref{eq:inv_A expression MLBW}), the diagonal level matrix $\boldsymbol{A^{-1}}_{\mathrm{MLBW}}$ greatly simplifies the radioactive states eigenproblem (\ref{eq::invA det roots z-space}): it is now diagonal and the poles $\left\{ p_j \right\}$ are the roots of:
\begin{equation}
\begin{IEEEeqnarraybox}[][l]{rcl}
E_\lambda - p_j^2 - \sum_{c=1}^{N_c} \gamma_{\lambda,c}^2 \left( L_c(p_j) - B_c \right)   & \ =\ & 0
\IEEEstrut\end{IEEEeqnarraybox}
\label{eq: MLBW polefinding}
\end{equation}
We then have $\left[ \boldsymbol{\kappa_j}\boldsymbol{\kappa_j}^\mathsf{T}\right]_{{\mathrm{MLBW}}} = \boldsymbol{\gamma}^\mathsf{T} \left[\boldsymbol{\alpha_j}\boldsymbol{\alpha_j}^\mathsf{T}\right]_{{\mathrm{MLBW}}} \boldsymbol{\gamma}$ where normalization (\ref{eq:A residues normalization z-space}) entails
\begin{equation}
\begin{IEEEeqnarraybox}[][l]{rcl}
\left[\boldsymbol{\alpha_j}\boldsymbol{\alpha_j}^\mathsf{T}\right]_{{\mathrm{MLBW}}}  & \ =\ & \boldsymbol{\mathrm{diag}}_{N_\lambda}\left( \frac{-1}{2p_j + \sum_{c=1}^{N_c} \gamma_{\lambda,c}^2 \frac{\partial L_c}{\partial z }(p_j)}\right)
\IEEEstrut\end{IEEEeqnarraybox}
\label{eq: MLBW normalization}
\end{equation}
This approach will yield the same results as those in \cite{Jammes_1999}.

Interestingly, besides adding the spurious poles $\big\{\omega_n\big\}$ of the potential cross section $\sigma_c^{\text{pot}}$ (eq. (1) and (2) of \cite{Hwang_1998} or eq. (3) and (4) of \cite{Hwang_2003}), Hwang also accounted for too many $\big\{p_j\big\}$ poles, in eq. (35a) section III.A, p.197 of \cite{Hwang_1987}. This is for two fundamental reasons: 1) lemma 3 of \cite{Ducru_shadow_Brune_Poles_2019} on diagonal divisibility and capped multiplicities means Hwang's $q_\ell(\sqrt{E})$ functions can be taken out of his product in equation (36) of \cite{Hwang_1987}; 2) these same $q_\ell(\sqrt{E})$ functions are not the same as our $q_\ell(\rho)$ functions, which are the denominator of $L_\ell(\rho)$. Instead, Hwang's $q_\ell(\sqrt{E})$ functions are the denominator of the penetration $P_\ell(\rho)$ and shift $S_\ell(\rho)$ functions -- defined as $L_\ell(\rho) = S_\ell(\rho) + \mathrm{i} P_\ell(\rho)$ in (29) of \cite{Ducru_shadow_Brune_Poles_2019} where a thorough and in-depth study of these functions is undertaken -- and these denominators are different from the denominator of $L_\ell(\rho)$, as we show in table III of \cite{Ducru_shadow_Brune_Poles_2019}. In essence, this is because by writing $L_\ell(\rho) = S_\ell(\rho) + \mathrm{i}P_\ell(\rho)$, the denominator is brought to its squared modulus: $L_\ell(\rho) = \frac{u_{\ell + 1}(\rho)\cdot q_\ell^*(\rho)}{q_\ell(\rho)\cdot q_\ell^*(\rho)}$, which is no longer its irreducible form, and which therefore doubles the number of $L_\ell(\rho)$ poles by introducing superfluous complex conjugate poles from $ q_\ell^*(\rho)$. These superfluous poles have always been overlooked until now, recent examples are eq. (9) and (10) of \cite{Jammes_1999}, or in eq. (2.29) p.75 of \cite{freimanMultipoleMethodOnthefly2020}, where they count them to find $N_L = 2 N_\lambda  + 2 \sum_{c=1}^{N_c} \ell_c$, which is actually the number of alternative analytic poles $N_S$ we establish in eq. (49) theorem 3 of \cite{Ducru_shadow_Brune_Poles_2019}, instead of the correct number (\ref{eq::NL number of poles no threshold neutrons}) of radioactive poles $N_L = 2 N_\lambda  + \sum_{c=1}^{N_c} \ell_c $ we demonstrated in theorem 1 of \cite{Ducru_Scattering_Matrix_of_Complex_Wavenumbers_2019}. 

Because polynomial root-finding is no simple endeavor -- see \cite{Ducru_PHYSOR_conversion_2016} or \cite{freimanMultipoleMethodOnthefly2020} for methods applied to the radioactive problem (\ref{eq: no threshold neutrons radioactive problem}) and see \cite{Wilkinson_1959,  Durand_1960, Aberth_1973, Durand-Kerner_1989, Schonhage_1982, X_Gourdon_these_1996, panOptimalNearlyOptimal1996a, panSolvingPolynomialEquation1997, malajovichFastStableAlgorithm1997, Pan_approximating_complex_polynomial_zeros_2000,   emirisImprovedAlgorithmsComputing2005, panNewProgressReal2011, mcnameeEfficientPolynomialRootrefiners2012, panAcceleratedApproximationComplex2017} for more general methods --
searching for the wrong number of poles (in particular too many) can have dire numerical consequences. For instance, Hwang explains in \cite{Hwang_1987} how he had to go to quadruple precision in this code ``WHOPPER''. He was finding the poles one by one using a Newton-Raphson method, and then removing them to search for the next pole. But once he had eliminated all the true poles, he was still searching for additional ones which did not actually exist. Numerically, though, one can never fully cancel out a pole, and thus will always find fictitious poles in the immediate vicinity of the cancelled ones. This is exactly what happened to Hwang, and why he had many spurious poles clustered around the non-resonant $N_\lambda$ ones.
Hence knowing the correct number $N_L$ of poles (\ref{eq::NL number of poles no threshold neutrons}) -- and more generally (\ref{eq::NL number of poles}) -- is crucial in practice.

%%%%%%%%%%%%%%%%%%%%%%%%%%%%%%%%%%%%%%%%%%%%%%%%%%%%%%%%%%%%%%%%%%%%%%%%%%%%%%%%
\subsubsection{\label{subsubsec:Exact multipole representations}Exact multipole representations}
%%%%%%%%%%%%%%%%%%%%%%%%%%%%%%%%%%%%%%%%%%%%%%%%%%%%%%%%%%%%%%%%%%%%%%%%%%%%%%%%

Hwang also spent a lot of subsequent work performing a pole expansion of the potential cross section $\sigma_c^{\text{pot}}$ as well as of the energy dependence he found in his scattering residues, in eq. (1) and (2) of \cite{Hwang_1998} or eq. (3) and (4) of \cite{Hwang_2003}.
We recall that though the potential cross section does have poles -- roots $\big\{\omega_n\big\}$ of the $H_+(\rho)$ function here reported in table \ref{tab::L_values_neutral} or expressed by radicals in table II of \cite{Ducru_shadow_Brune_Poles_2019} -- these poles actually have zero residues in the scattering matrix, and thus cancel out of the partial and total cross sections, as we prove in theorem 3 of \cite{Ducru_Scattering_Matrix_of_Complex_Wavenumbers_2019}.
It will thus suffice to write that for the case of neutron cross sections with zero threshold, (\ref{eq: |D|^2 to Re[iD]}) and (\ref{eq: potential cross section}) entail the potential cross section takes the form:
\begin{equation}
\begin{IEEEeqnarraybox}[][l]{rcl}
\sigma_c^{\text{pot}}(z) & \ = \ & 4\pi a_c^2 \frac{g_{J^\pi_c}}{{\rho_0}_c^2} \frac{\sin^2\phi_c(z)}{z^2} \\
& \ = \ & 4\pi a_c^2 \frac{g_{J^\pi_c}}{{\rho_0}_c^2} \frac{1}{z^2}  \Re_{\mathrm{conj}}\left[\frac{1 - \mathrm{e}^{-2\mathrm{i}\phi_c(z)}}{2}\right] 
\IEEEstrut\end{IEEEeqnarraybox}
\label{eq: potential cross section neutron no threshold}
\end{equation}

With all this in mind, we can now finish the explicit Windowed Multipole Representation of no-threshold neutral particles cross sections.
Upon finding the $N_L$ roots $\big\{p_j\big\}$ of the polynomial radioactive problem (\ref{eq: no threshold neutrons radioactive problem}), we can then solve for the nullspace of the inverse level matrix (which we here assume is an eigenline and we refer to \cite{Ducru_Scattering_Matrix_of_Complex_Wavenumbers_2019} for the degenerate cases), and notice that that the degrees of the level matrix components $\boldsymbol{A}$ is at most -2, which leads to the the following, exact, partial fraction decomposition of the level matrix and of the Kapur-Peierls operator:
\begin{equation}
\begin{IEEEeqnarraybox}[][c]{rcl}
\boldsymbol{A}(z) & = & \sum_{j= 1}^{N_L} \frac{\boldsymbol{\alpha_j}\boldsymbol{\alpha_j}^\mathsf{T}}{z - p_j} \\
\boldsymbol{\kappa_j} & \triangleq & \boldsymbol{\gamma}^\mathsf{T} \boldsymbol{\alpha_j} \\
\boldsymbol{R}_L(z) & = & \sum_{j= 1}^{N_L} \frac{\boldsymbol{\kappa_j}\boldsymbol{\kappa_j}^\mathsf{T}}{z - p_j}
\IEEEstrut\end{IEEEeqnarraybox}
\label{eq::A Mittag Leffler in z-space neutron no threshold}
\end{equation}
We can then build a pole expansion of the resonance matrix $\boldsymbol{Z}$ from (\ref{eq:Transmission matrix decomposition Z and D}) by noticing that the $\boldsymbol{\rho}^{1/2}(z) = \sqrt{z} \boldsymbol{\rho_0}^{1/2}$ lead to an additional $z$ term for each residue, and that the degrees of the numerator and denominator of $z$ times the level matrix, $z\boldsymbol{A}(z)$, is still negative (degree of at most -1), guaranteeing the level matrix is a proper rational fraction with simple poles in $z$-space:
\begin{equation}
\begin{IEEEeqnarraybox}[][c]{rcl}
\boldsymbol{\rho}^{\frac{1}{2}}\boldsymbol{R}_L\boldsymbol{\rho}^{\frac{1}{2}}(z)  & = & \sum_{j= 1}^{N_L}\frac{p_j \boldsymbol{\rho_0}^{\frac{1}{2}}  \boldsymbol{\kappa_j}\boldsymbol{\kappa_j}^\mathsf{T} \boldsymbol{\rho_0}^{\frac{1}{2}}}{z - p_j}
\IEEEstrut\end{IEEEeqnarraybox}
\label{eq::TEST}
\end{equation}
This has as consequence the remarkable property that for zero-threshold neutral cross sections, the $z$-space radioactive squared widths $\boldsymbol{\kappa_j}\boldsymbol{\kappa_j}^\mathsf{T}$ (rank-one residues of the Kapur-Peierls operator at poles $p_j$) add up to nullity:
\begin{equation}
\sum_{j= 1}^{N_L}  \boldsymbol{\kappa_j}\boldsymbol{\kappa_j}^\mathsf{T}  =  \boldsymbol{0}
\label{eq::radioactive widths add up to nullity}
\end{equation}
From (\ref{eq:Transmission matrix decomposition Z and D}), we therefore obtain the following expression for the resonance matrix:
\begin{equation}
\boldsymbol{Z}(z) = \boldsymbol{H_+}^{-1}(z) \sum_{j= 1}^{N_L}\frac{p_j \boldsymbol{\rho_0}^{\frac{1}{2}}  \boldsymbol{\kappa_j}\boldsymbol{\kappa_j}^\mathsf{T} \boldsymbol{\rho_0}^{\frac{1}{2}}}{z - p_j} \boldsymbol{H_+}^{-1}(z)  
\label{eq::Z Mittag Leffler in z-space neutron no threshold}
\end{equation}
Where we deliberately left the energy dependence of $\boldsymbol{H_+}^{-1}(z) $, and recall that for neutral particles $H_+(\rho) = O(\rho)$. Polar decomposition (\ref{eq:def H_+ polar decomposition}) entails:${H_+}_{\ell}^{-1}(\rho) = \left| d_\ell^{-1}(\rho) \right| \mathrm{e}^{-\mathrm{i}\phi_\ell(\rho)}  = \mathrm{e}^{-\mathrm{i}\rho} \cdot  d_\ell^{-1}(\rho)  $, which is Hwang's notation in eq (3) of \cite{Hwang_2003}.
A closer look at the last column of table \ref{tab::L_values_neutral} shows that $d_\ell^{-1}(\rho) $ is the rational function of degree zero (that is a proper rational fraction plus a constant) with $\ell$ poles -- the roots $\big\{\omega_n\big\}$ -- that is the square root of that which Hwang identified in  eq. (1) of \cite{Hwang_1998}.
Careful analysis of this functions, using table \ref{tab::L_values_neutral} and conjugate continuation definitions (\ref{eq::|f|2 conjugate continuation}), yields the following expressions:
\begin{equation}
\begin{IEEEeqnarraybox}[][l]{rcl}
\mathrm{e}^{-2\mathrm{i}\phi_c(z)}  & = & \mathrm{e}^{-2\mathrm{i}\rho_c(z)}  \frac{d_{\ell_c}^{-1}(z)}{{d_{\ell_c}^{-1}}^*(z)} = Y_c(\rho_c(z))\\
d_{\ell}(\rho) & = & \mathrm{e}^{\mathrm{i}\ell \frac{\pi}{2}}\frac{\rho^\ell}{\prod_{n=1}^{\ell}\left( \rho - \omega_n\right)}\\
\frac{d_{\ell}^{-1}(\rho)}{{d_{\ell}^{-1}}^*(\rho)} & = &  (-1)^\ell \prod_{n=1}^\ell \left( \frac{\rho - \omega_n^*}{\rho - \omega_n}\right)  \\
 \left| d_{\ell} \right|^{-2}_{\mathrm{conj}}(\rho) & = & \frac{\rho^{2\ell}}{\prod_{n=1}^\ell\left(\rho-\omega_n\right)\left(\rho-\omega_n^*\right)}
\IEEEstrut\end{IEEEeqnarraybox}
\label{eq: d_l and other expressions}
\end{equation}
The diagonal elements of (\ref{eq::Z Mittag Leffler in z-space neutron no threshold}) are therefore exactly:
\begin{equation}
Z_{cc}(z) =   \left| d_{\ell_c} \right|^{-2}_{\mathrm{conj}}(z) \mathrm{e}^{-2\mathrm{i}\phi_c(z)} \sum_{j= 1}^{N_L}\frac{ {\rho_0}_c p_j \left[ \boldsymbol{\kappa_j}\boldsymbol{\kappa_j}^\mathsf{T}\right]_{cc} }{z - p_j} 
\label{eq::Z Mittag Leffler in z-space neutron no threshold diagonal element}
\end{equation}
which, upon partial fraction decomposition, yields the Hwang multipole representation \cite{Hwang_1987}:
\begin{equation}
\begin{IEEEeqnarraybox}[][l]{rcl}
\sigma_c(z) & \ = \ & \sigma_c^{\text{pot}}(z) + \sigma_c^{\text{Hwang}}(z) \\  & & + \frac{1}{z^2}\Re_\mathrm{conj}\left[ -\mathrm{i} \; \mathrm{e}^{-2\mathrm{i}\phi_c(z)}  \sum_{j= 1}^{N_L}\frac{ ^{\text{Hwang}}r_j^{c} }{z - p_j} \right] \\
^{\text{Hwang}}r_j^{c} & \ \triangleq \ & 4\pi a_c^2 \frac{ g_{J^\pi_c} }{{\rho_0}_c^2} \left| d_{\ell_c} \right|^{-2}_{\mathrm{conj}}(p_j)\cdot {\rho_0}_c p_j \left[ \boldsymbol{\kappa_j}\boldsymbol{\kappa_j}^\mathsf{T}\right]_{cc} \\
\sigma_c^{\text{Hwang}}(z) & \triangleq &  4\pi a_c^2 \frac{ g_{J^\pi_c} }{{\rho_0}_c^2} \frac{1}{z^2}\Re\left[ - \mathrm{i} \; \mathrm{e}^{-2\mathrm{i}\phi_c(z)} \Delta(z)  \right] \\
\Delta(z) & \triangleq &  \sum_{n= 1}^{\ell_c} \left[ \frac{ \Delta_n \left[\boldsymbol{\rho}^{\frac{1}{2}}\boldsymbol{R}_L\boldsymbol{\rho}^{\frac{1}{2}}\right]_{cc}\left(\frac{\omega_n}{\rho_0}\right)}{\rho - \omega_n} \right. \\ & & \left. + \frac{ \Delta_n^* \left[\boldsymbol{\rho}^{\frac{1}{2}}\boldsymbol{R}_L\boldsymbol{\rho}^{\frac{1}{2}}\right]_{cc}\left(\frac{\omega_n^*}{\rho_0}\right)}{\rho - \omega_n^*} \right] \\
\Delta_n & \triangleq & \frac{\omega_n^{2\ell_c}}{\prod_{k=1}^{\ell_c}\left( \omega_n - \omega_k^* \right)\prod_{k\neq n}^{\ell_c}\left( \omega_n - \omega_k \right)} 
\IEEEstrut\end{IEEEeqnarraybox}
\label{eq: total cross section WMP Hwang neutron no threshold}
\end{equation}
where the Hwang residues in (\ref{eq: total cross section WMP Hwang neutron no threshold}) are identical to eq. (2) of \cite{Hwang_2003}.
This scripture is the conjugate continuation real part of $N_L + \ell_c $ poles, as identified in \cite{freimanMultipoleMethodOnthefly2020}: the $N_L$ radioactive poles (\ref{eq::NL number of poles no threshold neutrons}), poles of the Kapur-Peierls operator (\ref{eq:R_L radioactive problem in z-space}), plus the roots $\big\{\omega_n\big\}$ of the outgoing wavefunction $\boldsymbol{O}(\boldsymbol{\rho})$. 
However, we have proved the latter cancel out of the transmission matrix, and thus of the cross section (theorem 3 of \cite{Ducru_Scattering_Matrix_of_Complex_Wavenumbers_2019}). Therefore, there must exist a multipole representation with only the $N_L$ Kapur-Peierls poles.
This can be achieved by developing the potential cross section (\ref{eq: potential cross section neutron no threshold}) into poles and residues, factoring the $\mathrm{e}^{-2\mathrm{i}\rho_c}$ component (which has no poles) using expressions (\ref{eq: d_l and other expressions}), and performing a partial fraction decomposition of the rational terms. Upon careful consideration, one will notice that this rational function is of degree zero, that its poles are the radioactive poles (and only those), and that the constant (obtained by evaluating at infinity $|\rho_c| \to \infty$) is $(-1)^{\ell_c}$. This shifts the potential cross section, so that the total cross section (\ref{eq: total cross section}) can be expressed as the sum of a background cross section (with no poles),
\begin{equation}
\begin{IEEEeqnarraybox}[][l]{rcl}
\sigma_c^{\text{back}}(z) & \triangleq & 4\pi a_c^2 \frac{g_{J^\pi_c}}{{\rho_0}_c^2} \frac{\sin^2\left(\rho_c(z) + \ell_c \frac{\pi}{2} \right)}{z^2} 
\IEEEstrut\end{IEEEeqnarraybox}
\label{eq: background cross section neutron no threshold}
\end{equation}
plus a resonant cross section with the $N_L$ radioactive poles: 
\begin{equation}
\begin{IEEEeqnarraybox}[][l]{rcl}
\sigma_c(z) & \ = \ & \sigma_c^{\text{back}}(z) +  \frac{1}{z^2}\Re_\mathrm{conj}\left[ -\mathrm{i} \; \mathrm{e}^{-2\mathrm{i}\rho_c(z) }  \sum_{j= 1}^{N_L}\frac{ ^{\text{tot}}r_j^{c} }{z - p_j} \right] \\
^{\text{tot}}r_j^{c} & \ \triangleq \ & 4\pi a_c^2 \frac{ g_{J^\pi_c} }{{\rho_0}_c^2}  d_{\ell_c}^{-2}(p_j)\cdot {\rho_0}_c p_j \left[ \boldsymbol{\kappa_j}\boldsymbol{\kappa_j}^\mathsf{T}\right]_{cc} \\ & = & 4\pi a_c^2 \frac{ g_{J^\pi_c} }{{\rho_0}_c^2} (-1)^{\ell} \frac{\left({\rho_0}_c p_j\right)^{2\ell + 1} \left[ \boldsymbol{\kappa_j}\boldsymbol{\kappa_j}^\mathsf{T}\right]_{cc}}{\prod_{n=1}^\ell\left(\rho_0p_j - \omega_n\right)^2}
\IEEEstrut\end{IEEEeqnarraybox}
\label{eq: total cross section WMP neutron no threshold}
\end{equation}
To the best of our knowledge, expression (\ref{eq: total cross section WMP neutron no threshold}) is the first time the exact multipole representation of no-threshold neutron cross sections is derived with the proper number of poles. It is exact and complete, in the sense that no window-by-window Laurent expansions are needed. This is only made possible in this specific case of neutron cross sections with zero threshold (no charged particles nor thresholds): though quite restrictive, it is still a case of great practical importance for nuclear reactor physics, as most heavy isotopes are evaluated with only two channels (neutron and fission) with all the other channels being eliminated under the Reich-Moore approximation. This significant difference with light isotopes (in which many more channels are explicitly treated) is partly due to the fact that for heavy isotopes the number of photon channels is large enough that one can average their contribution out, and also because the resonance region starts at lower energies for heavy isotopes, with many resonances before the first non-zero threshold. 

Note that the advantage of not needing local Laurent developments in (\ref{eq: total cross section WMP neutron no threshold}) comes at the computational cost of having to sum all the radioactive poles for each energy call, instead of accounting for the contributions of far-away poles in the Laurent expansion of each window -- in this sense, the windowing process is a form of local compression algorithm for improved efficiency \cite{Hwang_1992, Forget_2013}.

\begin{figure}[] % replace 't' with 'b' to force it to be on the bottom
  \centering
  \subfigure[\ First p-wave resonance.]{\includegraphics[width=0.49\textwidth]{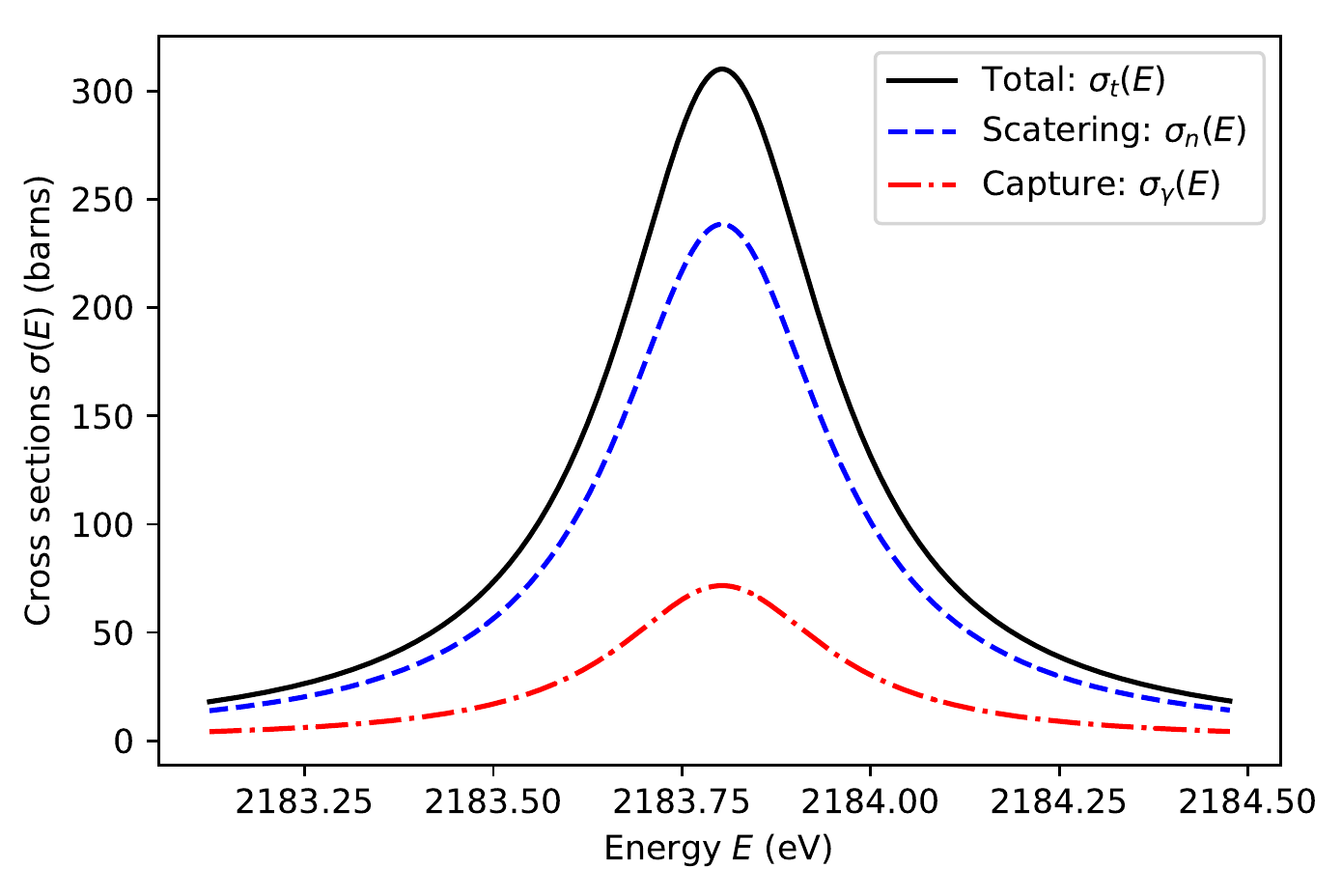}}
  \subfigure[\ Second p-wave resonance.]{\includegraphics[width=0.49\textwidth]{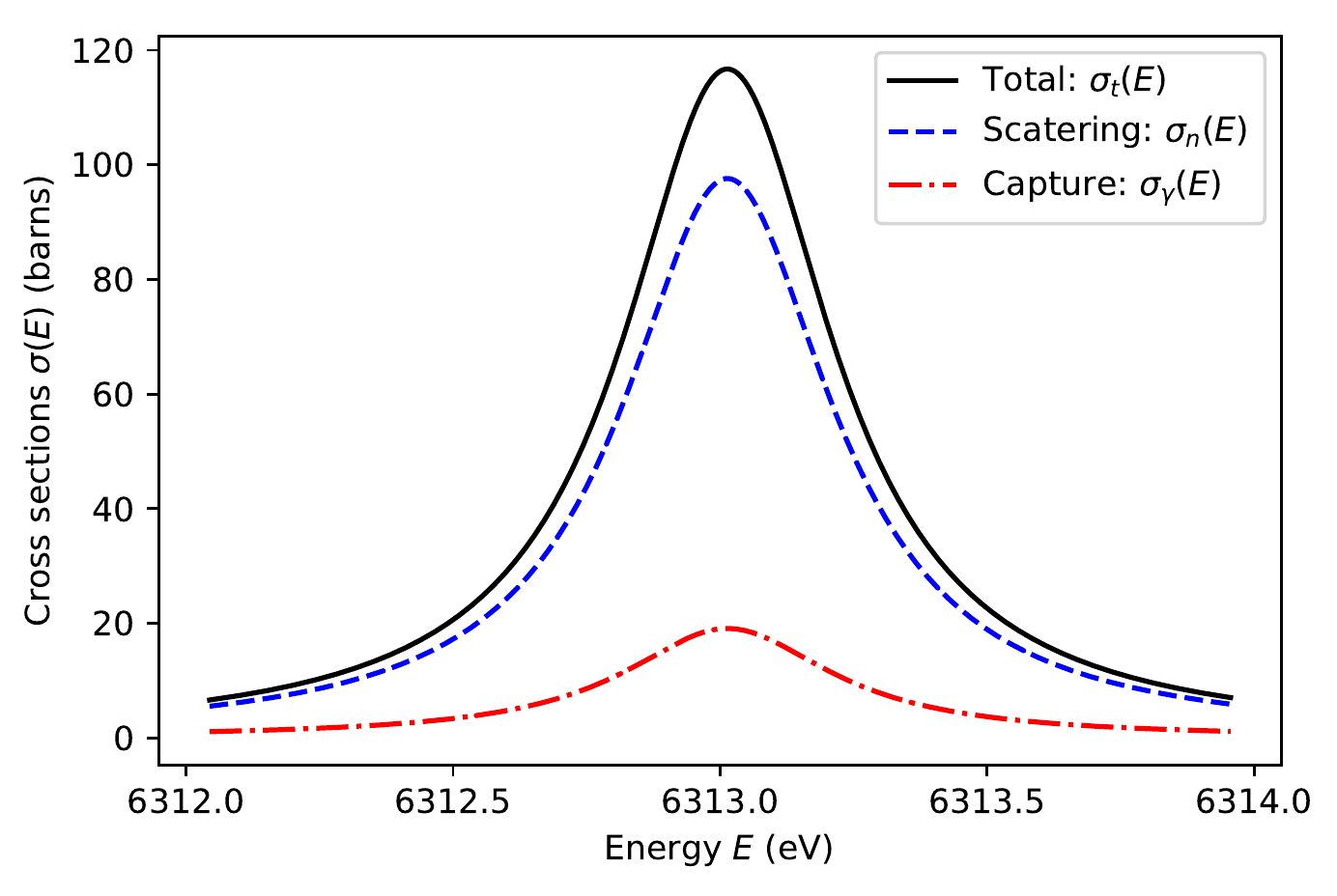}}
  \caption{\small{Xenon $^{\mathrm{134}}\mathrm{Xe}$ Reich-Moore cross sections for spin-parity group $J^\pi = 1/2^{(-)}$ p-wave resonances: the cross sections are generated using the multipole parameters from table \ref{tab:Xe-134 windowed multipole representation} in the multipole representation total cross section (\ref{eq: total cross section WMP neutron no threshold}), as well as the reaction cross section (\ref{eq: reaction cross section WMP Hwang neutron no threshold}) and interference one (\ref{eq: interference cross section WMP Hwang neutron no threshold}) to compute the scattering cross section as (\ref{eq: partial cross section breakdown}), while the capture cross section is the difference between the total and the scattering. All cross sections are identical to those computed using the Reich-Moore approximation R-matrix equations with the ENDF/B-VIII.0 resonance parameters. %(a MLBW evaluation).
  }}
  \label{fig:xenon-134 J=1/2(-) cross section}
\end{figure}

To compute the partial cross sections (\ref{eq: partial cross section breakdown}), we can calculate the reaction cross section (\ref{eq: reaction cross section}) and the interference one (\ref{eq: interference cross section}). 
For the reaction cross section, we use the square modulus conjugate continuation (\ref{eq::|f|2 conjugate continuation}), and notice that $\left|\boldsymbol{H_+}^{-1}\right|^{-2}_\mathrm{conj}(z) = \left|\boldsymbol{d}\right|^{-2}_\mathrm{conj}(z) \triangleq \boldsymbol{\mathrm{diag}}\left(d_{\ell_c}(\rho)\cdot d_{\ell_c}(\rho^*)^*\right)^{-2}$ is now a rational function (the $\mathrm{e}^{-\mathrm{i}\rho}$ terms cancel out). Therefore, evaluating at the pole values yields the partial fraction decomposition of the square modulus of the resonance matrix:
\begin{equation}
\left|\boldsymbol{Z}\right|_\mathrm{conj}^2(z) = \Re_\mathrm{conj}\left[ \sum_{j= 1}^{N_L} \frac{\boldsymbol{\aleph_j}}{z - p_j}  \right]
\label{eq:: |Z|^2 Mittag Leffler in z-space neutron no threshold}
\end{equation}
where the residues $\boldsymbol{\aleph_j}$ are explicitly constructed as
\begin{equation}
 \boldsymbol{\aleph_j} \mkern-2mu \triangleq 2 \mkern-2mu \left|\boldsymbol{d}\right|^{-2}_\mathrm{conj}\mkern-4mu (p_j)   \boldsymbol{\rho_0} \mkern-4mu \left(  p_j^2 \boldsymbol{\kappa_j}\boldsymbol{\kappa_j}^\mathsf{T} \mkern-3mu \circ \mkern-3mu  \Big[  \boldsymbol{R}_L(p_j^*) \Big]^* \right) \mkern-4mu \boldsymbol{\rho_0}   \left|\boldsymbol{d}\right|^{-2}_\mathrm{conj}\mkern-4mu(p_j)
\label{eq:: |Z|^2 Mittag Leffler in z-space neutron no threshold residues Alpeh}
\end{equation}
In summary, the energy dependence of the residues in (\ref{eq::Z Mittag Leffler in z-space neutron no threshold}) cancels out of the reaction residues, hence the reaction cross section (from channel $c$ to $c'$) is exactly 
\begin{equation}
\begin{IEEEeqnarraybox}[][l]{rcl}
\sigma_{cc'}^{\text{react}}(z)  & \ =  \ &  \frac{1}{z^2}  \Re_{\mathrm{conj}}\left[ \sum_{j\geq 1} \frac{^{\text{react}}r_j^{cc'} }{z - p_j}  \right] 
\IEEEstrut\end{IEEEeqnarraybox}
\label{eq: reaction cross section WMP Hwang neutron no threshold}
\end{equation}
where the residues can either simply be evaluated as (\ref{eq: reaction residues WMP}) or constructed as:
\begin{equation}
\begin{IEEEeqnarraybox}[][c]{rcl}
^{\text{react}}r_j^{cc'} & \ \triangleq \ &  \frac{4\pi a_c^2  }{{\rho_0}_c^2}g_{J^\pi_c}\left[ \boldsymbol{\aleph_j}\right]_{cc'}
\IEEEstrut\end{IEEEeqnarraybox}
\label{eq:reaction cross section WMP Residues z-space neutron no threshold}
\end{equation}
For the interference cross section (\ref{eq: interference cross section}), we notice using expressions (\ref{eq: d_l and other expressions}) that the phase behavior also cancels out of $\Re\left[-\mathrm{i}Z_{cc} Y_c^*\right]$, so that plugging the resonance matrix partial fraction decomposition (\ref{eq::Z Mittag Leffler in z-space neutron no threshold diagonal element}) into interference cross section expression (\ref{eq: interference cross section}) yields rational fraction
\begin{equation}
\begin{IEEEeqnarraybox}[][l]{rcl}
\sigma_c^{\text{int}}(z) & \ = \ &  \frac{1}{z^2}\Re_\mathrm{conj}\left[\sum_{j= 1}^{N_L}\frac{ ^{\text{int}}r_j^{c} }{z - p_j} \right] 
\IEEEstrut\end{IEEEeqnarraybox}
\label{eq: interference cross section WMP Hwang neutron no threshold}
\end{equation}
where the interference residues can simply be evaluated as (\ref{eq: interference cross section WMP}), or explicitly constructed as
\begin{equation}
\begin{IEEEeqnarraybox}[][l]{rcl}
^{\text{int}}r_j^{c} & \ \triangleq \ & - \mathrm{i} 4\pi a_c^2 \frac{ g_{J^\pi_c} }{{\rho_0}_c^2}  \left| d_{\ell_c} \right|^{-2}_{\mathrm{conj}}(p_j) \cdot {\rho_0}_c p_j \left[ \boldsymbol{\kappa_j}\boldsymbol{\kappa_j}^\mathsf{T}\right]_{cc} \\
 & = & - \mathrm{i} 4\pi a_c^2 \frac{ g_{J^\pi_c} }{{\rho_0}_c^2} \frac{\left({\rho_0}_c p_j \right)^{2\ell_c + 1} \left[ \boldsymbol{\kappa_j}\boldsymbol{\kappa_j}^\mathsf{T}\right]_{cc} }{\prod_{n=1}^\ell \left( {\rho_0}_c p_j - \omega_n \right)\left( {\rho_0}_c p_j - \omega_n^* \right)}
\IEEEstrut\end{IEEEeqnarraybox}
\label{eq: interference residues WMP Hwang neutron no threshold}
\end{equation}
Having explicitly constructed the total, potential, reaction, and interference cross sections, we can thus calculate the partial cross sections explicitly through (\ref{eq: partial cross section breakdown}).

%%%%%%%%%%%%%%%%%%%%%%%%%%%%%%%%%%%%%%%%%%%%%%%%%%%%%%%%%%%%%%%%%%%%%%%%%%%%%%%%
\subsubsection{\label{subsubsec:Evidence for exact multipole representation in xenon 134} Evidence for exact multipole representation in $^{\mathrm{134}}\mathrm{Xe}$}
%%%%%%%%%%%%%%%%%%%%%%%%%%%%%%%%%%%%%%%%%%%%%%%%%%%%%%%%%%%%%%%%%%%%%%%%%%%%%%%%

\begin{table*}[ht!!]
\caption{\label{tab:Xe-134 windowed multipole representation} 
Windowed multipole parameters of the two p-wave resonances of $^{\mathrm{134}}\mathrm{Xe}$, spin-parity group $J^\pi = 1/2^{(-)}$, converted from ENDF/B-VIII.0 evaluation (MLBW) to multipole representation using Reich-Moore level matrix (\ref{eq:inv_A expression Reich-Moore}).}
\begin{ruledtabular}
\begin{tabular}{l}
$ z = \sqrt{E}$ with $E$ in (eV)  \tabularnewline
$A = 132.7600$ \tabularnewline
$a_c = 5.80$ : channel radius (Fermis) \tabularnewline
$\rho_0 = \frac{A a_c \sqrt{\frac{2m_n}{h}} }{A + 1} $ in ($\mathrm{\sqrt{eV}}^{-1}$), so that $\rho(z) \triangleq \rho_0 z$ \tabularnewline
with $\sqrt{\frac{2m_n}{h}} = 0.002196807122623 $ in units ($1/(10^{-14}\text{m} \mathrm{\sqrt{eV}})$) \tabularnewline
\hline \tabularnewline
\textbf{Multipole parameters} (rounded to 5 digits): converted from R-matrix parameters using Reich-Moore equations. \tabularnewline
\hline 
\begin{tabular}{l|l|l|l|l}
 Radioactive poles $p_j$  & total residues $^{\text{tot}}r_j^{c} $ & reaction residues $^{\text{react}}r_j^{cc'}$ & interference residues $ ^{\text{int}}r_j^{c}$ & Hwang residues $^{\text{Hwang}}r_j^{c}$ \tabularnewline
   ($\mathrm{\sqrt{eV}}$) from (\ref{eq: no threshold neutrons radioactive problem}) &   ($\mathrm{barns\sqrt{eV}^3}$) from (\ref{eq: total cross section WMP neutron no threshold}) &  ($\mathrm{barns\sqrt{eV}^3}$) from (\ref{eq:reaction cross section WMP Residues z-space neutron no threshold}) & ($\mathrm{barns\sqrt{eV}^3}$) from (\ref{eq: interference residues WMP Hwang neutron no threshold}) &  ($\mathrm{barns\sqrt{eV}^3}$) from (\ref{eq: total cross section WMP Hwang neutron no threshold})\tabularnewline
\hline \hline
$6.4652\mkern+2mu\mathrm{E}\mkern-5mu-\mkern-5mu8$ 
&   $6.9766\mkern+2mu\mathrm{E}\mkern-5mu + \mkern-5mu 8 $ 
&  $ 2.8519 \mkern+2mu\mathrm{E}\mkern-5mu - \mkern-5mu 2$
&  $ -2.8446 \mkern+2mu\mathrm{E}\mkern-5mu - \mkern-5mu 2$
& $-4.6048 \mkern+2mu\mathrm{E}\mkern-5mu+\mkern-5mu5 $
\tabularnewline 
$-\mathrm{i} 7.9179\mkern+2mu\mathrm{E}\mkern-5mu+\mkern-5mu2 $ 
&  $-\mathrm{i} 5.5825\mkern+2mu\mathrm{E}\mkern-5mu - \mkern-5mu 2 $
& $+\mathrm{i} 4.6048 \mkern+2mu\mathrm{E}\mkern-5mu + \mkern-5mu 5 $
& $ +\mathrm{i} 4.6048  \mkern+2mu\mathrm{E}\mkern-5mu + \mkern-5mu 5 $
& $ -\mathrm{i} 2.8446\mkern+2mu\mathrm{E}\mkern-5mu - \mkern-5mu 2 $
\tabularnewline \hline
 $-4.6731 {\mkern+2mu\mathrm{E}\mkern-5mu + \mkern-5mu 1} $ 
 & $-1.2144 \mkern+2mu\mathrm{E}\mkern-5mu + \mkern-5mu 3 $
 &  $ -1.5693 \mkern+2mu\mathrm{E}\mkern-5mu - \mkern-5mu 1$
 &  $ -1.3518 \mkern+2mu\mathrm{E}\mkern-5mu - \mkern-5mu 1$
 & $-1.2229 \mkern+2mu\mathrm{E}\mkern-5mu + \mkern-5mu 3 $
 \tabularnewline 
 $ -\mathrm{i} 9.7105 {\mkern+2mu\mathrm{E}\mkern-5mu  -   \mkern-5mu 4 }$ 
 & $ +\mathrm{i} 1.4390  \mkern+2mu\mathrm{E}\mkern-5mu + \mkern-5mu 2 $
 & $ +\mathrm{i} 1.7479  \mkern+2mu\mathrm{E}\mkern-5mu + \mkern-5mu 3 $
 & $ +\mathrm{i} 1.2229  \mkern+2mu\mathrm{E}\mkern-5mu + \mkern-5mu 3 $
 & $ -\mathrm{i} 1.3518\mkern+2mu\mathrm{E}\mkern-5mu - \mkern-5mu 1 $
 \tabularnewline \hline
  $4.6731 {\mkern+2mu\mathrm{E}\mkern-5mu + \mkern-5mu 1}$ 
  &  $ -1.2144 \mkern+2mu\mathrm{E}\mkern-5mu + \mkern-5mu 3 $
  &  $ 1.5693 \mkern+2mu\mathrm{E}\mkern-5mu - \mkern-5mu 1$
  &  $ 1.7868 \mkern+2mu\mathrm{E}\mkern-5mu - \mkern-5mu 1$
  &  $-1.2229 \mkern+2mu\mathrm{E}\mkern-5mu + \mkern-5mu 3 $
  \tabularnewline 
    $-\mathrm{i} 1.8048 {\mkern+2mu\mathrm{E}\mkern-5mu  -   \mkern-5mu 3 }$
    & $-\mathrm{i} 1.4386 \mkern+2mu\mathrm{E}\mkern-5mu + \mkern-5mu 2 $
    & $ +\mathrm{i} 9.4043 \mkern+2mu\mathrm{E}\mkern-5mu + \mkern-5mu 2 $
    & $ +\mathrm{i} 1.2229 \mkern+2mu\mathrm{E}\mkern-5mu + \mkern-5mu 3 $
    & $ +\mathrm{i} 1.7868\mkern+2mu\mathrm{E}\mkern-5mu - \mkern-5mu 1 $
    \tabularnewline \hline
 $-7.9454 {\mkern+2mu\mathrm{E}\mkern-5mu + \mkern-5mu 1} $
 &   $ -1.0827 \mkern+2mu\mathrm{E}\mkern-5mu + \mkern-5mu 3 $
 &   $ -4.2538 \mkern+2mu\mathrm{E}\mkern-5mu - \mkern-5mu 1 $
 &   $ -4.1864 \mkern+2mu\mathrm{E}\mkern-5mu - \mkern-5mu 1$
 &  $-1.1047 \mkern+2mu\mathrm{E}\mkern-5mu + \mkern-5mu 3 $
 \tabularnewline 
  $ -\mathrm{i} 1.0084 {\mkern+2mu\mathrm{E}\mkern-5mu  -   \mkern-5mu 3 }$
 &  $ +\mathrm{i} 2.1937 \mkern+2mu\mathrm{E}\mkern-5mu + \mkern-5mu 2 $
 & $ +\mathrm{i} 1.3735 \mkern+2mu\mathrm{E}\mkern-5mu + \mkern-5mu 3 $
 & $ +\mathrm{i} 1.1047 \mkern+2mu\mathrm{E}\mkern-5mu + \mkern-5mu 3 $
 & $ -\mathrm{i} 4.1864 \mkern+2mu\mathrm{E}\mkern-5mu - \mkern-5mu 1 $
 \tabularnewline \hline
  $7.9454 {\mkern+2mu\mathrm{E}\mkern-5mu + \mkern-5mu 1}$
  & $ -1.0827 \mkern+2mu\mathrm{E}\mkern-5mu + \mkern-5mu 3  $
  & $ 4.2538 \mkern+2mu\mathrm{E}\mkern-5mu - \mkern-5mu 1 $
  & $ 4.3211\mkern+2mu\mathrm{E}\mkern-5mu - \mkern-5mu 1 $
  & $-1.1047\mkern+2mu\mathrm{E}\mkern-5mu + \mkern-5mu 3 $
  \tabularnewline 
    $- \mathrm{i} 1.4991 {\mkern+2mu\mathrm{E}\mkern-5mu  -   \mkern-5mu 3 }$  
    & $ -\mathrm{i} 2.1936 \mkern+2mu\mathrm{E}\mkern-5mu + \mkern-5mu 2 $
    & $ +\mathrm{i} 9.2389 \mkern+2mu\mathrm{E}\mkern-5mu + \mkern-5mu 2 $
    & $ +\mathrm{i} 1.1047\mkern+2mu\mathrm{E}\mkern-5mu + \mkern-5mu 3 $
    & $ +\mathrm{i} 4.3211 \mkern+2mu\mathrm{E}\mkern-5mu - \mkern-5mu 1 $
\end{tabular} \\
\hline\\
\textbf{R-matrix parameters}: reference ENDF/B-VIII.0 evaluation (MLBW) used with Reich-Moore level matrix (\ref{eq:inv_A expression Reich-Moore}). \\
\hline \\
$E_1 = 2186.0$ : first resonance energy (eV) \\
$\Gamma_{1,n} = 0.2600$ : neutron width of first resonance \\ (not reduced width), i.e. $\Gamma_{\lambda,c} = 2P_c(E_\lambda) \gamma_{\lambda,c}^2$ \\
$\Gamma_{1,\gamma} = 0.0780$ : eliminated capture width (eV) \\
$E_2 = 6315.0$ : second resonance energy (eV) \\
$\Gamma_{2,n} = 0.4000$ (eV) \\
$\Gamma_{2,\gamma} = 0.0780$ (eV) \\
$g_{J^\pi} = 1/3$ : spin statistical factor \\
$B_c = - 1 $ \\
\end{tabular}
\end{ruledtabular}
\end{table*}

We discovered shadow alternative poles of $^{\mathrm{134}}\mathrm{Xe}$ spin-parity group $J^\pi = 1/2^{(-)}$ two p-wave resonances in \cite{Ducru_shadow_Brune_Poles_2019}, and found the radioactive parameters poles and residues in \cite{Ducru_Scattering_Matrix_of_Complex_Wavenumbers_2019}. We now complete this xenon trilogy by here providing the exact multipole representation of $^{\mathrm{134}}\mathrm{Xe}$ spin-parity group $J^\pi = 1/2^{(-)}$ cross section. The multipole parameters are documented in table \ref{tab:Xe-134 windowed multipole representation}, and the corresponding cross sections are plotted in figure \ref{fig:xenon-134 J=1/2(-) cross section}.
In ENDF/B-VIII.0, $^{\mathrm{134}}\mathrm{Xe}$ is a MLBW evaluation with only one explicit (neutron) channel, all other channels are eliminated using Wigner-Teichmann and Reich-Moore approximations. One can thus compute the total cross section using mutlipole representation (\ref{eq: total cross section WMP neutron no threshold}), and the scattering cross section as the partial cross section for $\sigma_{nn}(E)$ from (\ref{eq: partial cross section breakdown}), using reaction cross section (\ref{eq: reaction cross section WMP Hwang neutron no threshold}) and interference one (\ref{eq: interference cross section WMP Hwang neutron no threshold}). The capture cross section is then computed as the difference between the total and the scattering cross section.
The p-waves ($\ell_c = 1$) entail there are $N_L = 5$ radioactive poles -- validating (\ref{eq::NL number of poles no threshold neutrons}) -- and the corresponding residues are documented in table \ref{tab:Xe-134 windowed multipole representation}.
As we see in this xenon example, the multipole representation is an exact alternative formalism to compute R-matrix cross sections. Nonetheless, if we want to treat charged particle channels and thresholds, we need to use local Laurent developments in energy windows, which makes the windowed multipole representation cumbersome and somewhat unsuited for standard nuclear data libraries.

%%%%%%%%%%%%%%%%%%%%%%%%%%%%%%%%%%%%%%%%%%%%%%%%%%%%%%%%%%%%%%%%%%%%%%%%%%%%%%%%
\subsubsection{\label{subsubsec:Exact to windowed multipole representations}Exact to windowed multipole representations}
%%%%%%%%%%%%%%%%%%%%%%%%%%%%%%%%%%%%%%%%%%%%%%%%%%%%%%%%%%%%%%%%%%%%%%%%%%%%%%%%

Note that we can convert the exact multipole total cross section expression (\ref{eq: total cross section WMP neutron no threshold}) -- which has energy-dependent residues due to $\mathrm{e}^{-2\mathrm{i}\rho_c(z)}$ -- into the general windowed multipole representation (\ref{eq:: sigma(z) Windowed Multipole representation}), with energy-independent residues plus a Laurent expansion of order no less than $-2$.
It suffices to evaluate the residues at the pole values, and note that the Laurent expansion part $\underset{\mathrm{tot}}{\mathrm{Laur}}(z)$ is then the difference of the two remaining components, that is:
\begin{equation}
\begin{IEEEeqnarraybox}[][l]{rcl}
\underset{\mathrm{tot}}{\mathrm{Laur}}(z) & \ = \ &  \frac{1}{z^2}\Re_\mathrm{conj}\mkern-6mu\left[\sum_{j= 1}^{N_L}\frac{-\mathrm{i} \;  ^{\text{tot}}r_j^{c} }{z - p_j}  \left( \mathrm{e}^{-2\mathrm{i}{\rho_0}_c z}  - \mathrm{e}^{-2\mathrm{i}{\rho_0}_c p_j} \right) \right] \\
& \ = \ &  \Re_\mathrm{conj}\mkern-6mu\left[\sum_{j= 1}^{N_L} \frac{ ^{\text{tot}}r_j^{c}}{\mathrm{i}z^2}  \left( \sum_{n\geq 0} \frac{\left(-2\mathrm{i}{\rho_0}_c\right)^n }{n!} \frac{z^n - p_j^n}{z-p_j}\right)  \right] \\
& \ = \ & \sum_{n\geq 1} \sum_{m=1}^n \frac{z^{m}}{z^3} \Re_\mathrm{conj}\mkern-6mu\left[\sum_{j= 1}^{N_L} \frac{ ^{\text{tot}}r_j^{c}}{\mathrm{i}}  p_j^{n-m}\frac{\left(-2\mathrm{i}{\rho_0}_c\right)^n }{n!} \right] \\
& \ = \ & \sum_{n\geq -2} a_n z^n
\IEEEstrut\end{IEEEeqnarraybox}
\label{eq: Hwang's Holomorphic part WMP total cross section neutron no threshold}
\end{equation}
so that the total cross section (\ref{eq: total cross section WMP Hwang neutron no threshold}) can be expanded as
\begin{equation}
\begin{IEEEeqnarraybox}[][l]{rcl}
\sigma_c(z) & \ = \ & \sigma_c^{\text{back}}(z) +  \frac{1}{z^2}\Re_\mathrm{conj}\left[  \sum_{j= 1}^{N_L}\frac{ -\mathrm{i} \; \mathrm{e}^{-2\mathrm{i}{\rho_0}_c p_j }  \; ^{\text{tot}}r_j^{c} }{z - p_j} \right]  \\ & & + \underset{\mathrm{tot}}{\mathrm{Laur}}(z)
\IEEEstrut\end{IEEEeqnarraybox}
\label{eq: total cross section WMP Hwang neutron no threshold Laurent Expansion energy-independent residues}
\end{equation}
where the residues are now independent of energy. 
By further performing the analytic expansion of the background cross section (\ref{eq: background cross section neutron no threshold})
\begin{equation}
\begin{IEEEeqnarraybox}[][l]{rcl}
\sigma_c^{\text{back}}(z) & \triangleq & 4\pi a_c^2 \frac{g_{J^\pi_c}}{{\rho_0}_c^2}\mkern-3mu \left[ \mkern-3mu\frac{1 \mkern-3mu - \mkern-3mu (-1)^{\ell}}{2 \; z^2} \mkern-3mu + \mkern-3mu \sum_{n\geq1} \mkern-3mu \frac{(-1)^n}{2}\frac{(2{\rho_0}_c)^{2n}}{ (2n)!} \frac{z^{2n}}{z^2}\right]
\IEEEstrut\end{IEEEeqnarraybox}
\label{eq: background cross section neutron no threshold analytic expansion}
\end{equation}
one recovers the general windowed multipole representation (\ref{eq:: sigma(z) Windowed Multipole representation}).

As we see, in this special case of neutron channels without threshold, we can explicitly construct the full windowed multipole representation (\ref{eq: total cross section WMP Hwang neutron no threshold Laurent Expansion energy-independent residues}) without the need of local expansions for each energy window $\mathcal{W}(E)$. Somewhat ironically, it is also much more cumbersome to explicitly construct both the Laurent expansion and the residues, compared to the more general approaches of theorem \ref{theo::WMP Representation}.  
Alternatively, one can solve for the radioactive poles $\big\{ p_j \big\}$ through polynomial root-finding (\ref{eq: no threshold neutrons radioactive problem}), and then revert to the general methods of theorem \ref{theo::WMP Representation} to compute the corresponding residues, after what the Laurent expansions can be locally curve-fitted.

%%%%%%%%%%%%%%%%%%%%%%%%%%%%%%%%%%%%%%%%%%%%%%%%%%%%%%%%%%%%%%%%%%%%%%%%%%%%%%%%
%*******************************************************************************
%%%%%%%%%%%%%%%%%%%%%%%%%%%%%%%%%%%%%%%%%%%%%%%%%%%%%%%%%%%%%%%%%%%%%%%%%%%%%%%%
\section{\label{sec::Windowed Multipole Covariance} Windowed Multipole Covariances }
%%%%%%%%%%%%%%%%%%%%%%%%%%%%%%%%%%%%%%%%%%%%%%%%%%%%%%%%%%%%%%%%%%%%%%%%%%%%%%%%
%*******************************************************************************
%%%%%%%%%%%%%%%%%%%%%%%%%%%%%%%%%%%%%%%%%%%%%%%%%%%%%%%%%%%%%%%%%%%%%%%%%%%%%%%%

In section \ref{sec:From R-matrix to Windowed Multipole}, we established the windowed multipole representation as a general alternative way to parametrize the energy dependence of R-matrix cross sections (theorem \ref{theo::WMP Representation}).
In this section, we consider how the Windowed Multipole Representation can account for R-matrix cross section epistemic uncertainties.
Such uncertainties exist because nuclear cross sections are known from experiments, and experimental measurements always come with error-bars.
Therefore, in addition to evaluating R-matrix parameter values, evaluators add resonance parameters covariance matrices to standard nuclear data libraries (File 32 in the ENDF/B-VIII.0 library \cite{ENDFBVIII8th2018brown}), aimed at reproducing the empirical uncertainty observed in nuclear cross sections.

%%%%%%%%%%%%%%%%%%%%%%%%%%%%%%%%%%%%%%%%%%%%%%%%%%%%%%%%%%%%%%%%%%%%%%%%%%%%%%%%
\subsection{\label{subsec::Converting R-matrix parameters covariances} Converting R-matrix parameters covariances}
%%%%%%%%%%%%%%%%%%%%%%%%%%%%%%%%%%%%%%%%%%%%%%%%%%%%%%%%%%%%%%%%%%%%%%%%%%%%%%%%

If it exists, the covariance matrix $\mathbb{V}\mathrm{ar}\left( X \right)$ of a random vector $X$ with expectation value $\mathbb{E}\left[ X \right]$ is a defined as:
\begin{equation}
    \mathbb{V}\mathrm{ar}\left( X \right) \triangleq \mathbb{E}\left[ X X^\dagger \right] - \mathbb{E}\left[ X\right]\mathbb{E}\left[ X\right]^\dagger
\end{equation}
We denote $\big\{ \Gamma \big\}$ the set of R-matrix resonance parameters $\big\{ \Gamma \big\} \triangleq \Big\{  E_{\lambda}, \gamma_{\lambda,c} \Big\}$, which are implicitly considered to be the expectation value of the underlying distribution, $ \big\{ \Gamma \big\} \triangleq \mathbb{E}\left[\Gamma \right]$, and $\mathbb{V}\mathrm{ar}\left( \Gamma \right)$ their corresponding covariance matrices.
These represent the resonance parameters epistemic uncertainty, which is accounted for by assuming the parameters are drawn from the multivariate normal distribution: $\mathcal{N}\left( \Gamma, \mathbb{V}\mathrm{ar}\left( \Gamma \right) \right)$.
Recall that both the channel radii $a_c$ and the boundary conditions $B_c$ are arbitrarily set constants, and therefore have no uncertainty. Also, we here do not explicitly treat the uncertainty on threshold energies $E_{T_c}$, but our approach could readily be extended to them.

We consider the unified Windowed Multipole Representation of R-matrix cross sections (\ref{eq:: sigma(z) Windowed Multipole representation}), which we proved is an exact representation of R-matrix cross sections everywhere but for windows containing a non-zero threshold $E_{T_c} \in \mathcal{W}(E)$ -- in these threshold windows, form (\ref{eq:: sigma(z) Windowed Multipole representation}) is only an approximation of exact R-matrix cross sections of theorem \ref{theo::WMP Representation}, yet this approximation (\ref{eq:: sigma(z) Windowed Multipole representation}) can be made to reach any target accuracy. 
For each energy window $\mathcal{W}(E) $, we denote $\big\{ \Pi \big\}$ the windowed multipole parameters -- that is the set of poles $\big\{ p_j \big\}$, residues $\big\{ \widetilde{r}_j^{cc'}, r_j^c \big\}$, and Laurent expansion coefficients $\big\{ a_n \big\}$ that parametrize cross section (\ref{eq:: sigma(z) Windowed Multipole representation}) in that energy window:  $\big\{ \Pi \big\} \triangleq \big\{ p_j , \widetilde{r}_j^{cc'}, r_j^c  , a_n \big\}$. 

The main result of this section -- theorem \ref{theo::WMP covariance} -- establishes a framework to convert R-matrix resonance parameters covariance matrices $\mathbb{V}\mathrm{ar}\left( \Gamma \right)$ into \textit{Windowed Multipole Covariances} $\mathbb{V}\mathrm{ar}\left( \Pi \right)$. It is based on the following lemma \ref{lem::WMP sensitivities}, which derives sensitivities of R-matrix cross sections $\sigma(E)$ to both resonance parameters $\big\{ \Gamma \big\}$ and multipoles $\big\{ \Pi \big\}$. 

\begin{lem}\label{lem::WMP sensitivities}\textsc{Cross sections parameter sensitives}\\ 
Let $z \in \mathbb{C}$ be the complex, analytically continued square-root-of-energy: $z = \sqrt{E}$.
Consider Windowed Multipole cross section (\ref{eq:: sigma(z) Windowed Multipole representation}), i.e. locally of the form:
\begin{equation*}
\begin{IEEEeqnarraybox}[][c]{C}
     \sigma(z)  \underset{\mathcal{W}(E)}{=} \sum_{n\geq -2} a_n z^n + \frac{1}{z^2}\Re_{\mathrm{conj}}\left[\sum_{j \geq 1} \frac{r_{j}}{z-p_j}\right]
%\label{eq:: sigma(z) Windowed Multipole representation}
\IEEEstrut\end{IEEEeqnarraybox}
\end{equation*}
We recall the Cauchy-Poincar\'e-Wirtinger holomorphic complex differential definition for $z = x + \mathrm{i} y $, $x , y \in \mathbb{R}$
 \begin{equation}
      \partial_z \triangleq \frac{1}{2} \left(  \partial_x - \mathrm{i} \partial_y \right)
  \end{equation}
so that $\partial_z z = 1$, and $\partial_z z^* = 0$, where $z^* \triangleq x - \mathrm{i}y$.

The cross section sensitivities to multipoles $\frac{\partial\sigma}{\partial \Pi}(z)$ (i.e. the partial differentials of the cross section with respect to multipoles) are then given, for each window $\mathcal{W}(E)$, by:
\begin{equation}
\begin{IEEEeqnarraybox}[][c]{rClcrcl}
      \frac{\partial \sigma}{\partial p_j }(z) & \mkern-3mu = \mkern-3mu & \frac{1}{z^2} \frac{ \frac{r_{j}}{2} }{\left(z-p_j\right)^2} &  ,  &
      \frac{\partial \sigma}{\partial p_j^* }(z) & = & \frac{1}{z^2} \frac{ \frac{r_{j}^*}{2} }{\left(z-p_j^*\right)^2} \\
 %     \frac{\eth \sigma}{\eth p_j }(z) & = & \frac{1}{z^2} \frac{ r_{j}^* }{\left(z-p_j^*\right)^2} & \quad , \quad &
 %     \frac{\eth \sigma}{\eth p_j^* }(z) & = & \frac{1}{z^2} \frac{ r_{j} }{\left(z-p_j\right)^2} \\
      \frac{\partial \sigma (z)}{\partial \Re\left[p_j\right] } & \mkern-3mu = \mkern-3mu & \frac{1}{z^2} \Re_{\mathrm{conj}}\mkern-6mu\left[ \mkern-3mu \frac{r_{j} }{\left(z-p_j\right)^2} \mkern-3mu \right]  &  ,  &
      \frac{\partial \sigma (z) }{\partial \Im\left[p_j\right] } & = & \frac{1}{z^2} \Re_{\mathrm{conj}}\mkern-6mu\left[\mkern-3mu  \frac{\mathrm{i} \; r_{j} }{\left(z-p_j\right)^2} \mkern-3mu \right] 
      \\
      \frac{\partial \sigma}{\partial r_j }(z) & \mkern-3mu = \mkern-3mu & \frac{1}{z^2} \frac{ \frac{1}{2} }{z-p_j} &  ,  &
      \frac{\partial \sigma}{\partial r_j^* }(z) & = & \frac{1}{z^2} \frac{ \frac{1}{2} }{z-p_j^*} \\
%      \frac{\eth \sigma}{\eth r_j }(z) & = & \frac{1}{z^2} \frac{ 1}{z-p_j^*} & \quad , \quad &
%      \frac{\eth \sigma}{\eth r_j^* }(z) & = & \frac{1}{z^2} \frac{ 1 }{z-p_j} \\
      \frac{\partial \sigma (z)}{\partial \Re\left[r_j\right] } & \mkern-3mu = \mkern-3mu & \frac{1}{z^2} \Re_{\mathrm{conj}}\mkern-6mu\left[ \mkern-3mu \frac{1 }{z-p_j} \mkern-3mu \right]  &  ,  &
      \frac{\partial \sigma (z)}{\partial \Im\left[r_j\right] } & = & \frac{1}{z^2} \Re_{\mathrm{conj}}\mkern-6mu\left[ \mkern-3mu \frac{\mathrm{i} }{z-p_j} \mkern-3mu \right] \\
      \frac{\partial \sigma}{\partial a_n }(z) & \mkern-3mu = \mkern-3mu & z^n  & & & & 
\label{eq:: d sigma(z) / d Pi}
\IEEEstrut\end{IEEEeqnarraybox}
\end{equation}
Moreover, the cross section sensitivities to resonance parameters $\frac{\partial\sigma}{\partial \Gamma}(z)$ (i.e. the partial differentials of the cross section with respect to resonance parameters) are subject to the following multipole representation:
\begin{equation}
\begin{IEEEeqnarraybox}[][c]{rCl}
      \frac{\partial \sigma}{\partial \Gamma }(z) & \underset{\mathcal{W}(E)}{=} &  \frac{1}{z^2} \Re_{\mathrm{conj}}\left[ \sum_{j\geq1}\frac{\left(\frac{\partial r_{j}}{\partial \Gamma} \right)}{z-p_j} + \frac{\left(r_{j} \cdot \frac{\partial p_{j}}{\partial \Gamma} \right)}{\left(z-p_j\right)^2} \right]\\
       & & + \sum_{n\geq -2}\left(\frac{\partial a_n}{\partial \Gamma} \right) z^n 
\label{eq:: d sigma(z) / d Gamma double poles}
\IEEEstrut\end{IEEEeqnarraybox}
\end{equation}
\end{lem}

% \begin{proof}
% Sensitivities (\ref{eq:: d sigma(z) / d Pi}) and (\ref{eq:: d sigma(z) / d Gamma double poles}) are a direct consequence of partial differentiation of the windowed multipole representation (\ref{eq:: sigma(z) Windowed Multipole representation}), by $\frac{\partial \,  \cdot}{\partial \Pi}$ and $\frac{\partial \, \cdot}{\partial \Gamma}$ respectively.
% \end{proof}

We seek to convert R-matrix resonance parameters covariances $\mathbb{V}\mathrm{ar}\left( \Gamma \right)$ into Windowed Multipole covariances $\mathbb{V}\mathrm{ar}\left( \Pi \right)$. Yet obtaining multipoles $\big\{ \Pi \big\}$ from resonance parameters $\big\{ \Gamma \big\}$ is not a simple transformation: one must solve the radioactive problem (\ref{eq:R_L radioactive problem in z-space}) for the poles $\big\{p_j\big\}$ and then compute the corresponding residues $\big\{\widehat{r}_j^{cc'}, r_j^c\big\}$ (theorem \ref{theo::WMP Representation}). 
We therefore take an implicit functions approach, and locally invert the $\Gamma \to \Pi$ transformation by means of the Jacobian matrix $\left(\frac{\partial \Pi}{ \partial \Gamma}\right)$, that is the sensitivities of windowed multipole coefficients to the R-matrix resonance parameters (Cauchy-Dini implicit functions theorem). 
Under the assumption of small deviations from the mean (small relative uncertainties), this yields a first-order linear relation from multipoles $\big\{ \Pi \big\}$ to resonance parameters $\big\{ \Gamma \big\}$.
In which case, the chain rule entails the multipoles $\big\{ \Pi \big\}$ are also subject to a multivariate normal distribution $\mathcal{N}\left( \Pi, \mathbb{V}\mathrm{ar}\left( \Pi \right) \right)$, the covariance of which is given by (\ref{eq:: Var(Π) from Var(Γ)}) (sometimes called ``sandwich rule'').
Therefore, the key to converting resonance covariances $\mathbb{V}\mathrm{ar}\left( \Gamma \right)$ into multipole covariances $\mathbb{V}\mathrm{ar}\left( \Pi \right)$ lies in the sensitivities $\left(\frac{\partial \Pi}{ \partial \Gamma}\right)$.
Theorem \ref{theo::WMP covariance} establishes a contour-integrals method to calculate these sensitivities $\left(\frac{\partial \Pi}{ \partial \Gamma}\right)$, provided R-matrix cross sections sensitivities $\frac{\partial \sigma}{\partial \Gamma }(E)$ from lemma \ref{lem::WMP sensitivities}.

\begin{theorem}\label{theo::WMP covariance}\textsc{Windowed Multipole Covariances} \\
Let us be provided with the sensitivities $\frac{\partial\sigma}{\partial \Gamma}(z)$ of R-matrix cross sections (analytically continued) to resonance parameters (\ref{eq:: d sigma(z) / d Gamma double poles}).
Then the multipole sensitivities (Jacobian matrix) with respect to the resonance parameters, $\left(\frac{\partial \Pi}{\partial \Gamma}\right)$, can be obtained from the following system (\ref{eq:: dΠ/dΓ contour integrals system}) of contour integrals in the complex plane, where $\mathfrak{C}_{p_j}$ designates any positively oriented simple closed contour containing only pole $p_j$. For instance, $\mathfrak{C}_{p_j}$ can be a circle of radius $\epsilon > 0 $ around pole $p_j$.
\begin{widetext}
\begin{equation}
\begin{IEEEeqnarraybox}[][c]{rCl}
    \frac{1}{2}\frac{r_{j}}{p_j^2} \cdot \left(\frac{\partial p_{j}}{\partial \Gamma} \right)  & = &   \frac{1}{2\pi\mathrm{i}}\oint_{\mathfrak{C}_{p_j}}(z-p_j)\frac{\partial \sigma}{\partial \Gamma }(z) \mathrm{d}z = \frac{\epsilon^2}{2\pi}\int_{\theta=0}^{2\pi} \frac{\partial \sigma}{\partial \Gamma }(p_j+\epsilon\mathrm{e}^{\mathrm{i}\theta})\mathrm{e}^{2\mathrm{i}\theta}\mathrm{d}\theta \\
    \frac{1}{2}\frac{1}{p_j^2} \cdot \left( \frac{\partial r_{j}}{\partial \Gamma} \right) - \frac{r_{j}}{p_j^3} \cdot \left(\frac{\partial p_{j}}{\partial \Gamma} \right) & = & \frac{1}{2\pi\mathrm{i}}\oint_{\mathfrak{C}_{p_j}}\frac{\partial \sigma}{\partial \Gamma }(z) \mathrm{d}z = \frac{\epsilon}{2\pi}\int_{\theta=0}^{2\pi}\frac{\partial \sigma}{\partial \Gamma }(p_j+\epsilon\mathrm{e}^{\mathrm{i}\theta})\mathrm{e}^{\mathrm{i}\theta}\mathrm{d}\theta \\
    \left(\frac{\partial a_n}{\partial \Gamma} \right)  +  \delta_{-1,n}  \Re_{\mathrm{conj}}\left[ \sum_{j=1}^{N_L}  \frac{2\left(r_{j} \cdot \frac{\partial p_{j}}{\partial \Gamma} \right) - p_j\left(\frac{\partial r_{j}}{\partial \Gamma} \right)}{p_j^3} \right]  & + & \delta_{-2,n} \Re_{\mathrm{conj}}\left[ \sum_{j=1}^{N_L} \frac{ \left(r_{j} \cdot \frac{\partial p_{j}}{\partial \Gamma} \right) - p_j \left(\frac{\partial r_{j}}{\partial \Gamma} \right)}{p_j^2} \right] \\ & = &  \frac{1}{2\pi\mathrm{i}}\oint_{\mathfrak{C}_{0}} \frac{1}{z^{n+1}}\frac{\partial \sigma}{\partial \Gamma }(z) \mathrm{d}z = \frac{1}{2\pi \epsilon^n}\int_{\theta=0}^{2\pi}\frac{\partial \sigma}{\partial \Gamma }(\epsilon\mathrm{e}^{\mathrm{i}\theta})\mathrm{e}^{-\mathrm{i}n\theta}\mathrm{d}\theta
\IEEEstrut\end{IEEEeqnarraybox} 
\label{eq:: dΠ/dΓ contour integrals system}
\end{equation}
\end{widetext}
For each energy window $\mathcal{W}(E)$, the multipole sensitivities $\left(\frac{\partial \Pi}{ \partial \Gamma}\right)$ from system (\ref{eq:: dΠ/dΓ contour integrals system}) can then be converted to first order into Windowed Multipole covariances $\mathbb{V}\mathrm{ar}\left( \Pi \right)$ as:
\begin{equation}
\begin{IEEEeqnarraybox}[][l c r]{l C r}
      \mathbb{V}\mathrm{ar}\left( \Pi \right) = \left(\frac{\partial \Pi}{ \partial \Gamma}\right) \mathbb{V}\mathrm{ar}\left( \Gamma \right)\left(\frac{\partial \Pi}{ \partial \Gamma}\right)^\dagger
\label{eq:: Var(Π) from Var(Γ)}
\IEEEstrut\end{IEEEeqnarraybox}
\end{equation}
where $\left[ \; \cdot \; \right]^\dagger$ designates the Hermitian conjugate (adjoint). 
\end{theorem}

\begin{proof}
Partial fraction expansion of (\ref{eq:: d sigma(z) / d Gamma double poles}) lemma \ref{lem::WMP sensitivities} yields
\begin{equation*}
\begin{IEEEeqnarraybox}[][c]{rCr}
      \frac{\partial \sigma}{\partial \Gamma }(z) &  \underset{\mathcal{W}(E)}{=} &   \Re_{\mathrm{conj}}\left[ \sum_{j=1}^{N_L} \frac{ \left(r_{j} \cdot \frac{\partial p_{j}}{\partial \Gamma} \right) - p_j \left(\frac{\partial r_{j}}{\partial \Gamma} \right)}{p_j^2z^2}   \right. \\
      & & + \frac{\left(r_{j} \cdot \frac{\partial p_{j}}{\partial \Gamma} \right)}{p_j^2(z-p_j)^2} +    \frac{2\left(r_{j} \cdot \frac{\partial p_{j}}{\partial \Gamma} \right) - p_j\left(\frac{\partial r_{j}}{\partial \Gamma} \right)}{p_j^3z}  \\
      & &  +  \left.   \frac{ p_j\left(\frac{\partial r_{j}}{\partial \Gamma} \right) - 2 \left(r_{j} \cdot \frac{\partial p_{j}}{\partial \Gamma} \right)}{p_j^3(z-p_j)} \right] \\
      & & + \sum_{n\geq -2} \left(\frac{\partial a_n}{\partial \Gamma} \right) z^n
\label{eq:: d sigma(z) / d Gamma (z) partial fraction approach}
\IEEEstrut\end{IEEEeqnarraybox}
\end{equation*}
The different residues associated with poles $0$ or $p_j$ are then obtained by invoking Cauchy's residue theorem and multiplying correspondingly by $z^n$ or $(z-p_j)$, yielding (\ref{eq:: dΠ/dΓ contour integrals system}).
Importantly, these contour integrals cannot be performed without having an analytic representation of the partial derivatives of the cross section at complex energies, $\frac{\partial \sigma}{\partial \Gamma }(z)$, which is made possible for open channels by Hwang's conjugate continuations (\ref{eq::|f|2 conjugate continuation}) and (\ref{eq::continued conjugate real part}).
Finally, (\ref{eq:: Var(Π) from Var(Γ)}) is a direct application of the well-known chain-rule first-order perturbation covariance formula.
\end{proof}

\begin{figure}[h] % replace 't' with 'b' to force it to be on the bottom
  \centering
  \includegraphics[width=0.49\textwidth]{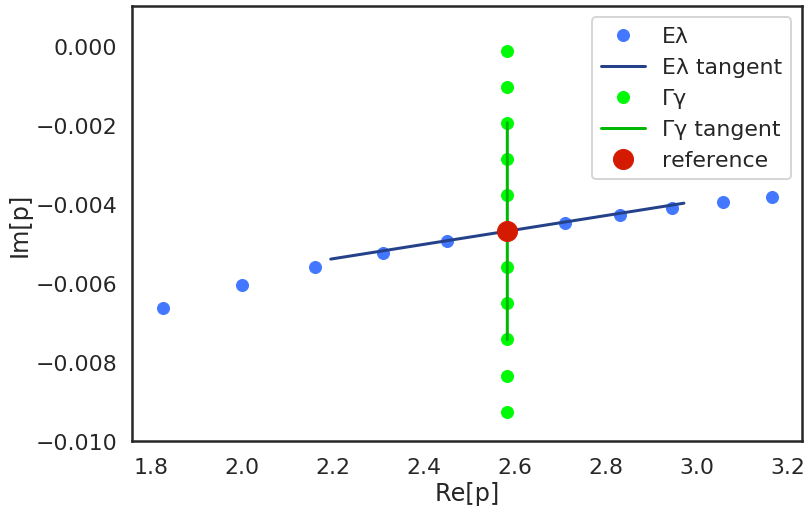}
  \caption{\small{Multipole sensitivities to R-matrix parameters $\left(\frac{\partial \Pi}{\partial \Gamma}\right)$. Trajectories of pole $p$ as resonance parameters $\left\{ \Gamma \right\} $ vary, using the SLBW approximation of the first resonance of $^{\mathrm{238}}\mathrm{U}$ (appendix \ref{appendix: Single Breit-Wigner resonance}). The blue points show how the pole changes as $E_\lambda$ is varied with equal spacing within 3 standard deviations of the enlarged covariance matrix, while the green points result from equally spaced variations of $\Gamma_\gamma$ within their uncertainty range (3 standard deviations of enlarged covariance matrix). The Jacobian $\left(\frac{\partial \Pi}{\partial \Gamma}\right)$ from system (\ref{eq:: dΠ/dΓ contour integrals system}) are the tangents of these trajectories from the mean pole $p$ (red reference point) and are shown in solid lines. Complex pole $p$ units are ($\sqrt{\mathrm{eV}}$). }}
  \label{fig: trajectores for senstivities}
\end{figure}

%%%%%%%%%%%%%%%%%%%%%%%%%%%%%%%%%%%%%%%%%%%%%%%%%%%%%%%%%%%%%%%%%%%%%%%%%%%%%%%%
\subsection{\label{subsec::Cross section uncertainties and parameter covariances} Cross section uncertainties and parameter covariances}
%%%%%%%%%%%%%%%%%%%%%%%%%%%%%%%%%%%%%%%%%%%%%%%%%%%%%%%%%%%%%%%%%%%%%%%%%%%%%%%%

By introducing resonance covariances $\mathbb{V}\mathrm{ar}\left( \Gamma \right)$, present standard nuclear data libraries are built with the implicit assumption that sampling resonance parameters from a multivariate normal distribution $\mathcal{N}\left( \Gamma , \mathbb{V}\mathrm{ar}\left( \Gamma \right)\right)$ and computing the corresponding cross sections 
$\sigma_\Gamma(E)$ generates outcome distributions commensurate to our experimental uncertainty. 
Note that this parameter uncertainty representation is not obvious \textit{in se}, because cross sections are measured at specific energies, and the measured cross section uncertainty is usually described with a given exogenous distribution (say normal, log-normal, or exponential), dictated by the experiment. 
Therefore, no parameter distribution (be it resonance parameters multivariate normal $\mathcal{N}\left( \Gamma , \mathbb{V}\mathrm{ar}\left( \Gamma \right)\right)$ or any other) can exactly reproduce the cross section uncertainty for each measurement energy.
And yet, these parameters distributions are our best way of balancing all the different uncertainties from disjointed experiments with the underlying R-matrix theory which unifies our understanding of nuclear interactions physics. 

Significant work has been carried out to infer parameter distributions that accurately reproduce our uncertainty of nuclear cross sections \cite{Frohner_Jeff_2000, baugeEvaluationCovarianceMatrix2007,capoteNewFormulationUnified2012,koningBayesianMonteCarlo2015, koningBayesianMonteCarlo2015a, helgessonCombiningTotalMonte2017,rochmanMonteCarloNuclear2018,alhassanBayesianUpdatingData2020,rochmanStatisticalAnalysisEvaluated2020}.
Assuming R-matrix cross section uncertainty is well represented by the resonance parameters multivariate normal distribution $\mathcal{N}\left( \Gamma , \mathbb{V}\mathrm{ar}\left( \Gamma \right)\right)$ documented in standard nuclear data libraries (file 32 in ENDF/B-VIII.0 \cite{ENDFBVIII8th2018brown}), there are two ways of translating this into cross section distributions: 1) first-order sensitivity propagation, or; 2) stochastic cross sections. 
\begin{enumerate}
    \item For any given energy $E$, first-order sensitivity propagation simply considers the R-matrix cross section sensitivities to resonance parameters $\frac{\partial \sigma}{\partial \Gamma }(E)$ and linearly converts the resonance parameter covariance $\mathbb{V}\mathrm{ar}\left( \Gamma \right)$ into a cross section covariance $\mathbb{V}\mathrm{ar}\left( \sigma(E) \right)$ at each energy $E$, using the chain rule:
    \begin{equation}
        \mathbb{V}\mathrm{ar}\left( \sigma_\Gamma(E) \right) = \left(\frac{\partial \sigma(E)}{\partial \Gamma }\right)\mathbb{V}\mathrm{ar}\left( \Gamma \right)\left(\frac{\partial \sigma(E)}{\partial \Gamma }\right)^\dagger
        \label{eq: sensitivity covariance of cross section from resonance parameters}
    \end{equation}
    The same approach can be undertaken using R-matrix cross section sensitivities to windowed multipoles $\frac{\partial \sigma}{\partial \Pi }(E)$, established in equations (\ref{eq:: d sigma(z) / d Pi}) of lemma \ref{lem::WMP sensitivities}, and then propagating the windowed multipole covariances $\mathbb{V}\mathrm{ar}\left( \Pi \right)$ to first order, yielding cross section covariances
    \begin{equation}
        \mathbb{V}\mathrm{ar}\left( \sigma_\Pi(E) \right) = \left(\frac{\partial \sigma(E)}{\partial \Pi }\right)\mathbb{V}\mathrm{ar}\left( \Pi \right)\left(\frac{\partial \sigma(E)}{\partial \Pi }\right)^\dagger
        \label{eq: sensitivity covariance of cross section from WMP}
    \end{equation}
    \item Stochastic cross sections consist of sampling resonance parameters $\big\{ \Gamma \big\}$ from their uncertainty distribution -- say multivariate normal $\mathcal{N}\left( \Gamma , \mathbb{V}\mathrm{ar}\left( \Gamma \right)\right)$ -- and computing the corresponding cross section $\sigma_\Gamma(E)$ as a function of energy
    \begin{equation}
        \mathrm{d}\mathbb{P}\Big(\sigma_\Gamma(E) \Big) = \sigma_{\mathrm{d}\mathbb{P}\left(\Gamma\right)}(E)  
        \label{eq: Stochastic cross sections resonance parameters}
    \end{equation}
    Alternatively, one could sample multipoles  $\big\{ \Pi \big\}$ from a windowed multipole distribution -- say multivariate normal $\mathcal{N}\left( \Pi , \mathbb{V}\mathrm{ar}\left( \Pi \right)\right)$ -- and correspondingly generate Windowed Multipole stochastic cross sections
    \begin{equation}
        \mathrm{d}\mathbb{P}\Big(\sigma_\Pi(E) \Big) = \sigma_{\mathrm{d}\mathbb{P}\left(\Pi\right)}(E)  
        \label{eq: Stochastic cross sections multipoles}
    \end{equation}
    
\end{enumerate}
Stochastic cross sections uncertainties only match first order sensitivity approaches (\ref{eq: sensitivity covariance of cross section from resonance parameters}) and (\ref{eq: sensitivity covariance of cross section from WMP}) for very small covariances. This is because normally distributed resonance parameters do not translate into normally distributed cross sections (\ref{eq: Stochastic cross sections resonance parameters}): sampling resonance parameters from $\mathcal{N}\left( \Gamma , \mathbb{V}\mathrm{ar}\left( \Gamma \right)\right)$ and then computing the corresponding cross sections (\ref{eq:partial sigma_cc'}) and (\ref{eq:total σ_c}) through R-matrix equations (\ref{eq: Transmission matrix}), (\ref{eq:U expression}), (\ref{eq: def Kapur-Peierls operator}), (\ref{eq:R expression}), and (\ref{eq: L operator}), cannot in general lead to normally distributed cross sections $\sigma_\Gamma(E)$ at all energies. However, they do in the linear case, which is a good first-order approximation for small covariances.

    Stochastic cross sections (\ref{eq: Stochastic cross sections resonance parameters}) are at the core of the TENDL library \cite{TENDLkoningModernNuclearData2012, koningTENDLCompleteNuclear2019}, and being able to sample them is a necessary prerequisite to the \textit{Total Monte Carlo} uncertainty propagation method \cite{TMC, fast_TMC, Lead_TMC}. In practice, this has been a major computational challenge, requiring to sample resonance parameters from standard nuclear data libraries, reconstruct the corresponding nuclear cross sections at zero Kelvin ($0 \mathrm{K}$), and then process each one (with codes such as NJOY \cite{NJOYmacfarlaneMethodsProcessingENDF2010}) to compute the corresponding cross sections at temperature $T$ (c.f. discussion of Doppler broadening and thermal scattering in section \ref{sec:DB of WMP}). All this is costly, and storing the pre-processed cross sections consumes vasts amount of memory. Because one can directly compute Doppler-broadened nuclear cross sections from Windowed Multipole parameters $\big\{\Pi\big\}$ (c.f. theorem \ref{theo::WMP Doppler broadening} section \ref{sec:DB of WMP}), the Windowed Multipole Library can generate stochastic cross sections (\ref{eq: Stochastic cross sections multipoles}) on-the-fly, without any pre-processing nor storage, a true physics-enabled computational breakthrough.

\begin{figure}[ht!!] % replace 't' with 'b' to force it to be on the bottom
  \centering
  \subfigure[\ $^{\text{238}}\text{U}$ first capture resonance, parameters sampled from ENDF/B-VIII uncertainty. 30 samples shown here.]{\includegraphics[width=0.49\textwidth]{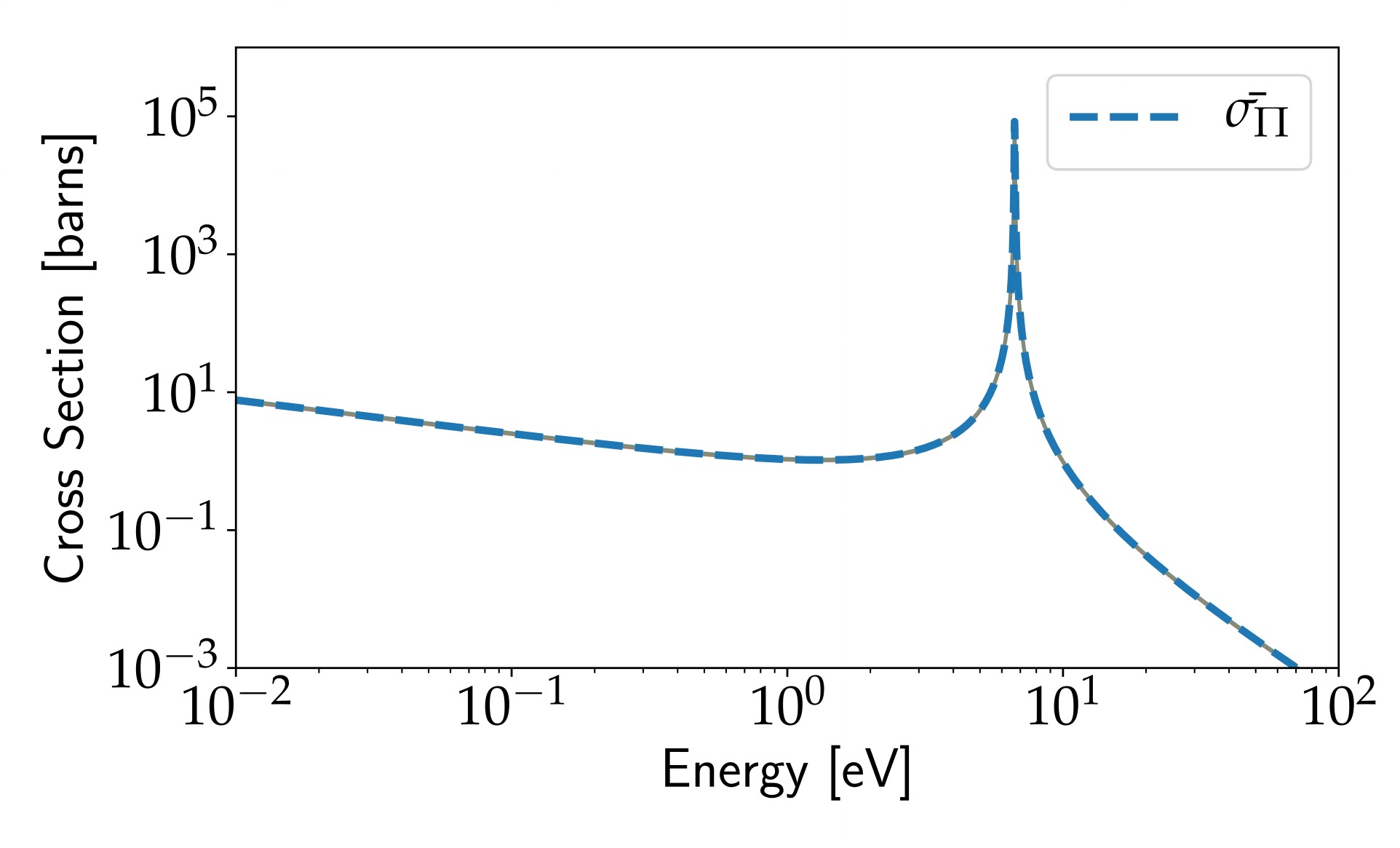}}
  \subfigure[\ Cross section histogram at resonance energy $E_\lambda = 6.67428  $ eV.]{\includegraphics[width=0.49\textwidth]{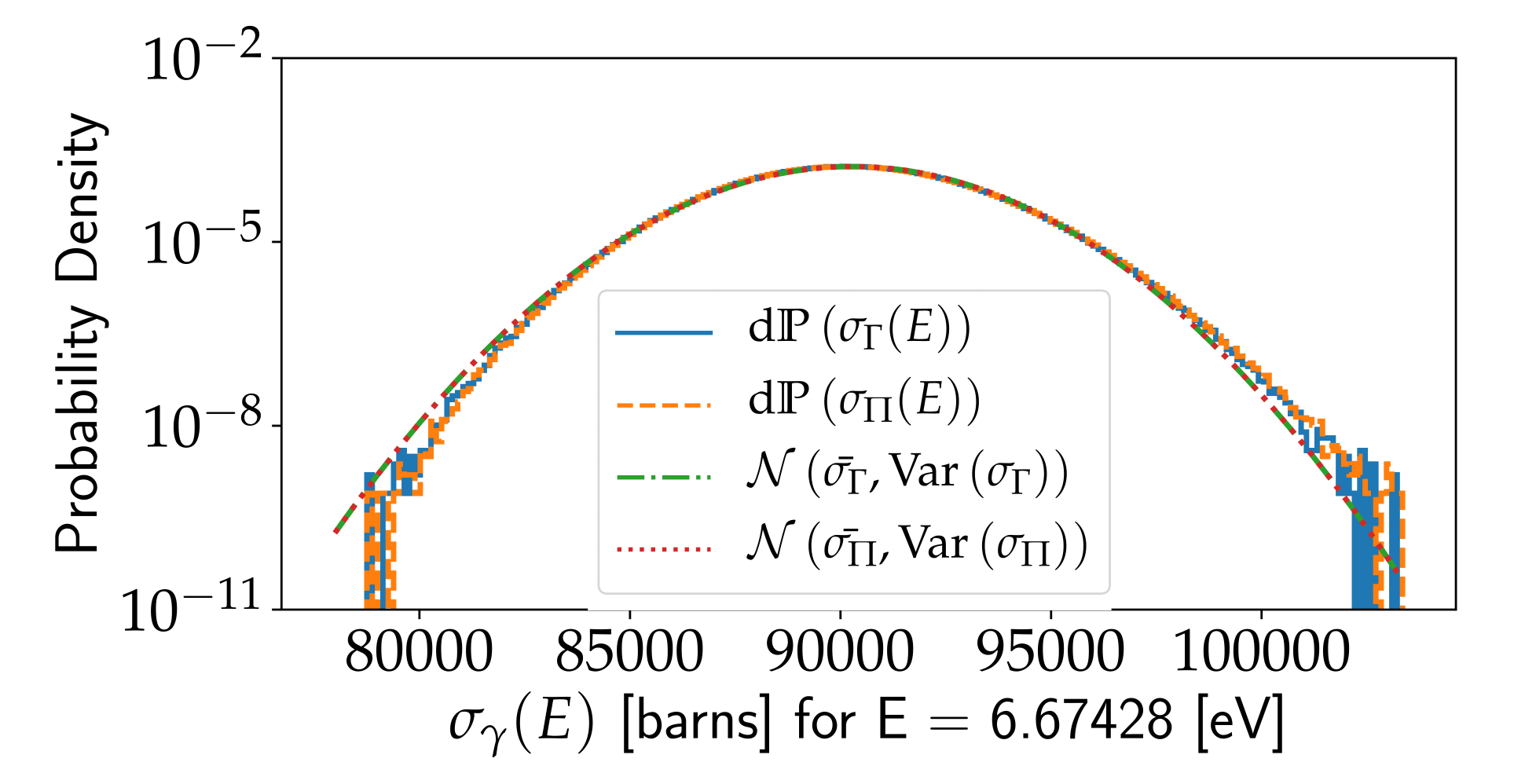}}
  \caption{\small{R-matrix cross sections uncertainty, computed either from the ENDF/B-VIII resonance parameters covariance $\mathbb{V}\mathrm{ar}\left( \Gamma \right)$ (table \ref{tab:resonance parameters U-238 first resonance} in appendix \ref{appendix: Single Breit-Wigner resonance}), or from the multipoles covariance $\mathbb{V}\mathrm{ar}\left( \Pi \right)$, as converted through (\ref{eq:: Var(Π) from Var(Γ)}), for both the stochastic cross sections (\ref{eq: Stochastic cross sections resonance parameters}, \ref{eq: Stochastic cross sections multipoles}) and the sensitivities approach (\ref{eq: sensitivity covariance of cross section from resonance parameters}, \ref{eq: sensitivity covariance of cross section from WMP}).}}
  \label{fig:ENDF Uncertainty}
\end{figure}

\begin{figure}[ht!!] % replace 't' with 'b' to force it to be on the bottom
  \centering
  \subfigure[\ $^{\text{238}}\text{U}$ first capture resonance, parameters sampled from an enlarged ENDF/B-VIII uncertainty. 30 samples shown here.]{\includegraphics[width=0.49\textwidth]{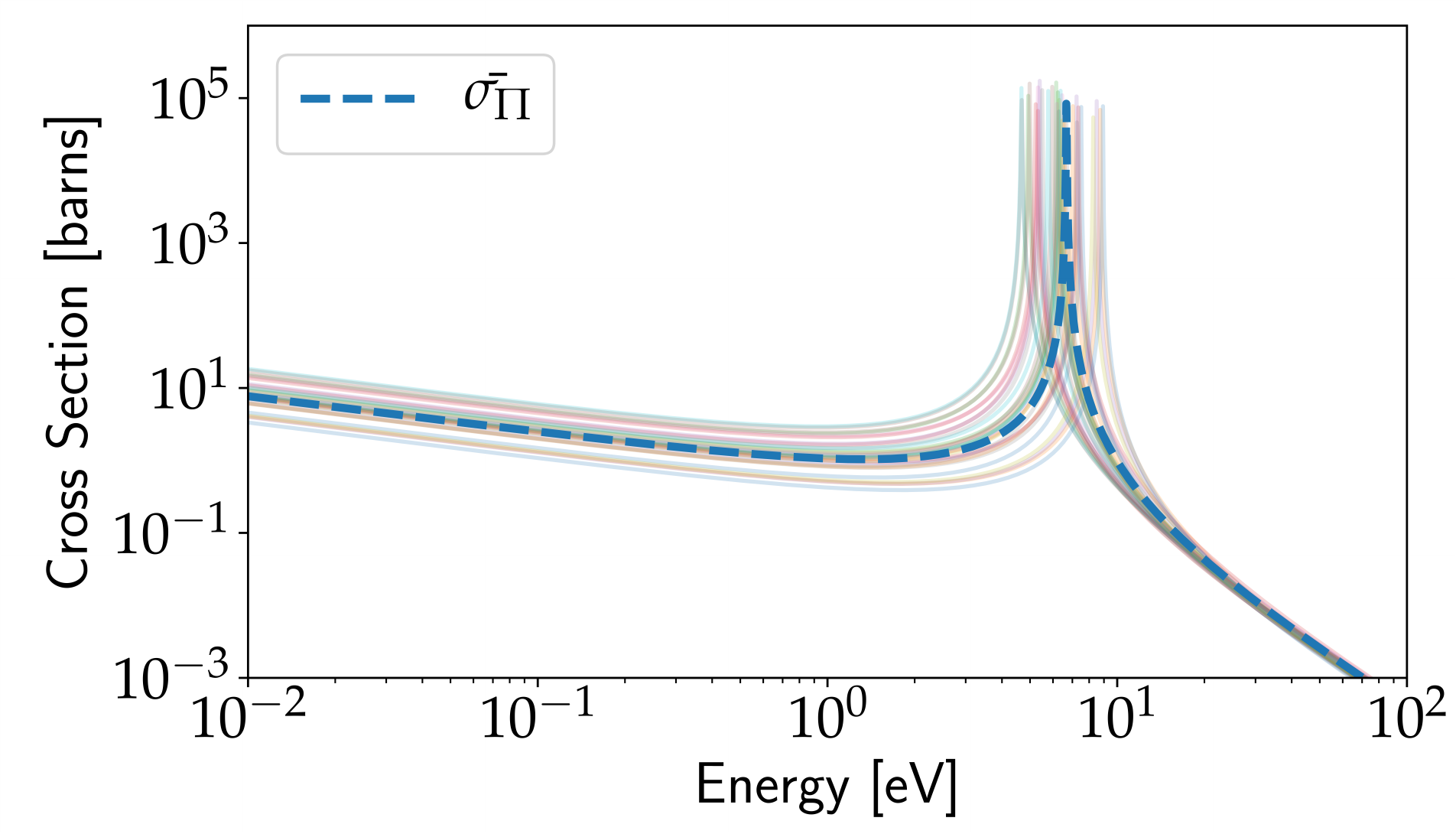}}
  \subfigure[\ Cross section histogram at resonance energy $E_\lambda = 6.67428  $ eV.]{\includegraphics[width=0.49\textwidth]{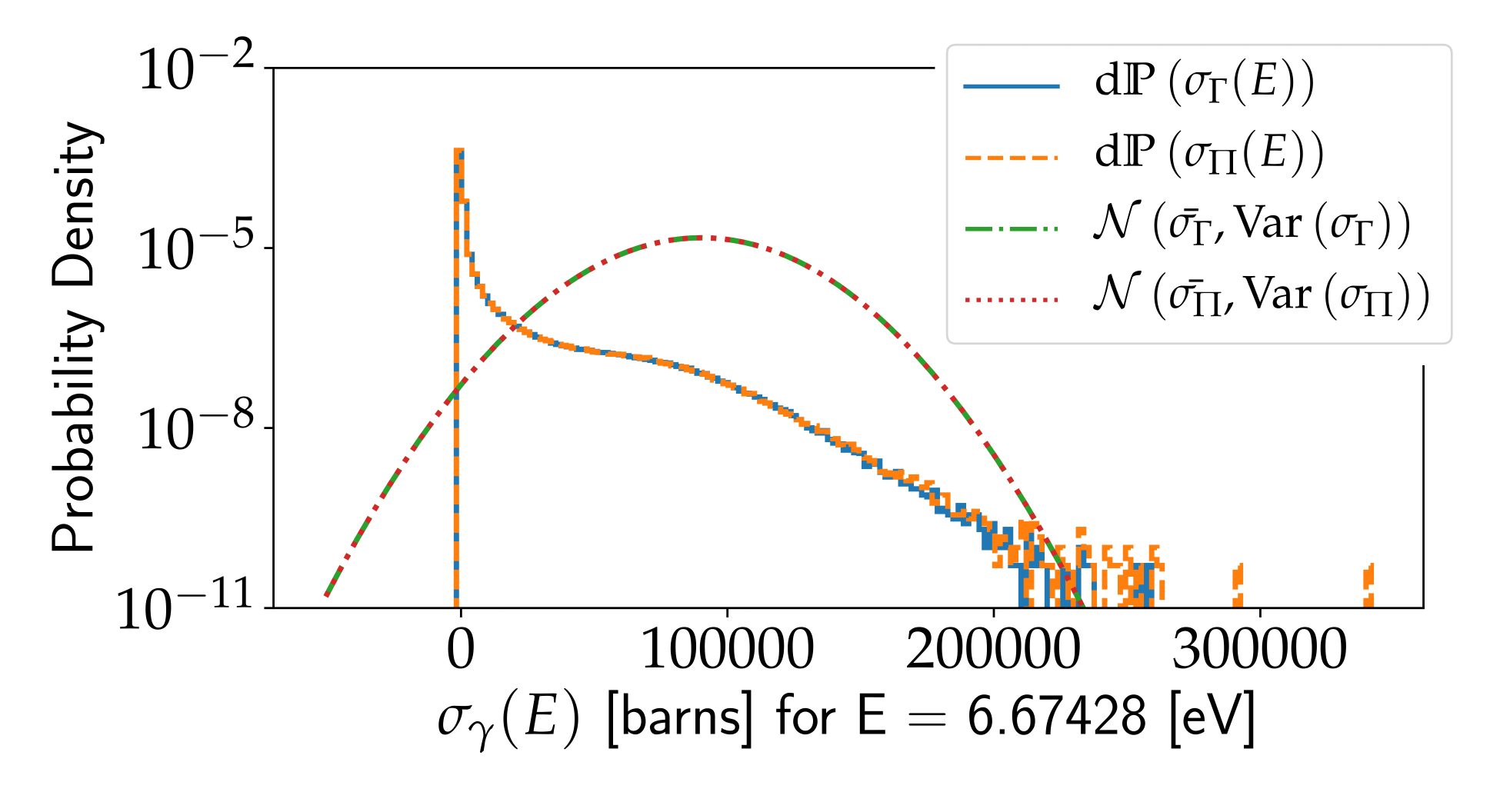}}
  \caption{\small{R-matrix cross sections uncertainty, computed either from the enlarged ENDF/B-VIII resonance parameters covariance $\mathbb{V}\mathrm{ar}\left( \Gamma \right)$ (table \ref{tab:resonance parameters U-238 first resonance} in appendix \ref{appendix: Single Breit-Wigner resonance}), or from the multipoles covariance $\mathbb{V}\mathrm{ar}\left( \Pi \right)$, as converted through (\ref{eq:: Var(Π) from Var(Γ)}), for both the stochastic cross sections (\ref{eq: Stochastic cross sections resonance parameters}, \ref{eq: Stochastic cross sections multipoles}) and the sensitivities approach (\ref{eq: sensitivity covariance of cross section from resonance parameters}, \ref{eq: sensitivity covariance of cross section from WMP}).}}
  \label{fig:LARGE ENDF Uncertainty}
\end{figure}

Regardless of the method employed to represent nuclear cross section uncertainty, it would be desirable that the uncertainties stemming from a windowed multipole representation $\big\{ \Pi \big\}$ are consistent with those stemming from resonance parameters $\big\{ \Gamma \big\}$ upon converting their covariances as indicated in equation (\ref{eq:: Var(Π) from Var(Γ)}) of theorem \ref{theo::WMP covariance}.
We undertook numerical experiments to measure the cross section uncertainty distributions generated by either covariances $\mathbb{V}\mathrm{ar}\left( \Gamma \right)$ or $\mathbb{V}\mathrm{ar}\left( \Pi \right)$, for both sensitivity method (\ref{eq: sensitivity covariance of cross section from resonance parameters}) and (\ref{eq: sensitivity covariance of cross section from WMP}), or stochastic cross sections (\ref{eq: Stochastic cross sections resonance parameters}) and (\ref{eq: Stochastic cross sections multipoles}). 
We treated the simple case of the first single-level Breit-Wigner capture resonance of uranium isotope $^{\text{238}}\text{U}$, which admits closed form explicit expressions for the multipoles, the cross sections, and the sensitivities, all documented in appendix \ref{appendix: Single Breit-Wigner resonance}. We compared the methods for both the ENDF/B-VIII.0 resonance parameters covariance (which is small as this is a very well known resonance), and an enlarged covariance matrix which conserves the same correlations but brings the cross section dependency past the linear regime. Both covariances are documented in table \ref{tab:resonance parameters U-238 first resonance} (appendix \ref{appendix: Single Breit-Wigner resonance}), and figures \ref{fig:ENDF Uncertainty} and \ref{fig:LARGE ENDF Uncertainty} show the following trends:
\begin{itemize}
    \item For the sensitivity method (\ref{eq: sensitivity covariance of cross section from resonance parameters}) or (\ref{eq: sensitivity covariance of cross section from WMP}), the cross section uncertainty is identical for either the resonance parameter covariance $\mathbb{V}\mathrm{ar}\left( \Gamma \right)$ or the windowed multipole covariance $\mathbb{V}\mathrm{ar}\left( \Pi \right)$, which is the immediate consequence of conversion (\ref{eq:: Var(Π) from Var(Γ)}).
    \item For the stochastic cross section method (\ref{eq: Stochastic cross sections resonance parameters}) or (\ref{eq: Stochastic cross sections multipoles}), sampling parameters from $\mathcal{N}\left( \Gamma , \mathbb{V}\mathrm{ar}\left( \Gamma \right)\right)$ or $\mathcal{N}\left( \Pi , \mathbb{V}\mathrm{ar}\left( \Pi \right)\right)$ generates similar cross section distributions. 
\end{itemize}

In the small covariance case of figure \ref{fig:ENDF Uncertainty}, the stochastic cross section distributions (\ref{eq: Stochastic cross sections resonance parameters}) and (\ref{eq: Stochastic cross sections multipoles}) are very close to the normal distributions from the sensitivity approach (\ref{eq: sensitivity covariance of cross section from WMP}), though at the tails they start differing. In the large covariance case of figure \ref{fig:LARGE ENDF Uncertainty}, the stochastic cross sections distributions are radically different from the normal distribution of the sensitivity method. This discrepancy is made more blatant because the cross section distribution is recorded at resonance peak energy $E_\lambda = 6.67428 \; \mathrm{eV}$, hence a small shift in resonance energy $E_\lambda $ can dramatically lower the cross section value. 
This illustrates the fact that in theorem \ref{theo::WMP covariance}, when converting the resonance parameters covariances $\mathbb{V}\mathrm{ar}\left( \Gamma \right)$ into  Windowed Multipoles covariances $\mathbb{V}\mathrm{ar}\left( \Pi \right)$ through (\ref{eq:: Var(Π) from Var(Γ)}), the linear assumption used for the local inversion using Jacobians $\left(\frac{\partial \Pi}{\partial \Gamma}\right)$ from (\ref{eq:: dΠ/dΓ contour integrals system}) holds for a wider range of resonance parameters than the liner assumption for the cross section sensitivity method (\ref{eq: sensitivity covariance of cross section from WMP}). This can be seen in figure \ref{fig: trajectores for senstivities}, where the tangent lines from Jacobians $\left(\frac{\partial \Pi}{\partial \Gamma}\right)$ are close to the conversion surface, trajectories of $\Pi(\Gamma)$, even after three standard deviations of the large covariance matrix, something clearly not true of the cross section linear behavior at peak energy from figure \ref{fig:LARGE ENDF Uncertainty}.

Therefore, whichever method is chosen to represent the nuclear cross sections uncertainty, the Windowed Multipoles covariances $\mathbb{V}\mathrm{ar}\left( \Pi \right)$ from theorem \ref{theo::WMP covariance} faithfully reproduce the uncertainty from the resonance parameters covariances $\mathbb{V}\mathrm{ar}\left( \Gamma \right)$.

%%%%%%%%%%%%%%%%%%%%%%%%%%%%%%%%%%%%%%%%%%%%%%%%%%%%%%%%%%%%%%%%%%%%%%%%%%%%%%%%
%*******************************************************************************
%%%%%%%%%%%%%%%%%%%%%%%%%%%%%%%%%%%%%%%%%%%%%%%%%%%%%%%%%%%%%%%%%%%%%%%%%%%%%%%%
\section{\label{sec:DB of WMP} Doppler broadening of Windowed Multipole cross sections}
%%%%%%%%%%%%%%%%%%%%%%%%%%%%%%%%%%%%%%%%%%%%%%%%%%%%%%%%%%%%%%%%%%%%%%%%%%%%%%%%
%*******************************************************************************
%%%%%%%%%%%%%%%%%%%%%%%%%%%%%%%%%%%%%%%%%%%%%%%%%%%%%%%%%%%%%%%%%%%%%%%%%%%%%%%%

Hitherto, we have established that the zero Kelvin (0 K) windowed multipole representation of cross sections is equivalent to the traditional Wigner-Eisenbud R-matrix parametrization, in both cross section values and their uncertainties. 
We henceforth study how temperature affects R-matrix cross sections at the nuclear level through \textit{Doppler broadening} (we do not address thermal neutron scattering at the crystalline level), and derive how the windowed multipole representation exhibits a major advantage: in its form (\ref{eq:: sigma(z) Windowed Multipole representation}) -- exact for zero-threshold channels or windows without thresholds, and otherwise an accurate approximation -- the window multipole representation of R-matrix cross sections can be Doppler broadened precisely by means of closed-form formulae (theorem \ref{theo::WMP Doppler broadening}). This enables the long sought-after computational capability of on-the-fly Doppler broadening of nuclear cross sections \cite{solbrig1961doppler, blackshawScatteringFunctionsLowEnergy1967, cullenExactDopplerBroadening1976a, Hwang_1987, Trumbull_2006, yesilyurtOntheFlyDopplerBroadening2012a, Forget_2013, ferranNewMethodDoppler2015a, Ducru_JCP_2017}.

%%%%%%%%%%%%%%%%%%%%%%%%%%%%%%%%%%%%%%%%%%%%%%%%%%%%%%%%%%%%%%%%%%%%%%%%%%%%%%%%
%*******************************************************************************
%%%%%%%%%%%%%%%%%%%%%%%%%%%%%%%%%%%%%%%%%%%%%%%%%%%%%%%%%%%%%%%%%%%%%%%%%%%%%%%%
\subsection{\label{subsec:Doppler broadening Solbrig Kernel}Doppler broadening of nuclear cross sections: Solbrig's Kernel}
%%%%%%%%%%%%%%%%%%%%%%%%%%%%%%%%%%%%%%%%%%%%%%%%%%%%%%%%%%%%%%%%%%%%%%%%%%%%%%%%
%*******************************************************************************
%%%%%%%%%%%%%%%%%%%%%%%%%%%%%%%%%%%%%%%%%%%%%%%%%%%%%%%%%%%%%%%%%%%%%%%%%%%%%%%%
As temperature rises, nuclei vibrate, so that the effective cross section for a beam of particles sent upon a target at a given energy and wavenumber is the statistical result of the zero Kelvin cross sections averaged out on all the possible relative energies at which the target and the beam interact. For non-relativistic, non-massless particles (not photons) in the semi-classical representation, Doppler broadening of nuclear cross section is the process of integration over the target velocity distribution, assuming the latter is an isotropic Maxwellian (that is a Boltzmann distribution of energies). 
Solbrig derived this Doppler broadening integral in eq. (3) p. 259 of \cite{solbrig1961doppler}, where the cross section $\sigma_T(E)$ at temperature $T$ and energy $E$ (in the laboratory coordinates) is related to the cross section $\sigma(E)$ at temperature $T_0$ as:
\begin{equation}
\begin{IEEEeqnarraybox}[][c]{rCl}
E\sigma_T(E) & = & \mkern-3mu \int\displaylimits^\infty_0 \mkern-3mu \frac{\sigma(E'){E'}^{\frac{1}{2}}}{2\beta\sqrt{\pi}} \mkern-3mu \left[ \mathrm{e}^{-\left(\frac{\sqrt{E'}-\sqrt{E}}{\beta}\right)^2} \mkern-10mu - \mathrm{e}^{-\left(\frac{\sqrt{E'}+\sqrt{E}}{\beta}\right)^2} \right]\mkern-6mu \mathrm{d}E'
\IEEEstrut \end{IEEEeqnarraybox}
\label{eq::E-space Doppler broadening}
\end{equation}
where $\beta$ is the square root temperature energy parameter:
\begin{equation}
\begin{IEEEeqnarraybox}[][c]{rCl}
\IEEEstrut
\beta & \triangleq & \sqrt{\frac{k_\mathbb{B}(T-T_0)}{A}}
\end{IEEEeqnarraybox}
\label{eq::beta def}
\end{equation}
where $A$ designates the atomic mass number, and $k_\mathbb{B}$ the universal Boltzmann constant. \\
Upon change of variable to $z = \sqrt{E}$, the Doppler broadening operation (\ref{eq::E-space Doppler broadening}) becomes Solbrig's kernel:
\begin{equation}
\begin{IEEEeqnarraybox}[][c]{rCl}
z^2\sigma_T(z) & = & \int^\infty_0 x^2\sigma_{T_0}(x) \cdot \mathcal{K}^{\mathbb{D}}_\beta(z,x) \mathrm{d}x  \\
$where: $
 \  \mathcal{K}^{\mathbb{D}}_\beta(z,x)  & \triangleq & \frac{1}{\beta\sqrt{\pi}}\left[ \mathrm{e}^{-\left(\frac{z-x}{\beta}\right)^2} - \mathrm{e}^{-\left(\frac{z+x}{\beta}\right)^2} \right]
 \label{eq::Solbrig_Kernel}
\IEEEstrut \end{IEEEeqnarraybox}
\end{equation}
Note that for zero-threshold channels, where $z \propto k_c(E)$, Solbrig kernel (\ref{eq::Solbrig_Kernel}) is an integral operator acting on $k_c^2(E)\cdot  \sigma_c(E)$, which is the transmission matrix square amplitudes from cross section definition (\ref{eq:partial sigma_cc'}).
The Solbrig kernel (\ref{eq::Solbrig_Kernel}) thus acts directly on the interaction probabilities, rather than the actual cross section, just as the channel reversibility equivalence (\ref{eq:cross section reversibility equivalence}). 

Solbrig kernel integral (\ref{eq::Solbrig_Kernel}) has presented major computational challenges in nuclear reactor physics.
When no information is provided as to the functional form of the zero Kelvin cross section $\sigma(E)$ -- i.e. it is considered a point-wise input -- the traditional way of computing the Doppler broadened cross section at any temperature $\sigma_T(E)$ has been to pre-tabulate exact cross sections $\sigma_{T_i}(E)$ (usually using the SIGMA1 algorithm of \cite{cullenExactDopplerBroadening1976a}) for a grid of reference temperatures $\big\{T_i\big\}$, and then interpolate between these points to obtain $\sigma_T(E)$ \cite{Conlin_pseudo_mateirals_2005, trumbullTreatmentNuclearData2006a, Chinese_Linear_Interpolation_2012}.
However, storing all these pre-computed cross sections at reference temperatures $\big\{T_i\big\}$ represents a considerable memory burden, which is why methods to minimize the memory footprint and perform Doppler broadening (\ref{eq::Solbrig_Kernel}) on-the-fly have been actively sought after \cite{yesilyurtOntheFlyDopplerBroadening2012a}. The most state-of-the-art approaches are either optimal temperature Doppler kernel reconstruction quadratures \cite{Ducru_JCP_2017} (which only require 10 reference temperatures $\big\{T_i\big\}$ for standard nuclear reactor codes), new Fourier transform methods \cite{ferranNewMethodDoppler2015a}, or Monte Carlo target motion sampling rejection schemes \cite{viitanenExplicitTreatmentThermal2012a, viitanenTargetMotionSampling2014, romanoImprovedTargetVelocity2018}.

To do better, one must look at the functional form of the cross section. 
When the reference temperature is zero Kelvin $T_0 = 0 \text{K}$, we have shown in section \ref{sec:From R-matrix to Windowed Multipole} that R-matrix cross sections are the sum of threshold behavior and resonances. Resonances have traditionally been Doppler broadened approximately, using Voigt profiles \cite{solbrig1961doppler}, as we here recall in section \ref{subsec:Approximate Doppler broadening of Breit-Wigner resonances: Voigt profiles}.

%%%%%%%%%%%%%%%%%%%%%%%%%%%%%%%%%%%%%%%%%%%%%%%%%%%%%%%%%%%%%%%%%%%%%%%%%%%%%%%%
%*******************************************************************************
%%%%%%%%%%%%%%%%%%%%%%%%%%%%%%%%%%%%%%%%%%%%%%%%%%%%%%%%%%%%%%%%%%%%%%%%%%%%%%%%
\subsection{\label{subsec:Approximate Doppler broadening of Breit-Wigner resonances: Voigt profiles}Approximate Doppler broadening of Breit-Wigner resonances: Voigt profiles}
%%%%%%%%%%%%%%%%%%%%%%%%%%%%%%%%%%%%%%%%%%%%%%%%%%%%%%%%%%%%%%%%%%%%%%%%%%%%%%%%
%*******************************************************************************
%%%%%%%%%%%%%%%%%%%%%%%%%%%%%%%%%%%%%%%%%%%%%%%%%%%%%%%%%%%%%%%%%%%%%%%%%%%%%%%%

The traditional approach to Doppler broadening nuclear cross sections has been to consider individual Single-Level Breit-Wigner resonances (\ref{eq:: SLBW Lorenzian profiles for resonances}) at zero Kelvin, with both symmetric (Cauchy-Lorentz distributions) and anti-symmetric components, assuming it has a zero-energy threshold where it behaves as an s-wave neutron channel (angular momentum $\ell = 0$), so that we can multiply the resonance (\ref{eq:: SLBW Lorenzian profiles for resonances}) by the threshold behavior $\frac{1}{\sqrt{E}}$, as described by Wigner in III.A.2 \cite{Wigner_Thresholds_1948}:
\begin{equation}
\begin{IEEEeqnarraybox}[][c]{rCl}
    \sigma_0^{\mathrm{SLBW}}(E) & \triangleq & \frac{1}{\sqrt{E}} \Re\left[\frac{a + \mathrm{i}b}{ E - \mathcal{E}_j}\right]  \\
    & = & \frac{1}{\sqrt{E}} \left[\left( \frac{a}{ \Gamma_j/2} \right) \chi_0(x)+   \left(\frac{b}{ \Gamma_j/2}\right) \psi_0(x) \right]
\label{eq:: SLBW Single Breit Wigner resonance}
\IEEEstrut\end{IEEEeqnarraybox}
\end{equation}
where $x \ \triangleq \  \left( \frac{E - E_j}{\Gamma_j/2} \right)  $ with $\mathcal{E}_j \triangleq E_j - \mathrm{i}\frac{\Gamma_j}{2}$ from (\ref{eq:E_j pole def}), and
\begin{equation}
\begin{IEEEeqnarraybox}[][c]{rcl}
\psi_0(x)  & \ \triangleq \ &  \frac{1 }{ x^2 + 1} = \frac{\Gamma_j^2/4 }{ \left( E - E_j \right)^2 + \frac{\Gamma_j^2}{4}} \\
\chi_0(x)  & \ \triangleq \ & \frac{x }{ x^2 + 1} =   \frac{\left( E - E_j \right) \frac{ \Gamma_j}{2} }{ \left(  E -E_j \right)^2 + \frac{\Gamma_j^2}{4}}  
\IEEEstrut\end{IEEEeqnarraybox}
\label{eq::0K psi chi functions}
\end{equation}
Upon Doppler broadening (\ref{eq::E-space Doppler broadening}), Single-Level Breit-Wigner resonance (\ref{eq:: SLBW Single Breit Wigner resonance}) becomes: 
\begin{equation}
\begin{IEEEeqnarraybox}[][c]{rcl}
\sigma^{\mathrm{SLBW}}_T(E)& \ = \ & \frac{1}{\sqrt{E}} \left[\left( \frac{a}{ \Gamma_j/2} \right)  \chi_T(E) +   \left(\frac{b}{ \Gamma_j/2}\right)   \psi_T(E) \right]
\IEEEstrut\end{IEEEeqnarraybox}
\label{eq:: Doppler broadened SLBW Single Breit Wigner resonance}
\end{equation}
where $ \chi_T$ and $ \psi_T$ are defined using $x' \triangleq \left( \frac{E' - E_j}{\Gamma_j/2} \right)$ as
\begin{equation}
\begin{IEEEeqnarraybox}[][c]{rcl}
\chi_T(E)  &  \mkern2mu \triangleq \mkern2mu&  \frac{E^{-\frac{1}{2}}}{2\beta\sqrt{\pi}} \mkern-5mu \int\displaylimits^\infty_0 \mkern-4mu \chi_0(x') \mkern-4mu \left[ \mathrm{e}^{-\left(\frac{\sqrt{E'}-\sqrt{E}}{\beta}\right)^2}\mkern-10mu - \mathrm{e}^{-\left(\frac{\sqrt{E'}+\sqrt{E}}{\beta}\right)^2} \mkern-2mu\right]\mkern-6mu \mathrm{d}E' \\
\psi_T(E)  &  \triangleq  & \frac{E^{-\frac{1}{2}}}{2\beta\sqrt{\pi}} \mkern-5mu \int\displaylimits^\infty_0 \mkern-4mu \psi_0(x') \mkern-4mu \left[ \mathrm{e}^{-\left(\frac{\sqrt{E'}-\sqrt{E}}{\beta}\right)^2} \mkern-10mu - \mathrm{e}^{-\left(\frac{\sqrt{E'}+\sqrt{E}}{\beta}\right)^2} \mkern-2mu\right]\mkern-6mu \mathrm{d}E'
\IEEEstrut\end{IEEEeqnarraybox}
\label{eq::T K psi chi functions}
\end{equation}
To compute these functions, the following approximations are then traditionally introduced (c.f. \cite{solbrig1961doppler}, or section 3.3.3 chapter 4, volume 1 of \cite{cacuciHandbookNuclearEngineering2010}):
\begin{enumerate}
    \item the Maxwell approximation, whereby we assume the second exponential term is vanishingly small:  $\mathrm{e}^{-\left(\frac{\sqrt{E'}+\sqrt{E}}{\beta}\right)^2} \ll 1 $. This is valid for $E \gg \beta^2 $, but fails at low energies or high temperatures, 
    \item a Taylor expansion around the energy of Doppler broadening: $E' = E + \epsilon$, with $\epsilon \ll 1$. This leads to $\sqrt{E'} - \sqrt{E} = \frac{\epsilon}{2 \sqrt{E}} + \mathcal{O}\left( \epsilon^2 \right)$, so that we approximate $\mathrm{e}^{-\left(\frac{\sqrt{E'}-\sqrt{E}}{\beta}\right)^2} \approx \mathrm{e}^{-\left(\frac{E'-E}{2 \sqrt{E}\beta}\right)^2} $ in the integrals, which under change of variable $E' \to x' $ become
    \begin{equation*}
    \begin{IEEEeqnarraybox}[][c]{rcl}
    \psi_T(x)  & \ \simeq \ & \frac{1}{2\beta\sqrt{\pi E }}\frac{\Gamma_j}{2} \int_{-2E_j/\Gamma_j}^{\infty} \psi_0(x')\mathrm{e}^{-\frac{\left(x'-x\right)^2}{4\tau}} \mathrm{d}x'  \\
    \chi_T(x)  & \ \simeq \ &  \frac{1}{2\beta\sqrt{\pi E }}\frac{\Gamma_j}{2} \int_{-2E_j/\Gamma_j}^{\infty} \chi_0(x')\mathrm{e}^{-\frac{\left(x'-x\right)^2}{4\tau}} \mathrm{d}x'
    \IEEEstrut\end{IEEEeqnarraybox}
    \label{eq::T K psi chi functions derivation}
    \end{equation*}
    where we defined
    \begin{equation}
    \begin{IEEEeqnarraybox}[][c]{rcl}
    \tau  & \ \triangleq \ &  4 E \left( \frac{\beta}{\Gamma_j}\right)^2 = 4 E \frac{k_\mathbb{B} (T-T_0)}{A\Gamma_j^2}
    \IEEEstrut\end{IEEEeqnarraybox}
    \label{eq::tau def}
    \end{equation}
    \item we furthermore assume $2 E_j \gg \Gamma_j$, so that we approximate the integral lower limit to $-\infty$, yielding
    \begin{equation}
    \begin{IEEEeqnarraybox}[][c]{rcl}
     \psi_T(x)  & \ \simeq \ & \frac{1}{\sqrt{4 \pi \tau }} \int^{+\infty}_{-\infty} \frac{1 }{ 1+ {x'}^2}\mathrm{e}^{-\frac{\left(x'-x\right)^2}{4\tau}} \mathrm{d}x'  \\
     \chi_T(x)  & \ \simeq \ &  \frac{1}{\sqrt{4 \pi \tau }} \int^{+\infty}_{-\infty} \frac{x' }{ 1+ {x'}^2}\mathrm{e}^{-\frac{\left(x'-x\right)^2}{4\tau}} \mathrm{d}x'
    \IEEEstrut\end{IEEEeqnarraybox}
    \label{eq::T K psi chi functions Voigt}
    \end{equation}
\end{enumerate}
The latter are the standard Voigt functions, $\mathrm{U}(x,\tau)$ and $\mathrm{V}(x,\tau)$, defined in section 7.19 of \cite{NIST_DLMF}, which are related to the Faddeyeva function
%*BEN Is this a new way of writing Fadeeva? Have we been wrong all this time?
%PABLO well... yes and no... Faddeeva works, it is a "standard" way (Wikipedia). BUT, the REAL literature writes it Faddeyeva (look at the math guys citations), and the DLMF 7.2. definitions proposes the 2 ways as equivalent. YET, if you ask VLAD, in Russian is it CLEARLY Fad éy è va (lire en français), so out of respect for this great woman mathematician, and consistently with the top literature, I chose the Faddeyeva way. The reason for this "error" is the accent: in the English version of its Russian name in Wikipedia, they write: Фаддеева. But in the Russian version they write: Фадде́ева. And that accent difference is a big one in Russian, e reads è, while é reads éy, which as often got lost in English translation... 
(\ref{eq: Faddeyeva function def}) defined in 7.2.3 \cite{NIST_DLMF}, by:
    \begin{equation*}
    \begin{IEEEeqnarraybox}[][c]{rcl}
    \sqrt{\frac{\pi}{4 \tau }}\mathrm{w}\left( \frac{ x  + \mathrm{i} }{2\sqrt{\tau}}\right)   & \ = \ &  \mathrm{U}(x,\tau) + \mathrm{i}\mathrm{V}(x,\tau)
    \IEEEstrut\end{IEEEeqnarraybox}
    \label{eq:: Faddeyeva and Voigt functions}
    \end{equation*}
for $\Im\left[ \frac{x + \mathrm{i}}{2\sqrt{\tau}}\right] > 0$. In the case $\Im\left[ \frac{x + \mathrm{i}}{2\sqrt{\tau}}\right] < 0$, we use
$- \left[\mathrm{w}\left(z^*\right)\right]^*$ to calculate the integral. So that the $\psi_T(E)$ and $\chi_T(E)$ functions can approximately be related to the Faddeyeva function as:
    \begin{equation}
    \begin{IEEEeqnarraybox}[][c]{rcl}
    \psi_T(E)  & \ \simeq \ &  \sqrt{\frac{\pi}{4 \tau }}\Re\left[\mathrm{w}\left( \frac{ x  + \mathrm{i} }{2\sqrt{\tau}}\right) \right] \\
    \chi_T(E)  & \ \simeq \ &   \sqrt{\frac{\pi}{4  \tau }}\Im\left[\mathrm{w}\left( \frac{ x  + \mathrm{i} }{2\sqrt{\tau}}\right) \right] 
    \IEEEstrut\end{IEEEeqnarraybox}
    \label{eq::T K psi chi functions Faddeyeva}
    \end{equation}
So that the Doppler broadened Breit-Wigner resonance, under these approximations, can be expressed as:
\begin{equation}
\begin{IEEEeqnarraybox}[][c]{rcl}
\sigma^{\mathrm{SLBW}}_T(E) & \ \simeq \ &   \frac{1}{\sqrt{E}}\Re\left[\frac{a + \mathrm{i}b}{\mathrm{i}\Gamma_j/2}\sqrt{\frac{\pi}{4\tau}} \mathrm{w}\left( \frac{ x  + \mathrm{i} }{2\sqrt{\tau}}\right) \right] \\
& =  &  \frac{1}{E}\Re\left[\sqrt{\pi} \frac{a + \mathrm{i}b}{\mathrm{i}2\beta}\mathrm{w}\left( \frac{E - \mathcal{E}_j }{2\beta\sqrt{E}}\right)  \right]
\IEEEstrut\end{IEEEeqnarraybox}
\label{eq:: Doppler broadened single Breit Wigner resonance Faddeyeva} 
\end{equation}
This has been the traditional ``psi-chi'' method to perform approximate Doppler broadening of nuclear resonances, though some improvements have been proposed (c.f. eq. (65) in \cite{Hwang_1987}).

Note that the single-level Breit-Wigner profile (\ref{eq:: SLBW Single Breit Wigner resonance}) does not represent higher-order angular momenta for neutron channels, nor does it represent charged particles or photon channels (of any angular momenta), neither does it consider non-zero-threshold behaviors.

%%%%%%%%%%%%%%%%%%%%%%%%%%%%%%%%%%%%%%%%%%%%%%%%%%%%%%%%%%%%%%%%%%%%%%%%%%%%%%%%
%*******************************************************************************
%%%%%%%%%%%%%%%%%%%%%%%%%%%%%%%%%%%%%%%%%%%%%%%%%%%%%%%%%%%%%%%%%%%%%%%%%%%%%%%%
\subsection{\label{subsec:Analytic Doppler broadening of Windowed Multipole cross sections}Analytic Doppler broadening of \\ Windowed Multipole cross sections}
%%%%%%%%%%%%%%%%%%%%%%%%%%%%%%%%%%%%%%%%%%%%%%%%%%%%%%%%%%%%%%%%%%%%%%%%%%%%%%%%
%*******************************************************************************
%%%%%%%%%%%%%%%%%%%%%%%%%%%%%%%%%%%%%%%%%%%%%%%%%%%%%%%%%%%%%%%%%%%%%%%%%%%%%%%%

Theorem \ref{theo::WMP Representation} establishes the Windowed Multipole Representation as an equivalent formalism to parametrize R-matrix cross sections. Windowed Multipole cross sections take the form (\ref{eq:: sigma(z) Windowed Multipole representation}) for zero-threshold cross sections of any kind (photons, charged, higher angular momenta), and other thresholds can be approximated with this form (\ref{eq:: sigma(z) Windowed Multipole representation}), though not exactly. 
Theorem \ref{theo::WMP Doppler broadening} shows how these Windowed Multipole cross sections (\ref{eq:: sigma(z) Windowed Multipole representation}) can be Doppler broadened analytically to high accuracy, without having to assume an energy dependence of the Single-Level Breit-Wigner cross section form $\sigma_0^{\mathrm{SLBW}}(E)$ as in (\ref{eq:: SLBW Single Breit Wigner resonance}).

\begin{theorem}\label{theo::WMP Doppler broadening}\textsc{Doppler broadening of Windowed Multipole cross sections} \\
Consider the Windowed Multipole Representation of R-matrix cross sections (\ref{eq:: sigma(z) Windowed Multipole representation}), i.e. locally of the form:
\begin{equation*}
\begin{IEEEeqnarraybox}[][c]{C}
     \sigma(z)  \underset{\mathcal{W}(E)}{=}  \frac{1}{z^2}\Re_{\mathrm{conj}}\left[\sum_{j \geq 1} \frac{r_{j}}{z-p_j}\right] + \sum_{n\geq -2} a_n z^n
%\label{eq:: sigma(z) Windowed Multipole representation}
\IEEEstrut\end{IEEEeqnarraybox}
\end{equation*}
Upon integration against the Solbrig kernel (\ref{eq::Solbrig_Kernel}), the Doppler broadened cross section at temperature $T$ takes the following analytic expression:
\begin{equation}
\begin{IEEEeqnarraybox}[][c]{rcl}
\sigma_{T}(z)  & \ \underset{\mathcal{W}(E)}{=} \ & \sum_{n\geq -2}  a_n \mathrm{D}_\beta^n(z)  \\ & & + \frac{1}{z^2}\Re\left[  \mathrm{i} \sqrt{\pi} \sum_{j \geq 1} \frac{r_j}{\beta} \cdot \mathrm{w}\left(\frac{z-p_j}{\beta}\right) \right]  \\ 
&  & - \frac{1}{z^2}\Re\left[ \mathrm{i} \sqrt{\pi} \sum_{j \geq 1} \frac{r_j}{\beta} \cdot \mathrm{C}\left(\frac{z}{\beta}, \frac{p_j}{\beta}\right) \right] 
\IEEEstrut\end{IEEEeqnarraybox}
\label{eq::WPR T exact with C-term correction}
\end{equation}
where $\mathrm{C}\left(\frac{z}{\beta}, \frac{p_j}{\beta}\right)$ is a correction term defined as:
\begin{equation}
    \mathrm{C}\left(\frac{z}{\beta}, \frac{p_j}{\beta}\right) \triangleq \frac{2 \; p_j}{\mathrm{i} \pi \beta}  \int^\infty_{0} \frac{\mathrm{e}^{-\left(\frac{z}{\beta} + t\right)^2}}{ t^2 - \left(\frac{p_j}{\beta}\right)^2}  \mathrm{d}t
    \label{eq: C-function correction term}
\end{equation}
which is negligible in most physical ranges of temperatures and energies, so that Doppler broadened Windowed Multipole cross sections can be well approximated as
\begin{equation}
\begin{IEEEeqnarraybox}[][c]{rCl}
\sigma_{T}(z)  & \underset{\mathcal{W}(E)}{\simeq} &   \frac{1}{z^2}\Re\left[\sqrt{\pi} \sum_{j \geq 1}  \frac{r_j}{\mathrm{i}\beta} \cdot  \mathrm{w}\left(\frac{z-p_j}{\beta}\right) \right] \\
 & & + \sum_{n\geq -2}  a_n \mathrm{D}_\beta^n(z)
\IEEEstrut\end{IEEEeqnarraybox}
\label{eq::WPR T}
\end{equation}
where $\mathrm{D}_\beta^n(z)$ are the Doppler broadened monomials:
\begin{equation}
\begin{IEEEeqnarraybox}[][c]{rcl}
\mathrm{D}^{n}_\beta(z) & \ \triangleq \ &\int^\infty_0 \frac{x^{n+2}}{z^2}\mathcal{K}^{\mathbb{D}}_\beta(z,x) \mathrm{d}x 
\IEEEstrut\end{IEEEeqnarraybox}
\label{eq::Doppler broadened monomials}
\end{equation}
which are subject to the following recurrence formulae from elemental Gaussian and error functions (defined in eq. 7.2.1 of \cite{NIST_DLMF}) \cite{Josey_JCP_2016}:
\begin{equation}
\begin{IEEEeqnarraybox}[][c]{rCl}
\mathrm{D}^{n+2}_\beta(z) & \underset{\forall n\geq 1}{=} & \left[\frac{\beta^2}{2}(2n+1) + z^2 \right] \mathrm{D}^n_\beta(z) \\ 
  & & - \left(\frac{\beta^2}{2}\right)^2 n(n-1) \mathrm{D}^{n-2}_\beta(z) \\
 \mathrm{D}^0_\beta(z) & = & \left[\frac{\beta^2}{2}+ z^2\right]\mathrm{D}^{-2}_\beta(z) + \frac{\beta}{z\sqrt{\pi}} \mathrm{e}^{-\left(\frac{z}{\beta}\right)^2} \\
 \mathrm{D}^{-1}_\beta(z) & = & \frac{1}{z} \\
 \mathrm{D}^{-2}_\beta(z) & = & \frac{1}{z^2}\mathrm{erf}\left(\frac{z}{\beta}\right)
\IEEEstrut \end{IEEEeqnarraybox}
\label{eq:: recurrence formula for polynomial Doppler broadening}
\end{equation}
and where $\mathrm{w}(z)$ is the Faddeyeva function (defined in eq. 7.2.3 of \cite{NIST_DLMF}),
\begin{equation}
\begin{IEEEeqnarraybox}[][c]{rCl}
\IEEEstrut
\mathrm{w}(z) & \triangleq & \mathrm{e}^{-z^2}\Big( 1 - \mathrm{erf}\left( -\mathrm{i}z \right) \Big) = \mathrm{e}^{-z^2}\left( 1 + \frac{2\mathrm{i}}{\sqrt{\pi}} \int_0^z  e^{t^2}\mathrm{d}t \right)
\end{IEEEeqnarraybox}
\label{eq: Faddeyeva function def}
\end{equation}
called at poles in the complex lower semi-plane, i.e. $\Im\left[\frac{z-p_j}{\beta}\right] > 0$. For all other poles, which satisfy $\Im\left[\frac{z-p_j}{\beta}\right] \leq 0$, we use the fact that the Windowed Multipole Representation has complex conjugate poles to call the Faddeyeva function at $- \left[\mathrm{w}\left(z^*\right)\right]^* = - \mathrm{w}\left(-z\right)$.
\end{theorem}

\begin{proof}
This analytic Doppler broadening comes from: 
\begin{equation*}
\begin{IEEEeqnarraybox}[][c]{rCl}
\sigma_{T}(z) & \ \underset{\mathcal{W}(E)}{=} \ &  \frac{1}{z^2}\Re\left[ \sum_{j \geq 1} r_j \int^\infty_0 \frac{\mathcal{K}^{\mathbb{D}}_\beta(z,x)}{ x - p_j} \mathrm{d}x \right] \\
& & + \sum_{n\geq -2}  a_n \mathrm{D}_\beta^n(z) 
\IEEEstrut\end{IEEEeqnarraybox}
\end{equation*}
Doppler broadening of the Laurent expansion part (\ref{eq::Doppler broadened monomials}), which describes the threshold behavior, was established in \cite{Ducru_JCP_2017} (eq. (14) to (16)), and the recurrence formulae (\ref{eq:: recurrence formula for polynomial Doppler broadening}) are obtained through integration by parts. \\
The resonances Doppler broadening was established in \cite{Hwang_1987} (eq. (70) to (75)), which we here recall
\begin{equation*}
\begin{IEEEeqnarraybox}[][c]{l}
\beta \sqrt{\pi} \int^\infty_0 \frac{\mathcal{K}^{\mathbb{D}}_\beta(z,x)}{ x - p_j} \mathrm{d}x \triangleq \int^\infty_0 \mkern-8mu \frac{\mathrm{d}x}{ x - p_j} \left[ \mathrm{e}^{-\left(\frac{z-x}{\beta}\right)^2} \mkern-8mu - \mathrm{e}^{-\left(\frac{z+x}{\beta}\right)^2} \right]  \\ 
= \mkern-3mu \int^\infty_{-\infty}\mkern-3mu \frac{\mathrm{e}^{-\left(\frac{z-x}{\beta}\right)^2} }{ x - p_j}\mathrm{d}x  - \mkern-6mu \int^0_{-\infty} \frac{\mathrm{e}^{-\left(\frac{z-x}{\beta}\right)^2} }{ x - p_j}  \mathrm{d}x  - \mkern-6mu \int^\infty_{0} \frac{\mathrm{e}^{-\left(\frac{z+x}{\beta}\right)^2} }{ x - p_j}  \mathrm{d}x  \\
= \mkern-3mu \int^\infty_{-\infty} \mkern-3mu \frac{\mathrm{e}^{-t^2}}{ t - \left(\frac{z-p_j}{\beta}\right)}\mathrm{d}t  + \mkern-3mu  \int^\infty_{0}\mkern-6mu \mathrm{e}^{-\left(\frac{z+x}{\beta}\right)^2} \mkern-3mu \left[  \frac{1}{ x + p_j} \mkern-3mu - \mkern-3mu \frac{1}{ x - p_j} \right]\mkern-4mu \mathrm{d}x  \\
= \mathrm{i}\pi\; \mathrm{w}\left(\frac{z-p_j}{\beta}\right) - 2 p_j \int^\infty_{0} \frac{\mathrm{e}^{-\left(\frac{z+x}{\beta}\right)^2}}{ x^2 - p_j^2}  \mathrm{d}x 
\IEEEstrut\end{IEEEeqnarraybox}
\label{eq::WPR T derivations}
\end{equation*}
where in the last line we introduced the Faddeyeva function (\ref{eq: Faddeyeva function def}), defined in eq. 7.2.3 of \cite{NIST_DLMF}, which admits the following integral representation for $\Im \left[ z \right] > 0$:
\begin{equation}
\begin{IEEEeqnarraybox}[][c]{rCl}
\mathrm{w}(z) & \underset{\Im\left[z\right]>0}{=} &  \frac{1}{\mathrm{i}\pi} \int_{-\infty}^\infty  \frac{\mathrm{e}^{-t^2}}{t-z}\mathrm{d}t = \frac{2z}{\mathrm{i}\pi} \int_{0}^\infty  \frac{\mathrm{e}^{-t^2}}{t^2-z^2}\mathrm{d}t 
\IEEEstrut \end{IEEEeqnarraybox}
\label{eq: Faddeyeva function integral representation for Im(z)>0}
\end{equation}
In the case $\Im\left[z\right]<0$, we then use the following integral representation: 
\begin{equation}
\begin{IEEEeqnarraybox}[][c]{rCl}
-\left[\mathrm{w}(z^*)\right]^* = - \mathrm{w}(-z)& \underset{\Im\left[z\right]<0}{=} &  \frac{1}{\mathrm{i}\pi} \int_{-\infty}^\infty  \frac{\mathrm{e}^{-t^2}}{t-z}\mathrm{d}t
\IEEEstrut \end{IEEEeqnarraybox}
\label{eq: Faddeyeva function integral representation for Im(z)<0}
\end{equation}
So that, calling the Faddeyeva function directly for the poles in the complex lower semi-plane, $\Im\left[\frac{z-p_j}{\beta}\right] > 0$, while for the others we use
$- \left[\mathrm{w}\left(z^*\right)\right]^* = - \mathrm{w}\left(-z\right)$ to calculate the integral representation (the pole representation has complex conjugate poles), the Solbrig kernel Doppler broadening operation yields (\ref{eq::WPR T exact with C-term correction}).
Hwang undertook an in-depth study of the correction term $\mathrm{C}\left(\frac{z}{\beta}, \frac{p_j}{\beta}\right)$ in section IV.D of \cite{Hwang_1987}, showing it is negligible in most physical applications. 
Therefore, approximation (\ref{eq::WPR T}) is effectively faithful, in particular at high energies-to-temperature ratios $z/\beta \gg 1$.
\end{proof}

\begin{figure*}[ht!!] % replace 't' with 'b' to force it to be on the bottom
  \centering
  \subfigure[\ T = 300 Kelvin]{\includegraphics[width=0.85\textwidth]{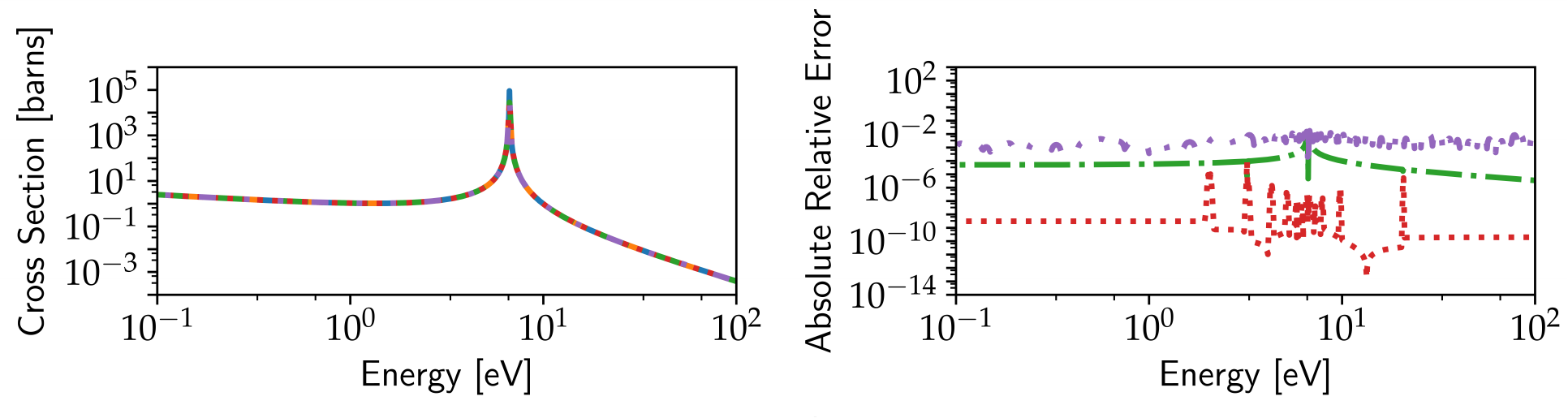}}
  \subfigure[\ T = $10^5$ Kelvin]{\includegraphics[width=0.85\textwidth]{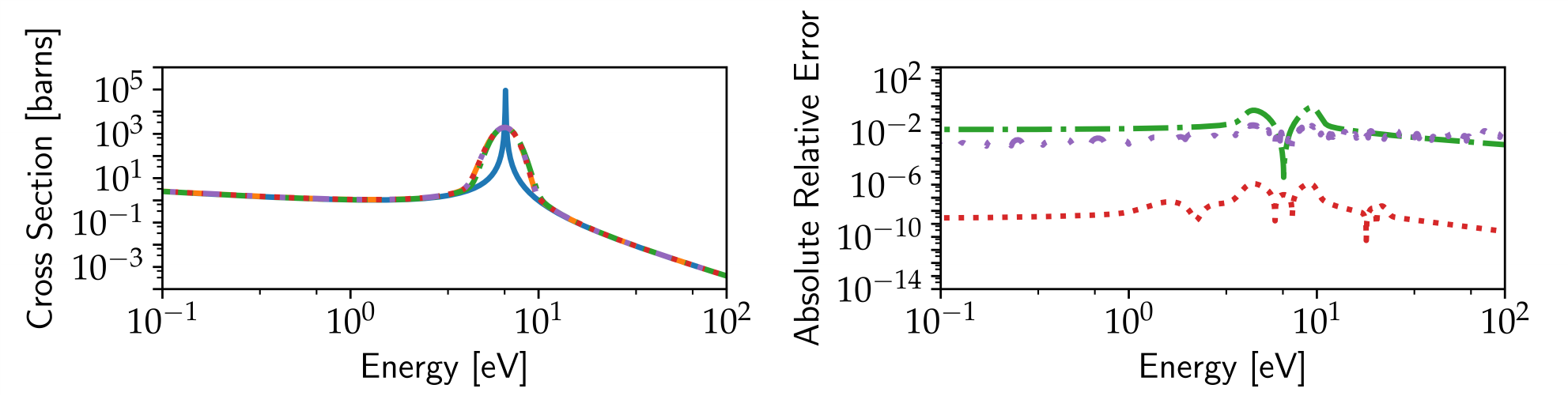}} \\
  \subfigure[\ T = $10^7$ Kelvin]{\includegraphics[width=0.85\textwidth]{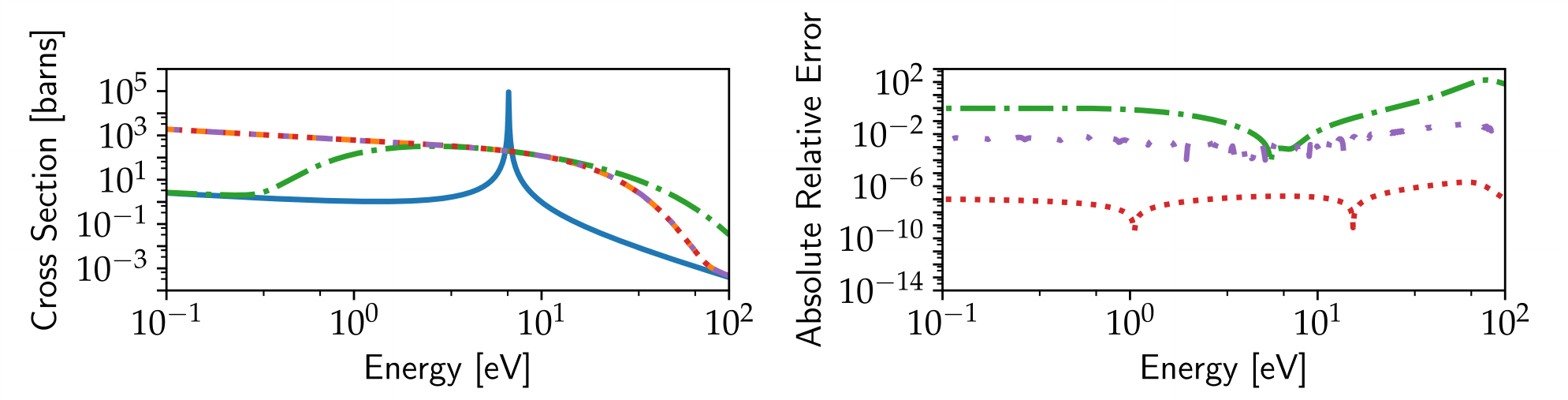}} \\
  % Legend
  \includegraphics[height=2.5cm]{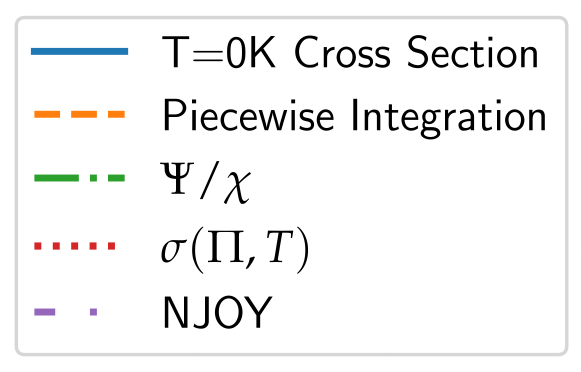}
  \caption{\small{Accuracy of different Doppler-broadening methods. Using the SLBW resonance description given in appendix \ref{appendix: Single Breit-Wigner resonance}, the cross section (\ref{eq:: caputre SLBW}) is reconstructed at T=0K. For each temperature $\{300, 10^5, 10^7\}$ Kelvin, the cross section is broadened using four different methods: (i) numerical integration of the Solbrig Kernel (\ref{eq::Solbrig_Kernel}); (ii) Using the $\psi_T/ \chi_T$ approximation (\ref{eq::T K psi chi functions Faddeyeva}) for SLBW Doppler broadening (\ref{eq:: Doppler broadened single Breit Wigner resonance Faddeyeva}); (iii) conversion (\ref{eq::caputre SLBW Real multipoles (poles and residues)}) of the resonance parameters $\left\{ \Gamma \right\}$ to multipoles $\left\{ \Pi \right\}$ and analytic Doppler broadening of Windowed Multipole Representation (\ref{eq::caputre SLBW Real multipole representation}) from theorem \ref{theo::WMP Doppler broadening} equation (\ref{eq::WPR T}); (iv) formulation of the parameters in ENDF format and processing using NJOY \cite{NJOYmacfarlaneMethodsProcessingENDF2010}. For each temperature, the right column shows the absolute relative error for methods (ii), (iii), and (iv) to the direct integration of the Solbrig Kernel (i). Note: NJOY was run with a tolerance parameter of $10^{-2}$ as higher accuracy required a prohibitively long computation time.}}
  \label{fig:Doppler boradening of WMP}
\end{figure*}

Compared to the traditional ``psi-chi'' method (\ref{eq:: Doppler broadened single Breit Wigner resonance Faddeyeva}), theorem \ref{theo::WMP Doppler broadening} gives a much more general way to Doppler broaden nuclear cross sections, applicable to charged or uncharged particles of any angular momentum.
Theorem \ref{theo::WMP Doppler broadening} also motivates why we decomposed the resonances in $z=\sqrt{E}$ space: it enables more accurate analytic Doppler broadening, since the latter happens in wavenumber space as Hwang showed in eq. (65) of \cite{Hwang_1987}.

Note that Hwang derived equations to analytically Doppler broaden his pole representation (\ref{eq: total cross section WMP Hwang neutron no threshold}), with energy-dependent residues, showing that the $\mathrm{e}^{-2\mathrm{i}\rho}$ component shifts the Faddeyeva function evaluation, adding a purely imaginary offset in eq. (6) of \cite{Hwang_2003}. 
Nonetheless, this approach is not generalizable to Coulomb channels nor to thresholds, while theorem \ref{theo::WMP Doppler broadening} is.

To compare these different Doppler broadening methods, we conducted numerical calculations on the first capture resonance of $^{\mathrm{238}}\mathrm{U}$, in the simple Single-Level Breit-Wigner resonance case of appendix \ref{appendix: Single Breit-Wigner resonance}, reporting the results in figure \ref{fig:Doppler boradening of WMP}. They show the analytic Windowed Multipole Doppler broadening exactly matches the direct piece-wise integration of Solbrig's kernel (\ref{eq::Solbrig_Kernel}) to $10^{-6}$ relative error, significantly outperforming the SIGMA1 method \cite{cullenExactDopplerBroadening1976a} of NYOJ \cite{NJOYmacfarlaneMethodsProcessingENDF2010}, while the traditional $\psi_T / \chi_T$ approximation (\ref{eq:: Doppler broadened single Breit Wigner resonance Faddeyeva}) breaks down at high temperatures. Note than in this particular SLBW case of appendix \ref{appendix: Single Breit-Wigner resonance}, the poles are exact opposites of one another, while the residues are the same, so that they cancel out of the $\mathrm{C}$-function correction (\ref{eq: C-function correction term}), hence the analytic Doppler broadening of the Windowed Multipole Representation (\ref{eq::WPR T}) is exact. This cancelling out of $\mathrm{C}$-function correction (\ref{eq: C-function correction term}) is also true in general of zero threshold neutral particles s-wave cross sections, which behave as $1/z$ at low energies, thereby yielding identical residues $r_j^+ = r_j^-$ for opposite $z$-poles pairs $p_j^+ = - p_j^-$.

%%%%%%%%%%%%%%%%%%%%%%%%%%%%%%%%%%%%%%%%%%%%%%%%%%%%%%%%%%%%%%%%%%%%%%%%%%%%%%%%
%*******************************************************************************
%%%%%%%%%%%%%%%%%%%%%%%%%%%%%%%%%%%%%%%%%%%%%%%%%%%%%%%%%%%%%%%%%%%%%%%%%%%%%%%%
\subsection{\label{subsec:Temperature derivatives of Doppler broadened Windowed Multipole cross sections} Temperature derivatives of Doppler broadened Windowed Multipole cross sections}
%%%%%%%%%%%%%%%%%%%%%%%%%%%%%%%%%%%%%%%%%%%%%%%%%%%%%%%%%%%%%%%%%%%%%%%%%%%%%%%%
%*******************************************************************************
%%%%%%%%%%%%%%%%%%%%%%%%%%%%%%%%%%%%%%%%%%%%%%%%%%%%%%%%%%%%%%%%%%%%%%%%%%%%%%%%

The analytic Doppler broadening of Windowed Multipole cross sections (theorem \ref{theo::WMP Doppler broadening}) has the additional advantage that one can compute all its temperature derivatives by means of simple recurrence formulae, as we here establish in theorem \ref{theo::WMP Doppler broadening derivatives}.

\begin{theorem}\label{theo::WMP Doppler broadening derivatives}\textsc{Temperature derivatives of Windowed Multipole cross sections} \\
Consider the approximate Doppler broadened Windowed Multipole Representation of R-matrix cross sections (\ref{eq::WPR T}) from theorem \ref{theo::WMP Doppler broadening}, upon change of variables $\theta \triangleq \frac{1}{\beta}$
\begin{equation*}
\begin{IEEEeqnarraybox}[][c]{rCl}
\sigma_{T}(z)  & \underset{\mathcal{W}(E)}{\simeq} &   \frac{1}{z^2}\Re\left[\mathrm{i}\sqrt{\pi} \sum_{j \geq 1}  r_j \cdot \theta \cdot  \mathrm{w}\Big(\theta\left(z-p_j\right)\Big) \right] \\
 & & + \sum_{n\geq -2}  a_n \mathrm{D}_\beta^n(z)
\IEEEstrut\end{IEEEeqnarraybox}
\label{eq::WPR T theta}
\end{equation*}
Then its $k$-th temperature derivative can be computed as
\begin{equation}
\begin{IEEEeqnarraybox}[][c]{rCl}
\partial^{(k)}_{T}\sigma_{T}(z)  & \underset{\mathcal{W}(E)}{\simeq} &   \frac{1}{z^2}\Re\left[\mathrm{i}\sqrt{\pi} \sum_{j \geq 1} r_j\cdot \mathrm{X}^{(k)}_\beta \left(z-p_j\right) \right] \\
 & & + \sum_{n\geq -2}  a_n \cdot  \partial^{(k)}_{T}\mathrm{D}_\beta^n(z)
\IEEEstrut\end{IEEEeqnarraybox}
\label{eq::WPR T temperature derivatives}
\end{equation}
$\mathrm{X}^{(k)}_\beta \left(z-p_j\right) $ are the $k$-th temperature derivatives of the Doppler broadened resonances:
\begin{equation}
\begin{IEEEeqnarraybox}[][c]{rcl}
\mathrm{X}^{(k)}_\beta \left(z-p_j\right) & \ \triangleq \ & \partial^{(k)}_{T}\left[ \theta \cdot  \mathrm{w}\Big(\theta\left(z-p_j\right)\Big) \right] \\
& = & \sum_{n=1 }^k \Bigg[ \left( \partial^{(n)}_{\theta} \theta \cdot \mathrm{w}\big(\theta (z-p_j)\big) \right) \quad \times \Bigg. 
\\
& & \quad \quad \Bigg.B_{k,n}\left( \theta^{(1)},  \theta^{(2)} , \hdots , \theta^{(k-n+1)} \right) \Bigg] 
\IEEEstrut\end{IEEEeqnarraybox}
\label{eq:: Temperature derivative of Doppler broadened resonances and Arbogast formula}
\end{equation}
where the sum is the Arbogast composite derivatives (Fa\`a di Bruno) formula \cite{Arbogast_1800}, linking the $\theta$-derivatives
\begin{equation}
\begin{IEEEeqnarraybox}[][c]{rCl}
\partial^{(n)}_{\theta} \theta \cdot \mathrm{w}\big(\theta (z-p_j)\big) & \underset{\forall n\geq 1}{=} & -\frac{(z-p_j)^{n-1}}{2}\mathrm{w}^{(n+1)}\big(\theta (z-p_j)\big)
\IEEEstrut\end{IEEEeqnarraybox}
\label{eq:: Faddeyeva resonances theta derivatives}
\end{equation}
to the $\theta^{(n)}$ temperature derivatives of $\theta$
\begin{equation}
\begin{IEEEeqnarraybox}[][c]{rCl}
\theta^{(n)} & \triangleq & \partial^{(n)}_{T}\theta = \frac{1}{\beta} \left( \frac{-1}{2}\right)^n \frac{\left(2n-1\right)!!}{\left( T - T_0 \right)^{n}}
\IEEEstrut\end{IEEEeqnarraybox}
\label{eq:: theta derivatives with T}
\end{equation}
by means of the partial exponential Bell polynomials $ B_{k,n}\left( \theta^{(1)},  \theta^{(2)} , \hdots , \theta^{(k-n+1)} \right)$ \cite{bellExponentialPolynomials1934, abbasNewIdentitiesBell2005, boyadzhievExponentialPolynomialsStirling2009}. \\
The derivatives of the Faddeyeva function can be computed using recurrence formulae (c.f. 7.10 in \cite{NIST_DLMF}):
\begin{equation}
\begin{IEEEeqnarraybox}[][c]{rCl}
%\mathrm{w}^{(0)}(z) & \triangleq & \mathrm{w}(z) \\
\mathrm{w}^{(1)}(z) & = & - 2 z \mathrm{w}(z) + \frac{2\mathrm{i}}{\sqrt{\pi}} \\
\mathrm{w}^{(n+2)}(z)  & = & - 2z \mathrm{w}^{(n+1)}(z) - 2 (n+1) \mathrm{w}^{(n)}(z)
\IEEEstrut \end{IEEEeqnarraybox}
\label{eq: Faddeyeva function derivatives recurrence}
\end{equation}

$\partial^{(k)}_{T}\mathrm{D}_\beta^n(z)$ are the temperature derivatives of the Doppler broadened monomials, which are subject to the following recurrence formulae, defining $a \triangleq \frac{k_\mathbb{B}}{A}$:
\begin{equation}
\begin{IEEEeqnarraybox}[][c]{rCl}
\partial^{(k)}_{T}\mathrm{D}^{n+2}_\beta(z) & \underset{\forall n\geq 1}{=} & \left[\frac{\beta^2}{2}(2n+1) + z^2 \right] \partial^{(k)}_{T} \mathrm{D}^n_\beta(z)  \\
 & & +  \frac{a}{2}(2n+1) k \; \partial^{(k-1)}_{T} \mathrm{D}^n_\beta(z)  \\ 
  & & - \frac{n(n-1)}{4}\Bigg[ \beta^4 \partial^{(k)}_{T}\mathrm{D}^{n-2}_\beta(z)  \\ & &  \quad \quad \quad + \quad 2a\beta k \;  \partial^{(k-1)}_{T}\mathrm{D}^{n-2}_\beta(z) \\ & &  \quad \quad \quad + \quad  a^2k(k-1)\partial^{(k-2)}_{T}\mathrm{D}^{n-2}_\beta(z)  \Bigg]\\
 \partial^{(k)}_{T}\mathrm{D}^0_\beta(z) & = & \mkern-10mu \left[\frac{\beta^2}{2} \mkern-3mu + \mkern-3mu z^2\right]\mkern-4mu\partial^{(k)}_{T}\mathrm{D}^{-2}_\beta \mkern-2mu (z) + \mkern-3mu \frac{a}{2}k \; \partial^{(k-1)}_{T}\mathrm{D}^{-2}_\beta \mkern-2mu (z) \\ & & \mkern-10mu +\mkern-1mu \frac{1}{z\sqrt{\pi}}\mkern-6mu\left[ \mkern-1mu \beta^2 \partial^{(k)}_{T} \mkern-3mu \theta \mathrm{e}^{-\left(z \theta \right)^2} \mkern-6mu + \mkern-2mu a k \, \partial^{(k-1)}_{T} \mkern-3mu \theta \mathrm{e}^{-\left(z \theta \right)^2} \mkern-2mu \right] \\
\partial^{(k)}_{T} \mathrm{D}^{-1}_\beta(z) & = & \frac{1}{z} \delta_{k,0}\\
 \partial^{(k)}_{T}\mathrm{D}^{-2}_\beta(z) & = & \frac{1}{z^2}\partial^{(k)}_{T}\mathrm{erf}\left(z\theta\right)
\IEEEstrut \end{IEEEeqnarraybox}
\label{eq:: recurrence formula for temperature derivatives of Doppler boradened monomials}
\end{equation}
In recurrence relations (\ref{eq:: recurrence formula for temperature derivatives of Doppler boradened monomials}), the terms $\partial^{(k)}_{T} \theta \mathrm{e}^{-\left(z \theta \right)^2}$ can themselves be computed using Arbogast's formula: 
\begin{equation}
\begin{IEEEeqnarraybox}[][c]{rCl}
\partial^{(k)}_{T} \theta \mathrm{e}^{-\left(z \theta \right)^2} & = & \mathrm{e}^{-\left(z \theta \right)^2} \sum_{n=1 }^k \Bigg[ F_z^{(n)}(\theta) \quad \times \Bigg. \\
& & \quad \quad \Bigg.B_{k,n}\left( \theta^{(1)},  \theta^{(2)} , \hdots , \theta^{(k-n+1)} \right) \Bigg]
%\sum_{n=1 }^k F_z^{(n)}(\theta) \cdot  \Bigg.B_{k,n}\left( \theta^{(1)},  \theta^{(2)} , \hdots , \theta^{(k-n+1)} \right)
\IEEEstrut \end{IEEEeqnarraybox}
\label{eq:: Arbogast formula for temperature derivatives of Doppler boradened Gaussian}
\end{equation}
where $F_z^{(n)}(\theta) $ are polynomials of degree $n+1$ defined as 
\begin{equation}
\begin{IEEEeqnarraybox}[][c]{rCl}
F_z^{(n)}(\theta) & \triangleq &  \mathrm{e}^{\left(z \theta \right)^2} \partial^{(n)}_{\theta} \theta \mathrm{e}^{-\left(z \theta \right)^2}  = \sum_{i=0}^{n+1}\alpha_i^{(n)} \theta^i 
\IEEEstrut\end{IEEEeqnarraybox}
\label{eq:: F polynomial definitions}
\end{equation}
which are recursively constructed from $F_z^{(0)}(\theta) = \theta$ as 
\begin{equation}
\begin{IEEEeqnarraybox}[][c]{rCl}
F_z^{(n+1)}(\theta) & = & \partial_{\theta}  F_z^{(n)}(\theta)  - 2z^2\theta F_z^{(n)}(\theta) 
\IEEEstrut\end{IEEEeqnarraybox}
\label{eq:: F polynomial recursive derivative construction}
\end{equation}
entailing these recurrence formulae on their coefficients:
\begin{equation}
\begin{IEEEeqnarraybox}[][c]{rClrCl}
\alpha_0^{(0)} & =  & 0  & \alpha_1^{(0)} & =  & 1\\
\alpha_{n+1}^{(n+1)} & =  & -2z^2 \alpha_{n}^{(n)} & \alpha_{n+2}^{(n+1)} & =  & -2z^2 \alpha_{n+1}^{(n)} \\
\alpha_{i}^{(n+1)} & \underset{1 \leq i \leq n}{=}  & (i+1)\alpha_{i+1}^{(n)} -2z^2 \alpha_{i-1}^{(n)}
\IEEEstrut \end{IEEEeqnarraybox}
\label{eq:: recurrence formula for the F coefficients}
\end{equation}

Finally, the terms $\partial^{(k)}_{T}\mathrm{erf}\left(z\theta\right)$ in recurrence relations (\ref{eq:: recurrence formula for temperature derivatives of Doppler boradened monomials}) can also be computed using Arbogast's formula: 
\begin{equation}
\begin{IEEEeqnarraybox}[][c]{rCl}
\partial^{(k)}_{T} \mkern-2mu \mathrm{erf}\left(z\theta\right) & = & \mkern-3mu \sum_{n=1 }^k \mkern-3mu \left( \partial^{(n)}_{\theta}\mathrm{erf}\left(z\theta\right)\right) \mkern-1mu \cdot \mkern-1mu B_{k,n}\mkern-4mu\left(\mkern-1mu \theta^{(1)} \mkern-3mu , \hdots , \theta^{(k-n+1)} \mkern-1mu \right)
\IEEEstrut\end{IEEEeqnarraybox}
\label{eq:: Arbogast composite derivatives Faa di Bruno formula for error function term}
\end{equation}
in which the $\theta$ derivatives can be expressed as
\begin{equation}
\begin{IEEEeqnarraybox}[][c]{rCl}
\partial^{(n)}_{\theta}\mathrm{erf}\left(z\theta\right) & \underset{n\geq 1}{=} & z^n (-1)^{n-1}\frac{2}{\sqrt{\pi}}\mathrm{H}_{n-1}\left( z\theta \right) \mathrm{e}^{-\left(z\theta\right)^2}
\IEEEstrut\end{IEEEeqnarraybox}
\label{eq:: theta derivatives of error function as function of Hermite polynomials}
\end{equation}
where the Hermite polynomials $\mathrm{H}_n(z)$ are recursively calculable from $\mathrm{H}_{0} = 1$ and $\mathrm{H}_{1} =  2z$ as:
\begin{equation}
\begin{IEEEeqnarraybox}[][c]{rCl}
\mathrm{H}_{n+1} & \underset{n\geq 1}{=} & 2z\mathrm{H}_{n} - 2n \mathrm{H}_{n-1}
\IEEEstrut\end{IEEEeqnarraybox}
\label{eq:: Hermite polynomial recurrence}
\end{equation}

\end{theorem}

\begin{proof}
The underlying assumption of the proof is that one can neglect the derivatives of the correction term (\ref{eq: C-function correction term}).
The proof consists of a series of derivatives expanded using the general Leibniz rule and the Arbogast formula for composite derivatives (Fa\`a di Bruno) (c.f. p.43 of \cite{Arbogast_1800}), in which the Bell polynomials can be computed as referenced in \cite{bellExponentialPolynomials1934, abbasNewIdentitiesBell2005, boyadzhievExponentialPolynomialsStirling2009}.
Direct differentiation yields the temperature derivatives of $\theta$ (\ref{eq:: theta derivatives with T}).
Expression (\ref{eq:: Faddeyeva resonances theta derivatives}) is obtained using the Faddeyeva function recurrence formula (\ref{eq: Faddeyeva function derivatives recurrence}), documented in 7.10 of \cite{NIST_DLMF}. 
The $F_z^{(n)}(\theta)$ polynomials (\ref{eq:: F polynomial definitions}) are defined from $\partial^{(n)}_{\theta} \theta \mathrm{e}^{-\left(z \theta \right)^2}  = F_z^{(n)}(\theta)  \mathrm{e}^{-\left(z \theta \right)^2}$ and their degree $n+1$ stems from the recursive derivatives (\ref{eq:: F polynomial recursive derivative construction}) initialized at  $F_z^{(0)}(\theta) = \theta$, entailing the recurrence formula for the coefficients (\ref{eq:: recurrence formula for the F coefficients}).
Similarly, expression (\ref{eq:: theta derivatives of error function as function of Hermite polynomials}) is derived from change of variable $z \to \theta z$, and using the derivative formula for the error function (c.f. Abramowitz \& Stegun, p.298, eq. 7.1.19 \cite{Abramowitz_and_Stegun}, or 7.10.1 in \cite{NIST_DLMF}):
\begin{equation*}
\begin{IEEEeqnarraybox}[][c]{rCl}
\IEEEstrut
 \mathrm{erf}^{(n+1)}(z) & = & (-1)^n \frac{2}{\sqrt{\pi}} \mathrm{H}_n(z) \mathrm{e}^{-z^2}
\end{IEEEeqnarraybox}
\end{equation*}
while the Hermite polynomials recurrence relation (\ref{eq:: Hermite polynomial recurrence}) is well known and documented (c.f. 18.9 of \cite{NIST_DLMF}).
\end{proof}

Underpinning this direct differentiation approach is the assumption that the $\mathrm{C}$-function correction term (\ref{eq: C-function correction term}), itself negligible, also has negligible temperature derivatives. It is nonetheless possible to extend this method to explicitly include thermal derivatives of the correction term (\ref{eq: C-function correction term}), by noticing that these derivatives follow a similar polynomial structure as (\ref{eq:: F polynomial definitions}) and are subject to a recurrence relation similar to (\ref{eq:: F polynomial recursive derivative construction}).
% PABLO: IT IS POSSIBLE TO GET A RECCURENCE FORMULA TO STUDY C TOO THIS WITH RECURENCE FORMULAE. 

%%%%%%%%%%%%%%%%%%%%%%%%%%%%%%%%%%%%%%%%%%%%%%%%%%%%%%%%%%%%%%%%%%%%%%%%%%%%%%%%
%*******************************************************************************
%%%%%%%%%%%%%%%%%%%%%%%%%%%%%%%%%%%%%%%%%%%%%%%%%%%%%%%%%%%%%%%%%%%%%%%%%%%%%%%%
\subsection{\label{subsec: Fourier transform approach to temperature treatment} Fourier transform approach to \\ temperature treatment}
%%%%%%%%%%%%%%%%%%%%%%%%%%%%%%%%%%%%%%%%%%%%%%%%%%%%%%%%%%%%%%%%%%%%%%%%%%%%%%%%
%*******************************************************************************
%%%%%%%%%%%%%%%%%%%%%%%%%%%%%%%%%%%%%%%%%%%%%%%%%%%%%%%%%%%%%%%%%%%%%%%%%%%%%%%%

Ferran developed a more general approach, based on Fourier transforms, to Doppler broaden nuclear nuclear cross sections (we here only discussed Doppler broadening of angle-integrated cross sections)\cite{ferranNewMethodDoppler2015a}.
In theorem \ref{theo::WMP Fourier transform Doppler broadening of Windowed Multipole cross sections}, we generalize Ferran's method, begetting arbitrary-order temperature derivatives of Doppler broadened cross sections, while setting a more general framework for temperature treatments such as low-energy thermal neutrons scattering with the phonons of the target's crystalline structure. Moreover, when applied to the Windowed Multipole Representation of R-matrix cross sections, this Fourier transform approach exactly accounts for the entire nuclear cross section, without neglecting the $\mathrm{C}$-function correction term (\ref{eq: C-function correction term}). This generality comes at the additional cost of having to compute Fourier transforms on-the-fly. Also, Fourier transforms can be numerically sensitive to the tails of distributions, meaning one has to be careful as to how the cross sections are extended beyond the treated windows (c.f. Ferran's discussion in section IV.B.2 of \cite{ferranNewMethodDoppler2015a}).

We here recall Ferran's general Fourier transform method from \cite{ferranNewMethodDoppler2015a}.
The function $ f \star g $ designates the convolution product between functions $f$ and $g$, defined as:
\begin{equation}
\begin{IEEEeqnarraybox}[][c]{lrl}
\IEEEstrut
 f \star g \; (x) \underset{\forall x \in \mathbb{R}}{\triangleq} \int_{\mathbb{R}} f(t)g(x-t)\mathrm{d}t
\end{IEEEeqnarraybox}
\label{def: convolution product}
\end{equation}
Ferran expressed Solbrig's kernel (\ref{eq::Solbrig_Kernel}) Doppler broadening operation as a convolution product by introducing the \textit{Ferran representation} odd-parity function \cite{ferranNewMethodDoppler2015a}:
\begin{equation}
s_T : z \in \mathbb{R} \mapsto \left\{
\begin{array}{rl}
  z^2 \sigma_T(z) &  \forall z \in \mathbb{R}^*_+ \\
  0 & \text{if  }  z=0 \\
  - z^2 \sigma_T(-z) & \forall z \in \mathbb{R}^*_-
  \end{array}
\right.
\label{defining the generalized cross section}
\end{equation}
Applying Solbrig's Kernel to $s_T$ yields a linear convolution product operator that transforms the Ferran representation $s_0$  of the cross section at temperature $T_0$, to $s_T$ at temperature $T > T_0$ as follows \cite{ferranNewMethodDoppler2015a}:
\begin{equation}
\begin{IEEEeqnarraybox}[][c]{lrl}
s_T = s_0 \star \mathcal{K}^\mathbb{B}_T
\end{IEEEeqnarraybox} \IEEEstrut
\label{Generalized cross section convolution product}
\end{equation}
where $\mathcal{K}^\mathbb{B}_T$ is the Maxwell-Bolztmann distribution of energies of the target
  \begin{equation}
\mathcal{K}^\mathbb{B}_T(z) \underset{\forall z \in \mathbb{R}}{\triangleq} \frac{1}{\beta\sqrt{\pi}}\mathrm{e}^{-\left(\frac{z}{\beta}\right)^2}
\label{eq: Maxwell-Bolztmann distibution of energies}
\end{equation}
The Fourier transform of a function $f$ is defined as (unitary, ordinary frequency convention)\cite{fourierTheorieAnalytiqueChaleur1822}:
\begin{equation}
\begin{IEEEeqnarraybox}[][c]{C}
\widehat{f}(\nu) \triangleq  \int_{\mathbb{R}} f(t)\cdot e^{-\mathrm{i}2\pi\nu t}\mathrm{d}t
\IEEEstrut \end{IEEEeqnarraybox}
\label{eq: Fourier transform}
\end{equation}
for which the inverse Fourier transform is: 
\begin{equation}
\begin{IEEEeqnarraybox}[][c]{C}
\IEEEstrut
f(x) = \int_{\mathbb{R}} \widehat{f}(\nu)\cdot e^{i2\pi\nu x}\mathrm{d}\nu
\IEEEstrut\end{IEEEeqnarraybox}
\label{eq: inverse fourier}
\end{equation}
The Fourier transform of any odd-parity function $g$ can be expressed as
\begin{equation}
\begin{IEEEeqnarraybox}[][c]{C}
\widehat{g}(\nu) =  - 2 i   \int_{\mathbb{R}_+} g(t) \sin \left(2 \pi \nu t \right) \mathrm{d}t
\IEEEstrut
\end{IEEEeqnarraybox}
\label{eq: odd-parity Fourier transform}
\end{equation}
Fourier transforms satisfy the convolution property:
\begin{equation}
\begin{IEEEeqnarraybox}[][c]{C}
\widehat{f \star g} = \widehat{f} \cdot \widehat{g}
\IEEEstrut
\end{IEEEeqnarraybox}
\label{eq: Fourier transform convolution}
\end{equation}
The Doppler broadening operation can therefore be performed by calculating the inverse Fourier transform of
\begin{equation}
\begin{IEEEeqnarraybox}[][c]{C}
\widehat{s_T} = \widehat{s_0 \star \mathcal{K}^\mathbb{B}_T} = \widehat{s_0} \cdot \widehat{\mathcal{K}^\mathbb{B}_T}
\IEEEstrut\end{IEEEeqnarraybox}
\label{eq: Ferran Fourier transform method Doppler boradening}
\end{equation}
Since the Fourier transform of Boltzmann kernel $\mathcal{K}^\mathbb{B}_T$ is well-known
\begin{equation}
\begin{IEEEeqnarraybox}[][c]{C}
\widehat{\mathcal{K}^\mathbb{B}_T}(\nu) = \mathrm{e}^{- \left(\pi \beta \nu\right)^2}
\IEEEstrut
\end{IEEEeqnarraybox}
\label{eq: Fourier of Boltzmann}
\end{equation}
given $\widehat{s_0}$, Doppler broadening can therefore be performed as the inverse Fourier transform of $\widehat{s_0}\cdot \mathrm{e}^{- \left(\pi \beta \nu\right)^2}$.

In theorem \ref{theo::WMP Fourier transform Doppler broadening of Windowed Multipole cross sections}, we derive the Fourier transform of windowed multipole cross sections, and generalize Ferran's method to account for arbitrary order temperature derivatives, as an alternative to theorem \ref{theo::WMP Doppler broadening derivatives}.

\begin{theorem}\label{theo::WMP Fourier transform Doppler broadening of Windowed Multipole cross sections}\textsc{Fourier transform Doppler broadening of Windowed Multipole cross sections} \\
Consider the zero Kelvin ($0$ $\mathrm{K} $) Ferran representation of Windowed Multipole R-matrix cross sections (\ref{eq:: sigma(z) Windowed Multipole representation}), i.e. the odd-parity function $s_0(z) = - s_0(-z)$ locally of the form:
\begin{equation}
\begin{IEEEeqnarraybox}[][c]{C}
     s_0(z) \underset{z>0}{\triangleq} z^2\cdot \sigma_0(z)  \underset{\mathcal{W}(E)}{=} \Re_{\mathrm{conj}}\mkern-6mu\left[\sum_{j \geq 1} \frac{r_{j}}{z-p_j}\right] \mkern-3mu + \mkern-3mu \sum_{n\geq 0} a_{n-2} z^n
\label{eq:: Ferran z^2 sigma(z) Windowed Multipole representation}
\IEEEstrut\end{IEEEeqnarraybox}
\end{equation}
Then its Fourier transform (\ref{eq: Fourier transform}) can be expressed as
\begin{equation}
\begin{IEEEeqnarraybox}[][c]{C}
     \widehat{s_0}(\nu)  \underset{\mathcal{W}(E)}{=} \Re_{\mathrm{conj}}\mkern-6mu\left[\sum_{j \geq 1} r_{j} \widehat{\mathrm{V}}_{p_j}(\nu) \right] + \sum_{n\geq 0} a_{n-2} \widehat{\mathrm{F}}_{n}(\nu)
\label{eq:: Ferran fourier transform of Windowed Multipole representation}
\IEEEstrut\end{IEEEeqnarraybox}
\end{equation}
where the Fourier transforms of the Laurent expansions $\widehat{\mathrm{F}}_{n}(\nu)$ can be expressed for either even or odd positive integers $n \geq 0$ as ($\delta^{(n)}(\nu)$ designates the $n$-th derivative of Dirac's Delta distribution):
\begin{equation}
\begin{IEEEeqnarraybox}[][c]{rCl}
     \widehat{\mathrm{F}}_{2n}(\nu) & \triangleq &  (-1)^{n+1}\frac{2\mathrm{i} (2n)!}{(2\pi\nu)^{2n+1}}\\
     \widehat{\mathrm{F}}_{2n+1}(\nu) & \triangleq &  \frac{(-1)^{n}\mathrm{i}}{(2\pi)^{2n+1}} \delta^{(n)}(\nu)
\label{eq:: Ferran fourier transform of Windowed Multipole representation Lauren term}
\IEEEstrut\end{IEEEeqnarraybox}
\end{equation}
and the Fourier transforms of the resonances at pole $p_j$  $\widehat{\mathrm{V}}_{p_j}(\nu)$ can be expressed as
\begin{equation}
\begin{IEEEeqnarraybox}[][c]{rCl}
    \widehat{\mathrm{V}}_{p_j}(\nu) & \underset{\left|\mathrm{ph}(p_j)\right| < \pi}{\triangleq} & - 2 \mathrm{i} \cdot \mathrm{sgn}(\nu) \cdot  f\left(- 2\pi |\nu| p_j \right)
\label{eq:: Ferran fourier transform of Windowed Multipole representation Resonance term}
\IEEEstrut\end{IEEEeqnarraybox}
\end{equation}
where $\mathrm{sgn}(z)$ designates the sign function, and $f$ is the auxiliary function defined in 6.2.17 of \cite{NIST_DLMF}.

The $k^{\text{th}}$-order temperature derivative of Windowed Multipole R-matrix cross sections is the convolution:
\begin{equation}
\begin{IEEEeqnarraybox}[][c]{lrl}
\IEEEstrut
\partial_T^{(k)}s_T = s_0 \star \partial_T^{(k)} \mathcal{K}^\mathbb{B}_T
\label{eq: high-order temperature derivative Fourier method Ferran}
\end{IEEEeqnarraybox}
\end{equation}
which is the inverse Fourier transform (\ref{eq: inverse fourier}) of product 
\begin{equation}
\begin{IEEEeqnarraybox}[][c]{lrl}
\IEEEstrut
\widehat{\partial_T^{(k)}s_T} = \widehat{s_0} \cdot \partial_T^{(k)} \widehat{\mathcal{K}^\mathbb{B}_T}
\label{eq: Fourier transform of Ferran high-order temperature derivative}
\end{IEEEeqnarraybox}
\end{equation}
whose expressions are (\ref{eq:: Ferran fourier transform of Windowed Multipole representation}) for $\widehat{s_0}$ and, defining $a \triangleq \frac{k_\mathbb{B}}{A}$, 
\begin{equation}
\begin{IEEEeqnarraybox}[][c]{lrl}
\IEEEstrut
\partial_T^{(k)} \widehat{\mathcal{K}^\mathbb{B}_T} = a^k (i \pi  \nu )^{2k} \widehat{\mathcal{K}^\mathbb{B}_T}  = a^k ( i \pi \nu )^{2k}  \mathrm{e}^{- \left(\pi \beta \nu\right)^2}
\label{eq: high-order temperature derivative of Boltzmann kernel}
\end{IEEEeqnarraybox}
\end{equation}

\end{theorem}

\begin{proof}
The proof consists of directly calculating the corresponding Fourier transforms by developing the linear operators.
Equation (\ref{eq: Fourier transform of Ferran high-order temperature derivative}) stems from the Fourier transform linear property $\widehat{\partial_T^{(k)}\mathcal{K}^\mathbb{B}_T} = \partial_T^{(k)} \widehat{\mathcal{K}^\mathbb{B}_T}$ applied to (\ref{eq: high-order temperature derivative Fourier method Ferran}). Expression (\ref{eq: high-order temperature derivative of Boltzmann kernel}) is obtained by direct differentiation of (\ref{eq: Fourier of Boltzmann}).
In key expression (\ref{eq:: Ferran fourier transform of Windowed Multipole representation}), the Fourier transforms of the Laurent development part (\ref{eq:: Ferran fourier transform of Windowed Multipole representation Lauren term}) are obtained by noticing that odd parity polynomials are already odd functions, while the even parity ones must be written as the difference of (\ref{eq:: Ferran z^2 sigma(z) Windowed Multipole representation}) multiplied by the Heaviside function for domains $\mathbb{R}_-$ and $\mathbb{R}_+$, and then applying standard Fourier transform properties. 
The Fourier transforms of resonance terms (\ref{eq:: Ferran fourier transform of Windowed Multipole representation Resonance term}) are obtained by identifying the integral representation 6.7.13 in \cite{NIST_DLMF}, and using identity $f\left( z \mathrm{e}^{\pm \mathrm{i}\pi}\right) = \pi \mathrm{e}^{\mp \mathrm{i}z} - f(z) $ (c.f. 6.4.6 \cite{NIST_DLMF}) if the phase of the pole $p_j$ does not respect $\left| \mathrm{ph}\left(p_j \right)\right| < \pi$. 
\end{proof}

The Fourier transform approach of theorem \ref{theo::WMP Fourier transform Doppler broadening of Windowed Multipole cross sections} to arbitrary order temperature derivatives is conceptually more elegant than the direct differentiations of theorem \ref{theo::WMP Doppler broadening derivatives}: there is no need for Arbogast - Fa\`a di Bruno composition expansions nor recurrences. It is also more general, as the correction $\mathrm{C}$-function term (\ref{eq: C-function correction term}) is not neglected in the Doppler broadening, and that the Fourier transform approach could potentially be expanded to treat thermal scattering with the phonon distributions of targets: one would then need to replace the Boltzmann distributions $\widehat{\mathcal{K}^\mathbb{B}_T}$ (\ref{eq: Fourier of Boltzmann}) with the corresponding phonon Fourier spectra (c.f. ``Neutron Slowing Down and Thermalization'' chapter in \cite{cacuciHandbookNuclearEngineering2010} or \cite{ballingerDirectAlphaBeta1995, vineyardScatteringSlowNeutrons1958}).
In practice, theorem \ref{theo::WMP Fourier transform Doppler broadening of Windowed Multipole cross sections} also runs into its own hurdles: nothing guarantees that numerically performing the on-the-fly Fourier transforms of theorem \ref{theo::WMP Fourier transform Doppler broadening of Windowed Multipole cross sections} -- using the Fast Fourier Transform FFT and subsequent algorithms \cite{cooleyAlgorithmMachineCalculation1965, baileyFastMethodNumerical1994, inverarityFastComputationMultidimensional2002} -- is more computationally efficient than calling the Faddeyeva functions -- which also have benefited of great algorithmic and computational performance gains \cite{weidemanComputationComplexError1994, poppeMoreEfficientComputation1990, zaghloulAlgorithm916Computing2012, Faddeyeva_2017, website_faddeevamit} -- and the recursive formulae of theorem \ref{theo::WMP Doppler broadening derivatives}. This is all the more so true than theorem \ref{theo::WMP Fourier transform Doppler broadening of Windowed Multipole cross sections} requires the computation of the $f$ auxiliary function (\ref{eq:: Ferran fourier transform of Windowed Multipole representation Resonance term}), which could be more costly than calling the Faddeyava function.
Also, Fourier transforms are global integrals, so the windowing process complicates this approach, and the windows have now to be selected according to the method Ferran discussed in IV.B.2 of \cite{ferranNewMethodDoppler2015a}, considering that the Doppler broadening only affects the cross section $\sigma(E)$ at a given energy $E$ for a convolution over an interval commensurate to the temperature energy $\beta^2$, say four times $E \pm 4 \beta^2$ \cite{solbrig1961doppler, cullenExactDopplerBroadening1976a, Hwang_1987, ferranNewMethodDoppler2015a}. 
Note that this locality problem already exists in the direct Doppler broadening of theorem \ref{theo::WMP Doppler broadening} and by extension theorem \ref{theo::WMP Doppler broadening derivatives}, and even in the windowing process itself, when selecting which poles $p_j$ to include in window $\mathcal{W}(E)$ as discussed in section \ref{subsubsec:Windowing process: Laurent background fit} and established in \cite{Forget_2013, Josey_JCP_2016}. Though in theory the Mittag-Leffler expansion converges on the entire energy domain between two thresholds $\left[ E_{T_c} , E_{T_c+1} \right]$, in practice it is too costly to compute the Faddeyeva functions for all poles, the essence of the windowing process is therefore to only account for the poles which affect the cross section in window $\mathcal{W}(E)$ upon Doppler broadening, in practice extending the domain (``external window'' in \cite{Forget_2013, Josey_JCP_2016}) for a couple of temperature energy variances in the Boltzmann distribution ($\left[\mathcal{W}(E_{\mathrm{min}} - 4 \beta_{\mathrm{max}(T)}^2) , \mathcal{W}(E_{\mathrm{max}} + 4 \beta_{\mathrm{max}(T)}^2) \right]$): this is a very similar process than Ferran's continuation of the function for the Fourier transform, discussed in section IV.B.2 of \cite{ferranNewMethodDoppler2015a}. 
Therefore, if the windowing process is well performed, the expression of theorem \ref{theo::WMP Fourier transform Doppler broadening of Windowed Multipole cross sections} will be valid within each window. Otherwise, one would need to truncate the Fourier transforms at the boundary of each energy window, and laboriously concatenate the Ferran representation window by window in the Fourier transforms. 

Finally, note that Ferran's Doppler broadening method presents similarities with the optimal temperature kernel reconstruction quadratures developed in \cite{Ducru_JCP_2017}: both are kernel methods operating on the cross sections, in particular the Boltzmann kernel eq. (6) of \cite{Ducru_JCP_2017}. 
Appendix D of \cite{Ducru_JCP_2017} studies the consequences the Windowed Multipole Representation of R-matrix cross sections on the Fourier transforms involved in theorem \ref{theo::WMP Fourier transform Doppler broadening of Windowed Multipole cross sections}. In particular are discussed the general shapes of the Fourier transforms of the nuclear resonances, compared to the $\mathcal{K}^\mathbb{B}_T$ Boltzmann kernel (\ref{eq: Fourier of Boltzmann}), and how this can entail properties of interest, such as frequency separation in $\mathrm{L_2}$ norm (c.f. eq. (D.9) and sections D.2 and D.3 of appendix D in \cite{Ducru_JCP_2017}).

%%%%%%%%%%%%%%%%%%%%%%%%%%%%%%%%%%%%%%%%%%%%%%%%%%%%%%%%%%%%%%%%%%%%%%%%%%%%%%%%
%*******************************************************************************
%%%%%%%%%%%%%%%%%%%%%%%%%%%%%%%%%%%%%%%%%%%%%%%%%%%%%%%%%%%%%%%%%%%%%%%%%%%%%%%%
\section{\label{sec:Conclusion}Conclusion}
%%%%%%%%%%%%%%%%%%%%%%%%%%%%%%%%%%%%%%%%%%%%%%%%%%%%%%%%%%%%%%%%%%%%%%%%%%%%%%%%
%*******************************************************************************
%%%%%%%%%%%%%%%%%%%%%%%%%%%%%%%%%%%%%%%%%%%%%%%%%%%%%%%%%%%%%%%%%%%%%%%%%%%%%%%%

This article establishes the theoretical foundations for the Windowed Multipole Library.

We derive how the Windowed Multipole Representation of R-matrix cross sections can be constructed by finding the poles of the Kapur-Peierls operator and performing Hwang's albebraic continuation (theorem \ref{theo::WMP Representation}).
In the process, we connect the Windowed Multipole Representation to both the Bloch and Wigner-Eisenbud R-matrix theory and to the Humblet-Rosenfeld pole expansions in wavenumber space.

We establish a method to convert R-matrix resonance parameters covariance matrices into Windowed Multipole covariances (theorem \ref{theo::WMP covariance}), and show they generate the same uncertainty distribution on nuclear cross sections, either through the sensitivity approach or by sampling stochastic cross sections.

We recall Windowed Multipole cross sections can be Doppler broadened analytically to high accuracy (theorem \ref{theo::WMP Doppler broadening}), and expand this on-the-fly capability to arbitrary-order temperature derivatives (theorem \ref{theo::WMP Doppler broadening derivatives}), whist deriving new capabilities for temperature treatment by means of Fourier transforms of Windowed Multipole cross sections (theorem \ref{theo::WMP Fourier transform Doppler broadening of Windowed Multipole cross sections}).

The Windowed Multipole Representation of R-matrix cross sections has already proved its efficacy on a vast range of nuclear physics applications.
We hope the foundational results of this article will allow for the widespread adoption of the Windowed Multipole Library, and underpin new research efforts to expand its capabilities.

%%%%%%%%%%%%%%%%%%%%%%%%%%%%%%%%%%%%%%%%%%%%%%%%%%%%%%%%%%%%%%%%%%%%%%%%%%%%%%%%
%*******************************************************************************
%%%%%%%%%%%%%%%%%%%%%%%%%%%%%%%%%%%%%%%%%%%%%%%%%%%%%%%%%%%%%%%%%%%%%%%%%%%%%%%%
\begin{acknowledgments}
%%%%%%%%%%%%%%%%%%%%%%%%%%%%%%%%%%%%%%%%%%%%%%%%%%%%%%%%%%%%%%%%%%%%%%%%%%%%%%%%
%*******************************************************************************
%%%%%%%%%%%%%%%%%%%%%%%%%%%%%%%%%%%%%%%%%%%%%%%%%%%%%%%%%%%%%%%%%%%%%%%%%%%%%%%%

This work was partly funded by the Consortium for Advanced Simulation of Light Water Reactors (CASL), an Energy Innovation Hub for Modeling and Simulation of Nuclear Reactors under U.S. Department of Energy Contract No. DE-AC05-00OR22725. 

The first author was also partly funded by: the 2019-2020 AXA Fellowship of the Schwarzman Scholars Program at Tsinghua University.
%; the 2017 Los Alamos National Laboratory summer research position in the T-2 division with G. Hale and M. Paris; the 2016 and 2015 Oak Ridge National Laboratory summer internships with V. Sobes. 
We would like to thank: Mark Paris and Gerald Hale from Los Alamos National Laboratory for their help on R-matrix theory and the R-matrix 2016 summer workshop in Santa Fe; Gr\'egoire Allaire from \'Ecole Polytechnique for his help on Fredholm's alternative and Perron-Frobenius theory; Semyon Dyatlov from MIT and U.C. Berkeley for his help on Gohberg-Sigal theory; Javier Sesma from Universidad de Zaragoza for his help on properties of the Hankel functions; Yoann Desmouceaux for his help in proving the diagonal divisibility and capped multiplicities lemma 3 of \cite{Ducru_shadow_Brune_Poles_2019}; Andrew Holcomb for his help in testing theorem 1 of \cite{Ducru_shadow_Brune_Poles_2019}; and Haile Owusu for his help on Hamiltonian degeneracy.

\end{acknowledgments}

\appendix

%%%%%%%%%%%%%%%%%%%%%%%%%%%%%%%%%%%%%%%%%%%%%%%%%%%%%%%%%%%%%%%%%%%%%%%%%%%%%%%%
%*******************************************************************************
%%%%%%%%%%%%%%%%%%%%%%%%%%%%%%%%%%%%%%%%%%%%%%%%%%%%%%%%%%%%%%%%%%%%%%%%%%%%%%%%
\section{\label{appendix: Single Breit-Wigner resonance}Single Breit-Wigner capture resonance}
%%%%%%%%%%%%%%%%%%%%%%%%%%%%%%%%%%%%%%%%%%%%%%%%%%%%%%%%%%%%%%%%%%%%%%%%%%%%%%%%
%*******************************************************************************
%%%%%%%%%%%%%%%%%%%%%%%%%%%%%%%%%%%%%%%%%%%%%%%%%%%%%%%%%%%%%%%%%%%%%%%%%%%%%%%%

In order to derive a simple reference case that is tractable analytically, we here study the multipole representation of the first radiative capture s-wave resonance of uranium $^{\text{238}}\text{U}$. 
We neglect the energy dependence of the widths in the resonance (this constitutes the B=S approximation), and denote $\Gamma_\lambda \triangleq \Gamma_\gamma + \Gamma_n$, so that the $\gamma$-channel cross section takes the form:

% Code used in Uncertainty Methods Toy Problem
% def exact_Σγ(E, Γ): ## The most simple SLBW caputre resonance, with BS approximation
%     """
%         Evaluate SLBW capture cross section (with BS approximation)
        
%         Parameters
%         ----------
%         Γ :  Resonance Parameters [Eλ resonance energy, Γn neutron width, Γγ gamma width]
        
%         Returns
%         -------
%         float : capture cross section
%         """
%     return (np.pi*Γ[1]*Γ[2]/(ρ0**2*Γ[0]**0.5*E**0.5))/((Γ[0]-E)**2+(Γ[1]+ Γ[2])**2/4)

% Isaac note: \label{eq:: caputre SLBW} matches exactly code use to evaluate cross section for plots

\begin{equation}
\begin{IEEEeqnarraybox}[][c]{C}
     \sigma_\gamma(E)  =  \pi g_{J^\pi} a_c^2 \frac{\Gamma_\gamma \Gamma_n}{\rho_0^2 \sqrt{E_\lambda}}\frac{1}{\sqrt{E}}\frac{1}{\left(E_\lambda-E\right)^2 +  \Gamma_\lambda^2/4}
\label{eq:: caputre SLBW}
\IEEEstrut\end{IEEEeqnarraybox}
\end{equation}
which is a Single-Level Breit-Wigner resonance (\ref{eq:: SLBW Single Breit Wigner resonance}) with $\mathcal{E_\lambda} \triangleq E_\lambda - \mathrm{i}\frac{\Gamma_\lambda}{2}$, $a=0$, and $b \triangleq 2\pi \frac{\Gamma_\gamma \Gamma_n}{\rho_0^2 \sqrt{E_\lambda}\Gamma_\lambda}$, i.e.
\begin{equation}
\begin{IEEEeqnarraybox}[][c]{C}
     \sigma_\gamma(E)  = \frac{1}{\sqrt{E}}\Re\left[\frac{\mathrm{i}b}{ E - \mathcal{E_\lambda}}\right]
\label{eq:: simplified caputre SLBW Real part notation}
\IEEEstrut\end{IEEEeqnarraybox}
\end{equation}
Let us now cast (\ref{eq:: caputre SLBW}) into the multipole representation (\ref{eq:: sigma(z) Windowed Multipole representation}). We perform this by change of variables $z^2 = E$, and $p^2 = \mathcal{E}_\lambda$, and partial fraction decomposition:
\begin{equation}
\begin{IEEEeqnarraybox}[][c]{rcl}
\frac{1}{\sqrt{E}}\Re\left[ \frac{\mathrm{i}b}{E-\mathcal{E}_\lambda }\right] & \ = \ & \frac{1}{z^2}\Re\left[ \frac{\mathrm{i}b/2}{z-p} + \frac{\mathrm{i}b/2}{z+p} \right]
\IEEEstrut\end{IEEEeqnarraybox}
\label{eq:: Breit Wigner resonance multipole representation}
\end{equation}
So that the multipole cross section in $z$-space is then:
\begin{equation}
\begin{IEEEeqnarraybox}[][c]{C}
     \sigma_\gamma(z)  =\frac{1}{z^2}\Re_{\mathrm{conj}}\left[ \frac{r}{z-p} + \frac{r}{z+p} \right]
\label{eq::caputre SLBW Real multipole representation}
\IEEEstrut\end{IEEEeqnarraybox}
\end{equation}
with
\begin{equation}
\begin{IEEEeqnarraybox}[][c]{rCl}
    r & \triangleq  & \mathrm{i}\pi \frac{\Gamma_\gamma \Gamma_n}{\rho_0^2 \sqrt{E_\lambda}\; \Gamma_\lambda} \\
    p & \triangleq & \sqrt{E_\lambda - \mathrm{i}\frac{\Gamma_\lambda}{2}}
\label{eq::caputre SLBW Real multipoles (poles and residues)}
\IEEEstrut\end{IEEEeqnarraybox}
\end{equation}
One can then verify the results of theorem \ref{theo::WMP Representation} with these explicit formulae.

In theorem \ref{theo::WMP covariance}, we develop a method to compute the Jacobian matrix $\left( \frac{\partial \Pi}{ \partial \Gamma}\right) $, using the sensitivities $\frac{\partial \sigma}{\partial \Gamma} (E)$ of the cross section $\sigma(E)$ to resonance parameters $\big\{\Gamma\big\}$.
These can here be derived by direct differentiation of (\ref{eq:: caputre SLBW}), yielding the relative sensitivities (derivatives):
\begin{equation}
\begin{IEEEeqnarraybox}[][c]{rCl}
   \frac{1}{\sigma_\gamma} \frac{\partial \sigma_\gamma}{\partial E_\lambda} & = &  \frac{-1}{2E_\lambda} + 2(E - E_\lambda) \sigma_\gamma \frac{\rho_0^2 \sqrt{E}\sqrt{E_\lambda}}{\pi \Gamma_n \Gamma_\gamma} \\
   \frac{1}{\sigma_\gamma} \frac{\partial \sigma_\gamma}{\partial \Gamma_n} & = &  \frac{1}{\Gamma_n} - \frac{\Gamma_\lambda}{2} \sigma_\gamma \frac{\rho_0^2 \sqrt{E}\sqrt{E_\lambda}}{\pi \Gamma_n \Gamma_\gamma} \\
    \frac{1}{\sigma_\gamma} \frac{\partial \sigma_\gamma}{\partial \Gamma_\gamma} & = &  \frac{1}{\Gamma_\gamma} - \frac{\Gamma_\lambda}{2} \sigma_\gamma \frac{\rho_0^2 \sqrt{E}\sqrt{E_\lambda}}{\pi \Gamma_n \Gamma_\gamma} 
\label{eq:: sensitivities in BS approximation}
\IEEEstrut\end{IEEEeqnarraybox}
\end{equation}
Alternatively, these same cross section sensitivities $\frac{\partial \sigma}{\partial \Gamma} (E)$ can be computed using (\ref{eq:: simplified caputre SLBW Real part notation}). For real $b\in\mathbb{R}$, the partial derivatives to any real coefficient $\Lambda \in \mathbb{R}$ follow 
\begin{equation*}
\begin{IEEEeqnarraybox}[][c]{rCl}
   \frac{1}{\sigma_\gamma} \frac{\partial \sigma_\gamma}{\partial \Lambda} & = &  \frac{1}{b}\frac{\partial b}{\partial \Lambda} + \frac{\Re\left[ \frac{\mathrm{i}}{( E - \mathcal{E_\lambda})^2}\frac{\partial\mathcal{E}_\lambda}{\partial \Lambda}\right]}{\Re\left[ \frac{\mathrm{i}}{ E - \mathcal{E_\lambda}}\right]}
\label{eq:: sensitivity to generic parameter}
\IEEEstrut\end{IEEEeqnarraybox}
\end{equation*}
Since we have:
\begin{equation}
\begin{IEEEeqnarraybox}[][c]{lCr}
 \frac{\partial\mathcal{E}_\lambda}{\partial E_\lambda} =  1 \quad & , & \quad \frac{1}{b}\frac{\partial b}{\partial E_\lambda}  =  - \frac{1}{2E_\lambda}  \\
        \frac{\partial\mathcal{E}_\lambda}{\partial \Gamma_n}   =   \frac{\partial\mathcal{E}_\lambda}{\partial \Gamma_\gamma}  =  -\frac{\mathrm{i}}{2}  \quad  & , & \quad \frac{1}{b}\frac{\partial b}{\partial \Gamma_n}  =  \frac{1}{\Gamma_n}- \frac{1}{\Gamma_\lambda} \\
    \quad & , & \quad \frac{1}{b}\frac{\partial b}{\partial \Gamma_\gamma} =  \frac{1}{\Gamma_\gamma}- \frac{1}{\Gamma_\lambda}  
\label{eq:: sensitivities to generic parameter}
\IEEEstrut\end{IEEEeqnarraybox}
\end{equation}
the cross section sensitivities $\frac{\partial \sigma}{\partial \Gamma} (E)$ to resonance energy $E_\lambda$, neutron scattering width $\Gamma_n$, and radiative capture width $\Gamma_\gamma$ are thus respectively
\begin{equation}
\begin{IEEEeqnarraybox}[][c]{rCl}
   \frac{1}{\sigma_\gamma} \frac{\partial \sigma_\gamma}{\partial E_\lambda} & = &  \frac{-1}{2E_\lambda} + \frac{\Re\left[ \frac{\mathrm{i}}{( E -\mathcal{E_\lambda})^2}\right]}{\Re\left[ \frac{\mathrm{i}}{E - \mathcal{E_\lambda} }\right]} \\
      \frac{1}{\sigma_\gamma} \frac{\partial \sigma_\gamma}{\partial \Gamma_n} & = & \frac{1}{\Gamma_n} - \frac{1}{\Gamma_\lambda} + \frac{1}{2} \frac{\Re\left[ \frac{1}{ (E -\mathcal{E_\lambda})^2}\right]}{\Re\left[ \frac{\mathrm{i}}{ E -\mathcal{E_\lambda}}\right]} \\
          \frac{1}{\sigma_\gamma} \frac{\partial \sigma_\gamma}{\partial \Gamma_\gamma} & = & \frac{1}{\Gamma_\gamma} - \frac{1}{\Gamma_\lambda} +  \frac{1}{2} \frac{\Re\left[ \frac{1}{ (E -\mathcal{E_\lambda})^2}\right]}{\Re\left[ \frac{\mathrm{i}}{ E -\mathcal{E_\lambda}}\right]}
\label{eq:: sensitivity to resonance energy, scattering, and gamma widths}
\IEEEstrut\end{IEEEeqnarraybox}
\end{equation}
where the derivatives could be taken within the real part because all the parameters were real.
Using the cross section sensitivities $\frac{\partial \sigma}{\partial \Gamma} (E)$ -- either from (\ref{eq:: sensitivities in BS approximation}) or (\ref{eq:: sensitivity to resonance energy, scattering, and gamma widths}) -- and performing the corresponding Hwang's conjugate continuation (section \ref{subsec:Hwang's conjugate continuation}), one can therefore compute the multipole sensitivities $\left( \frac{\partial \Pi}{ \partial \Gamma}\right) $ of theorem \ref{theo::WMP covariance} using the contour integrals system (\ref{eq:: dΠ/dΓ contour integrals system}).

In this simple case of a Single-Level Breit-Wigner resonance in multipole representation (\ref{eq::caputre SLBW Real multipole representation}), we are also able to explicitly calculate the multipole sensitivities to resonance parameters -- i.e. Jacobian  $\left( \frac{\partial \Pi}{ \partial \Gamma}\right) $ -- by direct differentiation of the explicit formulae (\ref{eq::caputre SLBW Real multipoles (poles and residues)}), yielding:
\begin{equation}
\begin{IEEEeqnarraybox}[][c]{rCl}
     \frac{\partial p_+}{\partial E_\lambda} = - \frac{\partial p_-}{\partial E_\lambda} & = & \frac{1}{2 p_+} \\
     \frac{\partial r_+}{\partial E_\lambda} = \frac{\partial r_-}{\partial E_\lambda} & = & - \frac{r_+}{2 E_\lambda} \\
     \text{and} \quad \quad \quad \quad \quad \quad \quad & & \\
     \frac{\partial p_+}{\partial \Gamma_n} = - \frac{\partial p_-}{\partial \Gamma_n} & = & \frac{-\mathrm{i}}{4 p_+} \\
     \frac{\partial r_+}{\partial \Gamma_n} =  \frac{\partial r_-}{\partial \Gamma_n}  & = & r_+ \left( \frac{1}{\Gamma_n} - \frac{1}{\Gamma_\lambda}\right) \\
     \text{and} \quad \quad \quad \quad \quad \quad \quad  & & \\
     \frac{\partial p_+}{\partial \Gamma_\gamma} = - \frac{\partial p_-}{\partial \Gamma_\gamma} & = & \frac{-\mathrm{i}}{4 p_+} \\
     \frac{\partial r_+}{\partial \Gamma_\gamma} =  \frac{\partial r_-}{\partial \Gamma_\gamma}  & = & r_+ \left( \frac{1}{\Gamma_\gamma} - \frac{1}{\Gamma_\lambda}\right) 
\label{eq:: multipole sensitivities SLBW case}
\IEEEstrut\end{IEEEeqnarraybox}
\end{equation}
The latter multipole sensitivities (\ref{eq:: multipole sensitivities SLBW case}) can then be used to validate theorem \ref{theo::WMP covariance}.

For verification and reproducibility purposes, we generated figures \ref{fig:ENDF Uncertainty}, \ref{fig:LARGE ENDF Uncertainty}, and \ref{fig:Doppler boradening of WMP} using cross section (\ref{eq:: caputre SLBW}) with the parameters from the neutron slowdown analytic benchmark \cite{Analytic_Benchmark_1_2020}, which we here report in table \ref{tab:resonance parameters U-238 first resonance}. These parameters are similar (but not identical) to those of ENDF/B-VIII.0 evaluations, yielding the same cross section to a multiplicative constant. The resonance energies and widths are those of ENDF/B-VIII.0, as well as their covariance matrix. The enlarged covariance matrix in table \ref{tab:resonance parameters U-238 first resonance} is that of the analytic benchmark \cite{Analytic_Benchmark_1_2020}, and was designed to bring the neutron slowdown problem past the linear regime in resonance sensitivity. 

\begin{table}[h]
\caption{\label{tab:resonance parameters U-238 first resonance} 
Resonance parameters of the first s-wave radiative $\gamma$-capture resonance of $^{\mathrm{238}}\mathrm{U}$ used for generating temperature tolerance plot (FIG. \ref{fig:Doppler boradening of WMP}) and sensitivities demonstration (FIG.\ref{fig:ENDF Uncertainty} and FIG.\ref{fig:LARGE ENDF Uncertainty}). The resonance energies and widths, as well as their covariance matrix, are those of ENDF/B-VIII.0 evaluation \cite{ENDFBVIII8th2018brown}. The enlarged covariance matrix, as well as the channel radius $a_c$, the atomic weight $A$, and $\rho_0$, are those of the analytic neutron slowdown benchmark \cite{Analytic_Benchmark_1_2020}}
\begin{ruledtabular}
\begin{tabular}{l}
$ z = \sqrt{E}$ with $E$ in ($\mathrm{eV}$) \\
$A = 238$ \\ %ISAAC/PABLO/VLAD: Not exact ENDF/B-VIII.0 
$a_c = 0.000948$ : channel radius, in Fermis \\ %ISAAC/PABLO/VLAD: VERIFY THIS: ANALYTIC BENCHMARK 0.000948 FERMIS, VERSUS 9.8 Fermis in ENDF/B-VIII.0 
$\rho_0 / a_c =  0.002196807122623/2$ ($\sqrt{\mathrm{eV}}^{\mathrm{-1}}$) \\ %ISAAC/PABLO/VLAD: ANALYTIC BENCHMARK HAS \rho_0 =  0.002196807122623/2, ENDF HAS \rho_0 =  0.002196807122623 x A/(A+1)
$E_\lambda = 6.674280$ : first resonance energy ($\mathrm{eV}$) \\
$\Gamma_{n} = 0.00149230$ : neutron width of first resonance \\ 
$\Gamma_{\gamma} = 0.0227110$ : eliminated capture width ($\mathrm{eV}$) \\
$g_{J^\pi} = 1$ : spin statistical factor \\

ENDF/B-VIII.0 covariance matrix:\\
$\begin{IEEEeqnarraybox}[][c]{l}
\mathbb{V}\mathrm{ar}\left(\left[E_0, \Gamma_{n}, \Gamma_{\gamma}\right]\right)   =   \\ 
\left[ \begin{array}{ccc}
 1.1637690 \mkern+2mu\mathrm{E}\mkern-5mu - \mkern-5mu 7 & -2.7442070 \mkern+2mu\mathrm{E}\mkern-5mu - \mkern-5mu 10 & 1.8617500 \mkern+2mu\mathrm{E}\mkern-5mu - \mkern-5mu 8 \\
-2.7442070 \mkern+2mu\mathrm{E}\mkern-5mu - \mkern-5mu 10 & 3.9366000 \mkern+2mu\mathrm{E}\mkern-5mu - \mkern-5mu 10 & -6.5102670 \mkern+2mu\mathrm{E}\mkern-5mu - \mkern-5mu 9 \\
 1.8617500 \mkern+2mu\mathrm{E}\mkern-5mu - \mkern-5mu 8 & -6.5102670 \mkern+2mu\mathrm{E}\mkern-5mu - \mkern-5mu 9 & 1.6255630 \mkern+2mu\mathrm{E}\mkern-5mu - \mkern-5mu 7 
\end{array} \right]
\IEEEstrut\end{IEEEeqnarraybox}\label{eq:ENDFVIII_covariance_matrix}$\\

Enlarged covariance matrix (same correlation):\\
$\begin{IEEEeqnarraybox}[][c]{l}
\mathbb{V}\mathrm{ar}\left(\left[E_0, \Gamma_{n}, \Gamma_{\gamma}\right]\right)   =   \\ 
\left[ \begin{array}{ccc}
 1.2373892     & -1.1217107 \mkern+2mu\mathrm{E}\mkern-5mu - \mkern-5mu 5 & 5.6993358 \mkern+2mu\mathrm{E}\mkern-5mu - \mkern-5mu 4 \\ 
-1.1217107 \mkern+2mu\mathrm{E}\mkern-5mu - \mkern-5mu 5 & 6.1859980 \mkern+2mu\mathrm{E}\mkern-5mu - \mkern-5mu 8 & -7.6617177 \mkern+2mu\mathrm{E}\mkern-5mu - \mkern-5mu 7 \\
 5.6993358  \mkern+2mu\mathrm{E}\mkern-5mu - \mkern-5mu 4 & -7.6617177 \mkern+2mu\mathrm{E}\mkern-5mu - \mkern-5mu 7& 1.4327486 \mkern+2mu\mathrm{E}\mkern-5mu - \mkern-5mu 5
\end{array} \right]
\IEEEstrut\end{IEEEeqnarraybox} \label{eq:ENDFVIII_covariance_matrix_ENLARGED}$
\end{tabular}
\end{ruledtabular}
\end{table}

% ABDULLA: C-TERM CORRECTION.
% EXPLAIN THE POLES WHICH ARE EXACTLY OPPOSITE CANCEL OUT, AND STUDY HERE THE FACT THAT s-WAVE RESONANCES IN E-SPACE YIELDS AUTOMATIC CANCELLATION. (They do because $1/\sqrt{E}$ terms yield same residues for opposite poles.). 

\clearpage

\bibliography{Windowed_multipole_representation_of_R_matrix_cross_sections}% Produces the bibliography via BibTeX.

\end{document}